\documentclass[11pt, a4paper]{article}

\usepackage[utf8]{inputenc}
\usepackage[fleqn]{amsmath}
\usepackage[mathscr]{euscript}
\usepackage{gensymb,amsfonts}
\usepackage{amssymb, }
\usepackage{array}
\usepackage[T1]{fontenc}
\usepackage[margin=1in]{geometry}
\usepackage{lmodern}
\usepackage{tabularx}
\usepackage{fancyhdr}
\usepackage{graphicx}
\graphicspath{{images/}}
\usepackage[utf8]{inputenc}
\usepackage[english]{babel}
\usepackage{amsthm}
\usepackage{braket}
\usepackage{tikz}
\usepackage{subfig, float}
\usepackage[utf8]{inputenc}
\usepackage[backend=biber, bibstyle=numeric, url=false, doi=false, sortcites]{biblatex}
\usepackage{authblk,csquotes}
\usetikzlibrary{cd}

\newtheorem{theorem}{Theorem}
\newtheorem{proposition}{Proposition}
\newtheorem{lemma}{Lemma}
\newtheorem{definition}{Definition}
\newtheorem{corollary}{Corollary}

\newcommand{\hlf}{\frac{1}{2}}

\newcommand{\tr}{\text{Tr}}

\newcommand{\supp}{\text{Supp}}

\newcommand{\mc}{\mathcal}

\newcommand{\mb}{\mathbb}
\newcommand{\id}{\text{id}}

\newcommand{\im}{\mathrm{Im}\,}
\newcommand{\hyph}{\text{-}}

\title{Distilling Fractons from Layered Subsystem-Symmetry Protected Phases}
\author{Albert T. Schmitz\footnote{contact: albert.schmitz.colorado.edu}} 
\affil{Department of Physics and Center for Theory of Quantum Matter\\University of Colorado, Boulder, CO 80309, USA}

\begin{document}
\maketitle

\abstract{
It is well-known that 3D Type-I fracton models can be obtained from the condensations of stacked layers of 2D anyons. It is less obvious if 3D Type-II fractons can be understood from a similar perspective. In this paper, we affirm that this is the case: we produce the paradigm Type-II fracton model, Haah's cubic code, from a 2D layer construction. However, this is not a condensation of 2D anyons, but rather we start with stacks of 2D {\it subsystem-symmetry protected topological states} (SSPT). As this parent model is not topologically ordered in the strict sense, whereas the final state is, we refer to this process as a {\it distillation} as we are forming a long-range entangled (LRE) state from several copies of a short-range entangled (SRE) state. We also show that Type-I fracton topological order can also be distilled from SSPT states in the form of the {\it cluster-cube model} which we introduce here. We start by introducing the Brillouin-Wigner perturbation theory of distillation. However, a more detailed analysis from which the perturbation theory follows, is also performed using linear gauge structures and an extension we introduce here referred to as {\it gauge substructures}. This allows us to rigorously define distillation as well as understand the process of obtaining LRE from SRE. We can diagnose the source of LRE as the distillation of subsystem symmetries into robust long-range ground state degeneracy as characterized by logical operators of the resulting stabilizer code. Furthermore, we find which Hamiltonian terms are necessary for selecting the ground state which results from the perturbation. This leads to a protocol for realizing a fracton quantum error correcting code initialized in a chosen fiducial state using only finite-depth circuits and local measurements.}

\tableofcontents

\section{Introduction}

Topologically ordered (TO) phases have garnered a lot of attention due the fact that they do not conform to the Landau symmetry-breaking paradigm. However, topological phases have also generated their own paradigmatic structure, i.e. finite braiding and fusion relations, effective field theory descriptions in terms of topological quantum field theories (TQFT), and fully deformable line-like logical operators, i.e. the Wilson loop operators which characterize the topological sectors of the phase. In the same way that TO brake the Landau paradigm, one can argue that fracton topological order (FTO) breaks some of these TO paradigms as the braiding and fusion relations are infinite \cite{Pai2019}, there is no TQFT effective description, and logical operators can take on exotic forms. For these reasons as well as others, fractons as introduced in \cite{Chamon2005,Nussinov2009,Nussinov2009a} and made prominent in \cite{Bravyi2011,Haah2011,Yoshida2013,Vijay2015,Vijay2016, Pretko2017a, Pretko2017b} have gained a lot of attention in recent years; see \cite{Nandkishore2018} for a review. All these properties, as well as the definitive property of the restrictive mobility of quasi-particles, come from a subsystem structure of the model in a similar way to how TO gains its properties from global structure \cite{Vijay2015, Vijay2016, Ma2017,Slagle2018b,Shirley2018, Shirley2019}. For Type-I FTO, these subsystems correspond to 2D planes, whereas anything else corresponds to Type-II FTO\footnote{We take the perspective of Refs. \cite{ Haah2011, HaahThesis} that Type-I is characterized by the existence of string logical operators or equivalently, some mobile composite excitations,  and Type-II is anything with no string logical operator, or equivalently, no mobile composite excitations.}, most notably fractal subsystems. With the introduction of layer constructions \cite{Ma2017, Vijay2017} and foliation \cite{Slagle2018b,Shirley2018, Shirley2019} it would seem well-established that Type-I FTO is given by stacking layers of 2D TO and coupling the layers, whereby  the subsystem structure is inherited from the layers. It is still an open question and seems to be assumed that Type-II cannot be obtained from such 2D layer constructions, as one would reasonably suspect such a model to have planar subsystem structure. In this letter, we provide a counter-example to this intuition by showing a 2D layer construction which realizes the paradigm Type-II model, Haah's cubic code \cite{Haah2011, HaahThesis}. However unlike the typical layer construction, we start with many copies of a subsystem-symmetry protected topological (SSPT) state \cite{Yizhi2018,Devakul2019,Devakul2018, Kubica2018, Stephen2018, Williamson2018, Devakul2019a} as opposed to some 2D TO phase. We term this SSPT model {\it the quasi-cluster model} which is new to the best of our knowledge. That is we start with stacks of a short-range entangled (SRE) SSPT phase and couple them to form a long-range entangled (LRE) FTO phase in a process we term {\it distillation} of LRE from SRE (SRE$\to$LRE).\footnote{We are borrowing this term from the process of entanglement distillation\cite{Bennett1996}, whereby lots of low-grade noisy entanglement (analogous to SRE for us) is distilled into fewer high-grade bell-pairs (analogous to LRE for us).}  One might then suspect that Type-II is a by-product of this SRE$\to$LRE distillation, but again, we present a counter-example to this as well by finding a Type-I model, what we term the {\it cluster-cube model} via a similar SRE$\to$LRE distillation process using the cluster state. The connection between SSPT and FTO in not new. For example, Ref. \cite{Schmitz2019b} found that when action on the cluster state SSPT is restricted to  respect the subsystem symmetry, the resulting excitations become effective 2D fractons. Other examples of this connection can be found in Refs. \cite{Vijay2016, Devakul2018, Shirley2018, Williamson2019}.

The SRE$\to$LRE construction in this paper closely follows that of the layer construction in Refs. \cite{Ma2017, Vijay2017}, where the paradigm Type-I X-cube model is formed from layers of the $d=2$ toric code. There, Brillouin-Wigner perturbation theory is used to show how an infinite strength perturbation to the stacked layers of the $d=2$ toric code realizes the X-cube model, a result we review below. However, we take the analysis of their result and ours further by relying on the fact that all three constitutive models: the parent layers, the target fracton model, and the perturbation, all form {\it stabilizer codes}. Stabilizer codes were first introduced in \cite{Gottesman1997} in the context of quantum error correcting codes (QECC), but they also serve as exactly solvable models for topological phases. In a paper by the author \cite{Schmitz2019a}, it is shown that all stabilizer codes can be abstracted to the notion of a {\it linear gauge structure}. In this paper, we extend the theory of gauge structures to include what we refer to as {\it gauge substructures}, whereby one obtains one gauge structure from another by essentially restricting the set of operators which act on the system while simultaneously restricting the set of commuting operators which define the stabilizer code. As a result, all the layer constructions we present can be couched in this gauge substructure language. This abstraction comes with several benefits: 1) the process is much more constructive, whereby one determines the parent and target models, and the perturbation can be extracted from the analysis. 2) One can rigorously define and thus identify what makes SRE$\to$LRE distillation qualitatively distinct from anyon condensation. 3) Finally and most importantly, we can rigorously answer questions about SRE$\to$LRE that are more obscure in the perturbation analysis and at best only intuited from the results. For example, from the results of the fracton distillation, we find that subsystem symmetries--in layer planes for Type-I and non-layer planes for Type-II--are responsible for the LRE as they become or are ``distilled'' into logical operators which, at least in part, define the topological sectors. One might recognize and infer this from the perturbation analysis, whereas the gauge substructure analysis proves this is generic: subsystem symmetries which obey certain rules are always distilled into logical operators. Similarly, we also note that as the parent phase is not topologically ordered, i.e. there is a unique ground state, the perturbation could take this state to anyone of the many ground states of the target. So one asks, which ground state do we land on and why? The perturbation analysis may answer which state we land on, but the gauge substructure answers why: the perturbation is not unique, and the  perturbation ``selects'' some set of distilled subsystem symmetries over others. The state we land on is the $+1$ eigenstate of those distilled logical operators. This result also suggests an interesting possibility: we can use SRE$\to$LRE fracton distillation to form a fracton QECC initialized in a fiducial state of our choosing using only finite depth circuits and local measurements or with some generalization, possibly realize a fault-tolerant measurement-based quantum computer by appropriately choosing the perturbation. We discuss this possibility near the end of the paper.

The structure of the remaining paper is as follows: The main body can be divided roughly into two parts. The first deals with the layer constructions from the perspective of perturbation theory (pages \pageref{sec:stab}-\pageref{sec:part2}), and the second part takes the perspective of gauge structures and substructures (pages \pageref{sec:part2}-\pageref{sec:con}). Both are valuable for understanding the results, but those readers less interested in technical details and abstraction may wish to focus more on the first part. We start by reviewing the stabilizer formalism in Section \ref{sec:stab}. From there we discuss the general perturbation analysis used for our examples in Section \ref{sec:PT}, and then apply these ideas to the examples in Section \ref{sec:pert}, where all models are explored in detail.  In Section \ref{sec:genprop}, we extract some of the general properties found for SRE$\to$LRE distillation, thus ending the first part. We then start the gauge structure analysis by reviewing the original definition in Section \ref{sec:lgsdef} and defining a gauge substructure in Section \ref{sec:gs}. Section \ref{sec:dsum} defines the direct sum of gauge structures as well as the distinction between condensation and distillation in the language of gauge substructures as a direct sum and also gives a useful necessary of sufficient condition for condensation. We then deep dive into SRE$\to$LRE distillation in Section \ref{sec:sretolre} where we prove that most logical operators are distilled from subsystem symmetries and how the exact form of the perturbation determines the resulting ground state of the fracton model. We then apply these results to our examples in Section \ref{sec:subex}. Section \ref{sec:ERGandMBQC} briefly discusses how gauge substructures are similar to entanglement renormalization and how SRE$\to$LRE can be used to form a fracton QECC as well as conjectured to be a means of realizing a fault-tolerant measurement-based quantum computer. We finish with concluding remarks in Section \ref{sec:con}.

\subsection{Notation}

We use small Roman letters for operators, capital Roman letters for finite simple sets and generic elements of sets (when there is no confusion) and script Roman letters for linear or algebraic spaces. All maps use Greek letters. The composition of maps $\alpha$ and $\beta$ is written as $\alpha \beta$, $\star$ represents the pullback, $\iota_{\mc A}$ the inclusion map for the space $\mc A$ and $\id_{\mc A}$ the identity map on $\mc A$. $\dim \mc A$ is the dimension of $\mc A$. For the sake of readability, we use square-brackets for the image and pre-image i.e, $\alpha[A]$ and $\alpha^{-1}[A]$, respectively. For spaces defined with some two-form $\omega :\mc A \times \mc A \to \mb F$, where $\mb F$ is the base field of $\mc A$, we define for all subspaces $\mc B \subseteq \mc A$ the perpendicular complement with respects to $\omega$, $\mc B^{\perp_\omega}$ as the set of all $A\in \mc A$ such that $\omega(A, B)=0$ for all $B \in \mc B$. We also use $+$ for the symmetric difference\footnote{Recall the symmetric difference is give by $A + B = A \cup B - A\cap B$.}, where there should be no confusion when used for sets, $\wp(A)$ for the power set of $A$ or the set of all subsets of $A$ and $|A|$ for the number of elements in the finite set $A$. For all figures, we use magenta to represent X-type support, cyan for Z-type support and color code plaquettes in different directions via: red $\rightarrow \hat{e}_1$ direction, green $\rightarrow\hat{e}_2$ direction and blue $ \rightarrow\hat{e}_3$ direction.

\section{The Basic Setup}

\subsection{Stabilizer Codes}\label{sec:stab}

All models in this paper can be described in the language of stabilizers codes which implies they all fit within the framework of a linear gauge structure as introduced in \cite{Schmitz2019a} and reviewed here in Section \ref{sec:lgsdef}. Most of the information in this section is extracted from Refs.\cite{Gottesman1997, Haah2011, Haah2013, Terhal2015, Schmitz2019a}. Those familiar with stabilizer codes may wish to skip this section. 

Consider a system of $N$ spin-$\hlf$ or qubit degrees of freedom, as contained in the set $Q$, such that the Hilbert space is $\mc H\simeq \mb C_2^{\otimes N}$. Let $\mc P$ be the set of all products of single-qubit X-, Y- and Z-type Pauli operators (simply referred to as Pauli operators) acting on the qubits of our system, modulo any phase of $\pm1, \pm i$. We use $x_i,y_i$ and $z_i$ for the single-qubit X-, Y-, and Z-type Pauli operators for qubit $i$. Ignoring the phase allows us to treat $\mc P$ as a vector space over the field of two elements, $\mb F_2$, where addition of two Pauli's is  given by their product and scalar multiplication corresponds to the power. Since $f^2 \propto \id_{\mc H}$, $\mb F_2$ is the appropriate field. For any member $f \in \mc P$ we define the {\it support} as
\begin{align}
\supp(f) =\{i \in Q: f_i \not \propto \id_{\mb C_2} \},
\end{align}
where $f_i$ is the part of $f$ acting in the product space associated with $i\in Q$. Likewise for any $F\subseteq \mc P$, $\supp(F) = \bigcup_{f \in F} \supp(f)$. As we have modded out the phase, it seems we have thrown away the commutation relations. We can recover this information by introducing the symplectic form $\lambda: \mc P\times \mc P \to \mb F_2$, which encodes the commutation relations via
\begin{align}\label{eq:lamdef1}
(f,g) \mapsto \lambda(f,g)= \begin{cases}
0 & \text{ if $f$ and $g$ commute,}\\
1 & \text{otherwise.}
\end{cases}
\end{align}
As any two Pauli's either commute or anti-commute, this encodes all commutation relations. It is symplectic since it is a bi-linear non-degenerate form which satisfies $\lambda(p,q)= -\lambda(q,p)$. We have a natural basis for $\mc P$, namely the set of all single qubit $X$- and $Z$-type Pauli's (as $y_i \propto x_i z_i$). As this is a basis, we conclude $\dim \mc P = 2 N$. Note this basis also has a special property that it divides into two subsets $\{x_i : i \in Q\}$ and $\{z_i: i \in Q\}$ such that $\lambda(x_i, x_j) = \lambda(z_i,z_j)=0$ and $\lambda(x_i, z_j) = \delta_{ij}$. This form of basis is general for a symplectic vector space. That is, given any basis, one can always form a {\it canonical basis},  $\{f_i\}_{i\in \mb Z_N}\cup \{g_i\}_{i\in \mb Z_N}$ such that $\lambda(f_i, f_j) = \lambda(g_i,g_j)=0$ and $\lambda(f_i,g_j) = \delta_{ij}$. As a corollary, the maximum number of independent, mutually commuting operators is $N$. 

We now define the stabilizer set and its associated group.

\begin{definition} \label{def:stab}
 The set $ S \subset \mc P$ is a stabilizer set if at minimum it satisfies the following:
\begin{enumerate}
\item $ \lambda[S \times S]=\{0\}$, i.e. it is composed of mutually commuting operators and
\item $ id_\mc H \notin  S$.
\end{enumerate}

Though these conditions are technically sufficient, it is often best to include a few additional restrictions:
\begin{enumerate}
\setcounter{enumi}{2}
\item $|S| \geq N$, i.e. there are at least as many stabilizers as there are qubits in the system,
 \item $\supp( S) =Q$, i.e. every qubit is acted upon nontrivially by at least one member of $S$ and
\item Any symmetry of the {\it stabilizer group} is also a symmetry of $S$.
\end{enumerate}
\end{definition}

The last requirement involves the stabilizer group denoted by $\mc G$ and defined as
\begin{align}
\mc G= \{ \prod_{s \in F} s \mod(\pm 1, \pm i) : F\in \wp (S)\}.
\end{align}
Thus $\mc G \subset \mc P$ is the set of all Pauli operators generated by taking products of members from $S$, modulo any phase. We refer to members of $S$ as stabilizers and members of $\mc G$ as stabilizer group elements.\footnote{This is not the usual terminology as all members of $\mc G$ are typically called stabilizers. But for the purposes of this discussion, it is better to distinguish only members of $S$ as the stabilizers.} To further explain requirement 5 of Definition \ref{def:stab}, we let the set of symmetries of $\mc G$ be denoted by $\Pi$ which contains members $\pi : \mc P \to \mc P$ which are bijective maps that preserves all of $\mc G$. i.e. $\pi[\mc G] = \mc G$. So our last requirement means that if $\pi \in \Pi$, then $\pi[S] = S$. 

The stabilizer code is then described as the subspace
\begin{align}
\mc H_{\text{code}} =\{ \ket{\psi} \in \mc H : s \ket{\psi} = \ket{\psi}, \text{ for all } s\in S\},
\end{align}

or all states which are simultaneous $+1$ eigenstates of  members of $S$. The stabilizer code is used to encode logical information in a fault-tolerant way \cite{Terhal2015}. $\mc H_{\text{code}}$ can also be described as the ground space of the Hamiltonian
\begin{align}\label{eq:stabham}
H_{S} = \hlf \sum_{s \in S}j_s (1-s),
\end{align}
where $j_s>0$ and the phase of $s \in S$ is chosen such that it is Hermitian. Note that for all $s$, $\hlf (1-s)$ is the projection onto the $+1$ eigenstates of $s$. We choose this form so that the ground space eigenvalue is zero. 

 All examples presented here are Calderbank, Shor and Steane (CSS) stabilizer codes \cite{Terhal2015}. That is, stabilizers are composed of either all $X$-type or $Z$-type operators. Even though our examples are of the CSS form, it is not an essential ingredient for the general results of this paper.  

As with the Pauli space, we introduce another $\mb F_2$ vector space $(\wp (S), +)$ as generated by the stabilizer set where we shall use $\mc A =\wp(S)$. This is also equipped with a two-form $\omega: \mc A \times \mc A \to \mb F_2$ such that
\begin{align}
(A,B) \mapsto \omega(A,B) =|A \cap B| \mod 2.
\end{align}
Unlike $\lambda$, this is not a symplectic form, nor is it an inner product as $\omega(A,A)=0$ does not imply $A=0$. However, it does have the required properties such as bi-linearity and non-degeneracy. If we were to map the space $\wp (S)$ onto $\{0,1\}^{|S|}$, we recognizes $\omega$ as the binary ``dot product,'' so this definition is natural. $\omega$ can be used to construct vectors using the special singleton basis $\{\{s\}: s\in S\}$ via 
\begin{align} \label{eq:recon}
A= \sum_s \omega(A, \{s\}) \{s\}.
\end{align}
It is not immediately obvious that this space is important for the perturbation analysis, but it as well as $\omega$ are used to form the linear gauge structure of a stabilizer code and is  consequential for the excitations of the stabilizer Hamiltonian as argued below.
 
To relate members of $\mc A$ back to operators, we define the {\it stabilizer map} $\phi:\mc A \to \mc P$ such that
\begin{align}
A \mapsto \phi(A)= \prod_{s\in A}s. 
\end{align}
Note that $\im \phi = \mc G$. It should be clear that $\phi$ is linear, however, it is neither surjective nor injective. We call the space $\ker \phi$  the {\it constraint space} and its members constraints for which $C \in \ker\phi$ implies $\prod_{s\in C} s \propto \id_{\mc H}$.  In all cases considered here, some of these constraints have special meaning with respects to the topology of the system. This distinction is made precise in \cite{Schmitz2019a}, but can be roughly understood as those constraints which are lost when the topology of the underlying system is changed.\footnote{When we say changing the topology of the system, we really mean changing the topology of the Hamiltonian/stabilizers, not the Hilbert space. This means that the number of qubits is unchanged, but some sub-extensive number of stabilizers have been altered so as to change the topology. A precise definition is given in the Reference.} We refer to a system with trivial topology as having {\it open boundary conditions} (obc) while a non-trivial topology--typically a $d$-torus in $d$ dimensions-- is referred to as having {\it periodic boundary conditions} (pbc). 

 Excitations for the Hamiltonian in Eq.~\eqref{eq:stabham} can be characterized by stabilizer eigenstates such that all member of a subset of $S$ have eigenvalue $-1$. We can characterize such excitations by the map $\psi: \mc P \to \mc A$ such that 
\begin{align}
f \mapsto \psi(f)= \{s \in S: \lambda(s,f)=1\}.
\end{align}
 We refer to this as the {\it syndrome map} and members of its image syndromes due to the connection with error syndromes in the context of error correction. Individual members of a syndrome are called fundamental excitations and subsets composite excitations. This map tells us which stabilizers ``flip'' (change eigenstate) under the action of $f$ on any ground state.  It also has a fundamental connection to $\phi$ and $\omega$ through the {\it braiding relation}
\begin{align}\label{eq:stabgauge}
\lambda(\phi(A), f) = \omega(A, \psi(f)).
\end{align}
for any $A \in \mc A$ and $ f \in \mc P$. 
%
%
This condition, along with the constraint space, is proven to imply a kind of $\mb Z_2$ charge conservation via the following statement\cite{Schmitz2019a}:
\begin{theorem}\label{prop:stabcon}
$J$ is a syndrome i.e. $J \in \im\psi$ if and only if $ \omega(C, J)= 0$ for all constraints $C \in \ker \phi$. 
\end{theorem}
which is to say any configuration of excitations is realizable if and only if that configuration overlaps with every constraint an even number of times. Importantly, this is used to argue that fractonic  behavior in stabilizer codes is a direct consequence of a sub-extensive number of intersecting topological constraints. Examples of this are described in the models discussed below.
We can use this formalism to write all projection operators onto the eigenspaces of the stabilizer Hamiltonian in Eq.~\eqref{eq:stabham} using $J \in \im \psi$ as a ``quantum number'' such that 
\begin{align} \label{eq:proj}
p^{(J)}= \frac{1}{|\mc A|} \sum_{A \in \mc A} (-1)^{\omega(A, J)} \phi(A),
\end{align}
with energy 
\begin{align}
E^{(J)}=  \sum_{s \in S} j_s \omega(\{s\},J).
\end{align}
Eq.~\eqref{eq:proj} is a consequence of group representation theory for Abelian groups\cite{Tinkham2003}.\footnote{The astute reader may recognize that $\im \phi = \mc G$ is the true abelian group, in which case we are over counting the number of unique operators of $\mc G$ in the sum. However, every unique operator has exactly $|\ker \phi|$ number of redundant terms, which is normalized by the over counting of $|\mc A|$ over $|\mc G|$ since $\dim \mc A = \dim \mc G + \dim \ker \phi$. Theorem~\protect \ref{prop:stabcon} also implies that the phase factor for all such non-unique terms is also the same. We write the projectors in this way for convenience of labeling the energy of the eigen spaces.}

As a final point, we describe the existence of {\it logical operators}. That is when constraints are present, it is often the case that there are too few independent stabilizers to form exactly one half of a canonical basis. Thus, there are more operators in $\mc P$ which commute with all of $\mc G$. With our maps $\psi$ and $\phi$, we capture all such operators in the {\it logical subspace} which we define as
\begin{align}
\mc P_\ell(S) = \ker \psi / \mc G = \ker \psi / \im \phi.
\end{align}
A member of $\mc P_\ell (S)$ is really an equivalence class of operators, each of which is referred to as a logical operator as they collectively form a Pauli sub-algebra for information stored in $\mc H_{\text{code}}$ \cite{Terhal2015}. This immediately implies $\log_2 \dim \mc H_{\text{code}} =\hlf \dim \mc P_\ell(S) = d_\ell$ as only that many independent operators in $\mc P_\ell(S)$ can mutually commute and thus have a fixed eigenstate in $\mc H_{\text{code}}$. Thus the ground state degeneracy (GSD) of the stabilizer Hamiltonian is $2^{d_\ell}$. This space also defines the {\it code distance} as $R = | \min_\supp \mc \mc P_{\ell}(S)|$, i.e. the size of the support of the smallest logical operator. If no logical operator exists, then $R=0$. 
 From the logical subspace, we define LRE as:
\begin{definition}
A stabilizer model (sequence of stabilizer codes as characterized by parameter $L$) is called \emph{LRE} if and only if the code distance scales with $L$, i.e. $R \sim L^a$, for any $a>0$. Any other stabilizer model which is not LRE is called \emph{SRE}. 
\end{definition}
So throughout this paper, we consider the existence of logical operators as our definition of LRE so long as the range as given by $R$ scales with the system size. This is appropriate as the existence of such logical operators implies that there does not exists a local, finite depth quantum circuit which take a single-qubit product state to a ground state of the model, as is the usual definition of LRE. 

 Just as with excitations, logical operators can be related to constraints, or more specifically to topological constraints. It is shown in Ref.~\cite{Schmitz2019a} that logical operators can be formed by the product of all operators of a topological constraint after going from pbc to obc. Thus, logical operators are formed in the intersection of any boundary with a topological constraint. 

\begin{figure}

\centering

\begin{tabular}{c c}

\subfloat[Smallest p-string.\label{fig:smlp}]{\includegraphics[scale=.65]{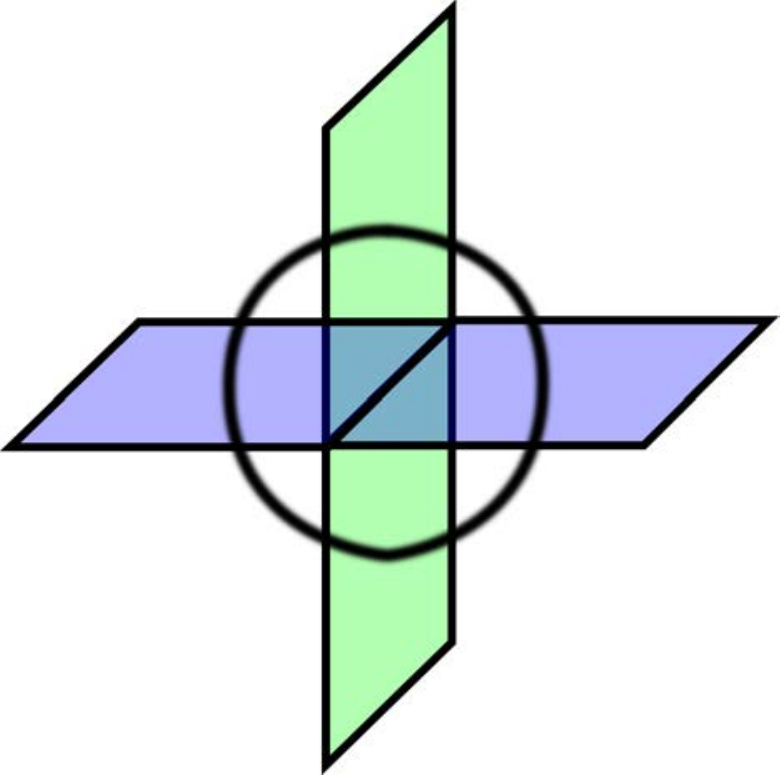}}&

\subfloat[Larger p-string.]{\includegraphics[scale=.35]{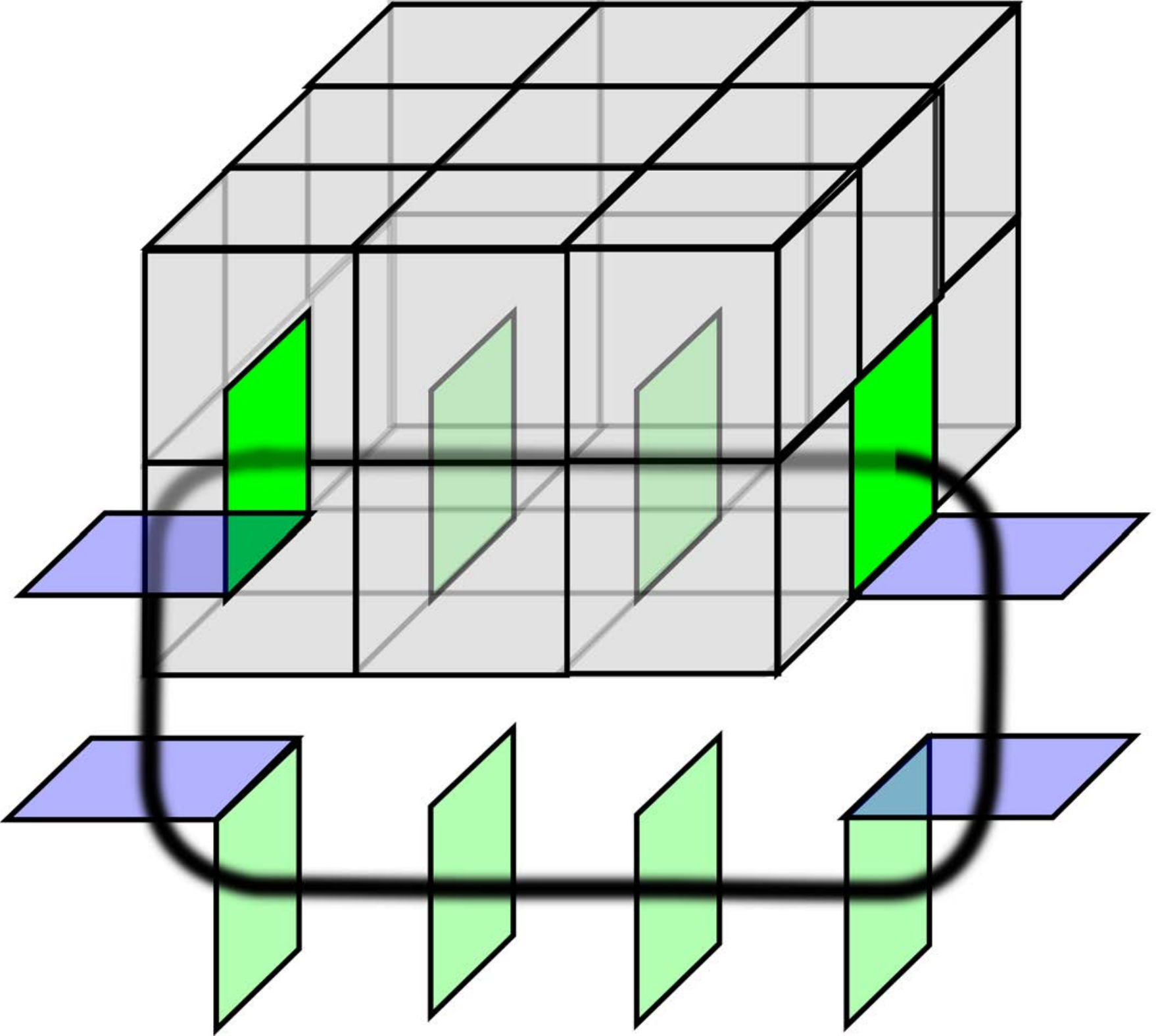}}

\end{tabular}

\caption{An example of how a p-string enters and exits the volume in grey or overlaps with the volume boundary an even number of times.}\label{fig:volume}

\end{figure}

\subsection{Perturbation Theory of Condensation and Distillation}\label{sec:PT}

We start by addressing condensation and distillation from the perspective of perturbation theory by looking at our examples. We follow the procedure used in Ref.~\cite{Ma2017} which uses  Brillouin-Wigner degenerate perturbation theory as reviewed in Appendix \ref{apx:PT}. We start by reviewing the main results of \cite{Ma2017}, namely that one can realize the X-cube model from three coupled stacks, one for each direction, of $d=2$ toric code by condensing so-called ``p-strings.'' Intuitively, a p-string is a collection of toric code plaquette excitations such that we can imagine a continuous string passing through each excited plaquette to eventually form a closed loop. More formally, a p-string is any member of a subspace $\left(\mc D^Z\right)^{\perp_\omega} \subseteq \wp\left(S^Z_{TC_2}\right)^{3L}= \left(\mc A_{TC_2}^Z\right)^{3L}$, i.e. a subspace of all sets of Z-type stabilizers of the toric code stacks. In particular, the space of p-strings is the orthogonal complement with respects to $\omega$ of the subspace $\mc D^Z \subseteq \left(\mc A^Z_{TC_2}\right)^{3L}$. $\mc D^Z$ is generated by all sets of stabilizers which from an elementary cube or, equivalently, is given by the  set of all sets which form contractable, closed 2D membranes.\footnote{ $\mc D$ is used here as we can view this as a discrete, $\mb Z_2$ analog to the space of all vector-valued functions which can be formed via the divergence of a scalar function. Such membrane objects are also referred to by some as ``1-form'' symmetries, though we avoid this terminology as it can cause confusion.} That is, a p-string is all collections which overlaps an even number of times with every elementary cube or, equivalently, a configuration which exits every volume of space it enters, as demonstrated in Fig. \ref{fig:volume}. If we abstract this notion of p-strings to the SSPT stacks, this definition of p-string is important to every model considered here, including the distillation cases.

We focus on three instances of a stacking procedure which result in three distinct models. The parent model consists of  three stacks of $L$ two-dimensional models--one for each direction in three-dimensions. The perturbation or {\it base model} Hamiltonian which couples these layers is then added to the parent model Hamiltonian with a control parameter $K$ making the system Hamiltonian
\begin{align}\label{eq:Ham1}
H= \hlf \sum_{s_p \in S_{\text{parent}}} j_{s_p}(1- s_p) + \frac{K}{2} \sum_{s_b \in S_{\text{base}}} j_{s_b}(1- s_b), 
\end{align}
where $S_{\text{base}}$ is the base model stabilizer set and $S_{\text{parent}}$ is the parent model stabilizer set. We then take $K\to \infty$ and our goal is to realize the effective Hamiltonian whose {\it entire } Hilbert space is that of the extensively degenerate ground space of the base model and whose action is given by
\begin{align}\label{eq:Ham2}
H_{\text{eff}}= \hlf\sum_{s_t \in S_{\text{target}}} j_{s_t}(1- s_t) + \mc O\left(K^{-(\alpha+1)}\right),
\end{align}
up to an overall constant, where $S_{\text{target}}$ is the target model stabilizer set and $\alpha$ is the highest order over all $j_{s_t}$ in $K^{-1}$ . In all cases we consider pbc and generally ignore the relation between the $j$-coefficients of Eqs.\eqref{eq:Ham1} and \eqref{eq:Ham2} though one can derive them in principle as discussed in Appendix \ref{apx:PT}. The parent model is topologically ordered in the case of condensation and an SSPT in the cases of distillation. The target in our three examples is a FTO model, though we do discuss a variation on the condensation case which results in a TO model. For condensation, the FTO is Type-I whereas for distillation, one is Type-I and the other is Type-II. To the best of our knowledge, there is no known condensation process which results in a Type-II model. The base model varies and is not generally unique as discussed in Section \ref{sec:genprop} and is a consequence of the parent and target models. Because the base model is not unique, its properties and phase vary and is generally not important to the final result. As a consequence of using only stabilizer models, a product of parent terms survive the projection onto the ground space of the base model and represents a term in the effective Hamiltonian if and only if that product commutes with all stabilizers of the base model as argued in Appendix \ref{apx:PT}. We then keep the lowest order terms such that only a sub-extensive degeneracy remains for the target model.\footnote{ By extensive and sub-extensive, we mean $\log_2$ of the ground space dimension  grows with the number of qubits either proportionally for extensive or to a power less than one for sub-extensive.}

\section{Models and Analysis of their Perturbation Theory}\label{sec:pert}

\subsection{Condensation Example for Type I Fractons: X-cube} \label{sec:pert1}

Our first example is a Type I fracton model  which can be described as a condensation of anyonic p-strings as argued in Ref. \cite{Ma2017, Vijay2017}. The parent model is layers of the $d=2$ toric code and the target mode is the X-cube model. We review these results as they serve as a contrasting example to distillation.  

\subsubsection{Parent Model}

\begin{figure}

\centering

\includegraphics[scale=.5]{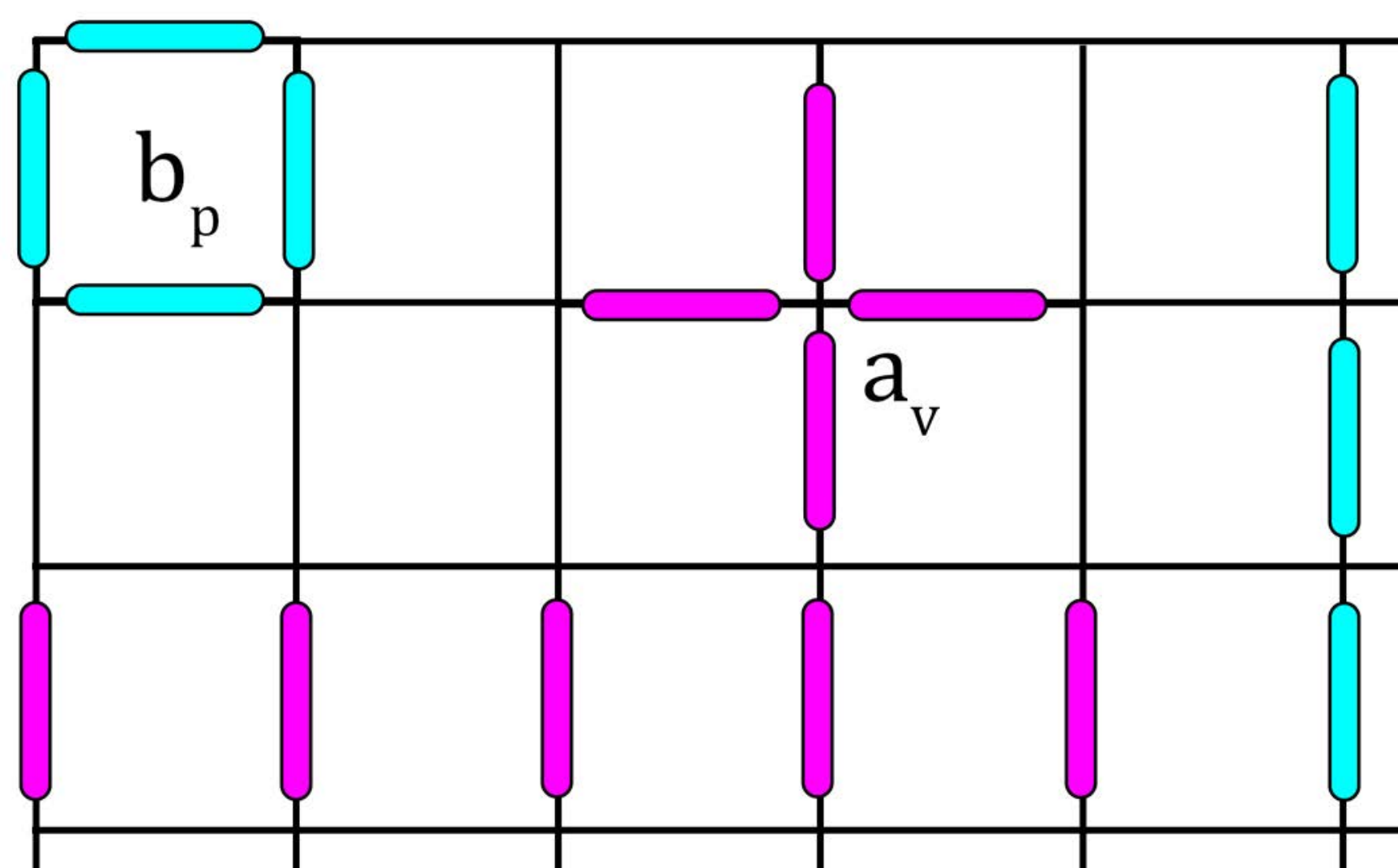}

\caption{Depiction of the operators defining the $d=2$ toric code.}\label{fig:tcdef}

\end{figure}

The parent two-dimensional model is the $d=2$ toric code as introduced in Ref. \cite{Kitaev2003}. The Hilbert space is that of $N=2L^2$ qubits  where each edge of an $L\times L$ square lattice is associated with a single qubit. The stabilizer set is given by 
\begin{align}
S_{TC_2} = \{a_v, b_p: v \text{ vertices and } p \text{ plaquettes}\}.
\end{align}
$b_p= \prod_{i \in p} z_i$, where $i$ indexes the edge qubits about the plaquette $p$ and $a_v = \prod_{i @ v} x_i$, where $i$ indexes the edge qubits coordinated to the vertex $v$. The toric code hosts two independent topological constraints given by the collection of all stabilizers of a given type. This enforces a global (in the plane) $\mb Z_2$ charge conservation via Theorem~ \ref{prop:stabcon}, whereby the $a_v$ terms or ``electic sector'' must be excited in pairs. Likewise is true for the $b_p$ terms or ``magnetic sector.'' $\mc P_\ell(TC_2)$   is given by basis members which are characterized as string operators--two for each no-trivial cycle of the 2-torus--given by forming a pair of electric charges or magnetic charges and wrapping them around the system such that the excitations cancel. All operators are shown in Fig. \ref{fig:tcdef} where we find that $\dim P_\ell(TC_2) =4$,  $R_{TC_2}=L$. Thus the $d=2$ toric code contains LRE.

\subsubsection{Target Model}

The target model is referred to as the X-cube model and is the paradigmatic example of a Type-I fracton phase \cite{Vijay2016}. The Hilbert space is that of $N= 3L^3$ qubits such that each edge of the cubic lattice is associated with one qubit. The stabilizer set is given by
\begin{align}
S_{\text{XC}} = \{a_v^{\hat e_1} , a_v^{\hat e_2}, a_v^{\hat e_3}, b_c : v \text{ vertices and } c \text{ primitive cubes}\}.
\label{eq:HXC}
\end{align}
$a_v^{\hat e_j} = \prod_{ i @ v_j} x_i$, where $i$ indexes the qubits coordinated to vertex $v$ and for edges confined to the plane normal to $\hat e_j$ and $b_c = \prod_{i \in c} z_i$, where $i$ indexes the qubits about the primitive cube $c$ (see Fig. \ref{fig:xcdef}). The properties of the X-cube model are well-known. Elementary excitations of the cube stabilizers are immobile as they cannot hop without generating additional excitations. However, pairs of cubic excitations are free to move in a plane perpendicular to the stacking direction of the two cubes (see Ref. \cite{Prem2017c} for a thorough discussion of excitation hopping in the X-cube model). This restriction on the mobility is a consequence of the sub-extensive number of constraints among these stabilizers. Such constraints are generated by any set which contains all cubes in a plane perpendicular to a coordinate direction. As a single cube lies at  the intersection of three such planar constraints, any syndrome containing that cube must contain at minimum three other excitations, the four of which form a $\mb Z_2$ quadrupole in order to satisfy Theorem \ref{prop:stabcon}. As for the vertex stabilizers, they must form a composite excitation of two terms at a given vertex due to the extensive number of trivial constraints. These constraints are generated by any set containing the three stabilizers associated to a given vertex. Furthermore, the collection of all vertex stabilizers whose support is confined to a coordinate plane is also a constraint. So for every composite excitation at a vertex, there is generally the same composite somewhere along the line at the intersection of the two planes containing the two stabilizers forming the composite (there are also three composite configurations which sit on three intersecting lines). It is in this sense that the composite is mobile only along lines and thus deemed ``lineons.'' Logical operators are the same sting operators as the $d=2$ toric code, but they can only be deformed in coordinate planes. For example, stacked pairs of $Z$-type strings can be deformed in the cubic planes, while $X$-type strings can only be deformed in the coordinate planes. Modding out by these deformations, one finds $\dim \mc P_\ell(XC) =12L-6$ \cite{Ma2017} and $R_{XC}=L$ indicating LRE.  

\begin{figure}

\centering

\includegraphics[scale=.5]{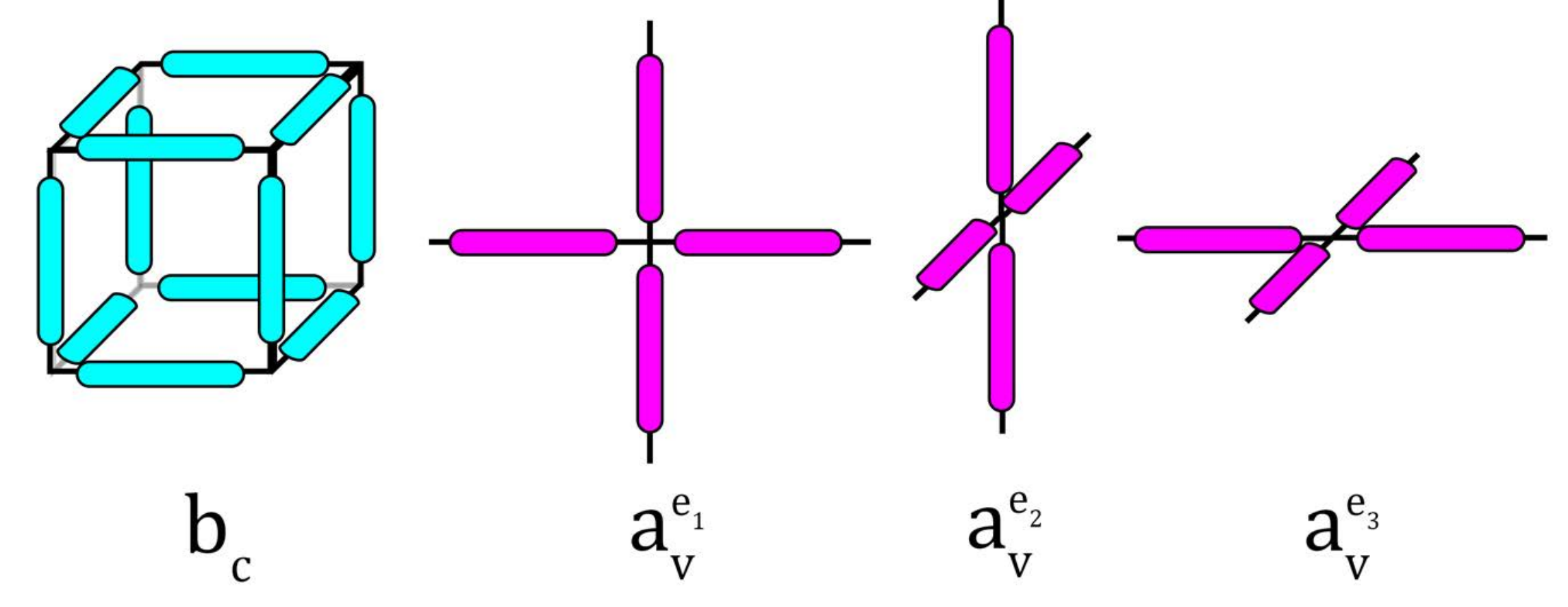}

\caption{Depiction of the stabilizer operators defining the X-cube model.}\label{fig:xcdef}

\end{figure}

\subsubsection{Base Model} 

Once we stack the layers of the toric code, there are $N= 6L^3$ qubits and every edge of the cubic lattice represents a two-qubit unit cell. The base model is then given by
\begin{align}
S_{XX}=\{x_e^1 x_e^2 : e \text{ edges}\}.
\end{align}
Such a stabilizer code is trivial in that $R_{XX}=1$ as given by the fact that $x_e^{1} \simeq x_e^{2} \in \mc P_\ell(XX)$ and its simultaneous eigenstates are product states within the unit cells.

\subsubsection{Perturbation Analysis}

We now apply Brillouin-Wigner perturbation theory to the X-cube example. As previously stated, the stabilizer nature of the parent and base models allows us to find the terms of the effective Hamiltonian in the Pauli operator basis by considering only products of the parent model which commute with all stabilizers of the base model. The $a_v$ terms of the toric code stacks survive at first order as they already commute with the base model stabilizers. These terms alone do not fully lift the extensive base model degeneracy, so we are forced to continue to higher orders. Considering the $b_p$  terms of the parent model, it is not until sixth order that some terms survive. In particular, the product of the six plaquette operators about a cube has support in each unit cell of $z_e^1 z_e^2$ which locally commutes with the base terms $x_e^1 x_e^2$. It was argued in Ref.~\cite{Ma2017} that no terms of order less than sixth survive, but we make an argument here as it carries over analogously for the other models. Consider the commutation of base model terms \{$(xx)_e =x_e^1 x_e^2\}$ with any product of parent model stabilizers. All such products are represented by the space $\mc A_{\left(TC_2\right)_{\text{stacks}}}=\wp\left(S_{TC_2}\right)^{3L}$. Let all products which commute with the base model be represented by a subspace $\mc L \subseteq \mc A_{\left(TC_2\right)_{\text{stacks}}}$. For $A \in \mc L$ and using  Eq.~\eqref{eq:stabgauge}, 
\begin{align}
0=\lambda\left(\phi(A), (xx)_e\right) = \omega\left(A,\psi_{\text{parent}}\left((xx)_e\right) \right) .
\end{align}
So we require that $\mc L = \left( \psi_{\text{parent}}[\mc G_{\text{base}}]\right)^{\perp_\omega}= \left( \im \psi_{\text{parent}}\phi_{\text{base}}\right)^{\perp_\omega}$ which is to say the terms which survive must be orthogonal to the image under $\psi_{\text{parent}}$ of the group of operators generated by the base model terms. That is, we look at the space of parent model excitations generated by base model terms, and any products which overlap an even number of times with those excitations represent the terms which survive the perturbation. The set $\psi_{\text{parent}}\left((xx)_e\right)$ is a small p-string wrapping the edge $e$ similar to the one shown in Fig. \ref{fig:smlp} and from this we see that $\im \psi_{\text{parent}} \phi_{\text{base}}$ is nearly the entire space of p-strings, $\left(\mc D^Z\right)^{\perp_\omega}$--the actual imagine does not contain the topologically non-trivial p-strings which wrap the 3-torus. Thus we conclude that $\mc L = \mc D^Z \oplus  \wp\left(S^X_{TC_2}\right)^{3L}$.\footnote{ Technically, there are additional members contained in $ \left( \im \psi_{\text{parent}}\phi_{\text{base}}\right)^{\perp_\omega}$. Because the image only contains closed, trivially-contractible p-strings, its perpendicular complement contains all of $\mc D^Z$, which includes all even numbers of membranes wrapping the cycles, as well as the topologically non-trivial odd numbers of membranes wrapping the cycles, which are not in $\mc D^Z$. However, the image of these wrapping membranes under $\phi_{\text{parent}}$ are also in the image of $\mc D^Z$ because of the planar toric code constraints, so this technical detail does not change the statement.} The lowest weight generators for $\mc D^Z$ are the sets $\{b_p: p \in c\}$ for every $c$ cube, i.e. the set of six plaquettes which form $c$. Thus no lower order product of plaquettes survive the perturbation. Also note that although we truncate the higher order terms of the effective Hamiltonian, such terms only survive if they land in $\mc D^Z$, so the target model eigenstates are exact and the higher order terms only modify the energy.   Finally to see that we completely recover all properties of the X-cube model, we recognize that the infinite perturbation fixes the degrees of freedom associated with the base model such that  $x_e^{1} \simeq x_e^{2}$ and as a result,
\begin{align}
\prod_j a_v^{\hat e_j} \simeq \id_{\mc H},
\end{align}
which we recognize as the trivial constraint among the vertex terms at $v$.

From the argument for the survival of the sixth order perturbation, we see this process as condensing the p-strings. Moreover, we also recognize the LRE of the X-cube model is inherited from the parent $d=2$ toric code layers. Namely, the string logical operators have an identical form as those from the toric code layers. The only distinction is that due to the coupling, the equivalence classes for these logical string operators are altered. For example, Z-type string operators no longer unambiguously belong to a specific plane, but rather lie at the intersection of two planes. Thus, a single string is ridge along that intersection, but a pair forming a ribbon can be deformed in the plane perpendicular to the stacking direction. This is related to the lack of mobility of vertex bound state excitations. One can also see that topological constraints are inherited from toric code layers. We leave a more detailed discussion of how the anyonic nature of the parent model can be used to understand the fractonic behavior of the X-cube model to Refs.~\cite{Ma2017, Vijay2017}. We do this in part as the parent models for the novel cases below are not anyonic so such details are not helpful for understanding the results below.  

\subsubsection{Connection to $d=3$ Toric Code as a Signature of Condensation} 

As pointed out in Ref.~\cite{Ma2017, Vijay2017}, if one were to instead alter the base model to be
\begin{align}
S_{ZZ}= \{z_e^1 z_e^2 : e \text{ edges } \},
\end{align}
the resulting target model becomes the $d=3$ toric code. The $d=3$ toric code Hilbert space is the same as that of the X-cube model and the stabilizer set is given by
\begin{align}
S_{TC_3} =\{a^{TC_3}_v, b_p : v \text{ vertices  and } p \text{ plaquettes } \},
\end{align}
where $a^{TC_3}_v = \prod_{e @ v} x_e$ or all X-type operators coordinated to $v$  and $b_p$ is the same plaquette operators from the $d=2$ toric code. This model is the prototypical stabilizer code example of TO in $d=3$.  Excitations of the electric sector, or $\{a^{TC_3}_v\}$ terms, must be created in pairs as enforced by the constraint consisting of all such operators. Excitations of the magnetic sector, or $\{b_p\}$ terms, must form p-strings  as enforced by the trivial constraint that the product of the six plaquette operators about any elementary cube is the identity. These excitations must also satisfy a total flux conservation as any p-string must pass through any plane wrapping the 3-torus--which is a constraint--an even number of times. The logical operators are given by Z-type strings that wrap a single direction of the 3-torus, where all such strings wrapping a given cycle are equivalent, and X-type planar operators which wrap any pair of directions of the 3-torus, where all such planar operators wrapping given directions are equivalent. As a consequence $\dim \mc P_\ell(TC_3)= 6$, $R_{TC_3} = L$ and this model is LRE.

The perturbation analysis works exactly the same way, but the subspace of the toric code stacks that survives the perturbation is $ \wp\left(S^Z_{TC_2}\right)^{3L} \oplus \overline{\mc D}^X$ where $\overline{\mc D}^X$ is generated by the sets $\{ a_v^{\hat e_1}, a_v^{\hat e_2}, a_v^{\hat e_3}\}$ for every $v$ vertex i.e. the three stabilizers associated to a vertex. Likewise, $z_e^1 \simeq z_e^2$, which implies for every cube $c$,
\begin{align}
\prod_{p \in c} b_p \simeq \id_{\mc H}.
\end{align}

One should recognize the connection between the three models $\left(TC_2\right)_{\text{stacks}}$, X-cube and $d=3$ toric code. The $XX$ base model is used exactly to condense the p-string excitations of the $d=3$ toric code, whereas the $ZZ$ base model is used to condense the lineon composite excitations of the X-cube model. Likewise, the subspace $\mc D^Z$ which survives the $XX$ base model perturbation is exactly the space which forms the trivial kernel of $\phi_{TC_3}$ i.e these operators are ``removed'' by becoming the identity, whereas the space $\overline{\mc D}^X$ which survives the $ZZ$ base model perturbation is exactly the space which forms the trivial kernel of $\phi_{XC}$. We also recognize that $ \dim \mc P_\ell \left(\left(TC_2\right)_{\text{stacks}}\right) = \dim \mc P_\ell(XC) + \dim \mc P_\ell(TC_3)$. Thus, it would seem that all the LRE of the toric code stacks is ``conserved'' between these two models. All these facts are not coincidental. We argue below that these are all a consequence of the statement, 
\begin{align}\label{eq:tcdecomp}
\left(TC_2\right)_{\text{stacks}} \simeq TC_3 \oplus XC, 
\end{align}
where equivalence and direct sum are defined in the category of $\mb F_2$-linear gauge structures. The existence of such a {\it co-target} model--which the $d=3$ toric code is in this case-- is exactly our definition of a condensation process and is made precise in Section \ref{sec:dsum}. Any case where no such co-target exists is then defined as a distillation process.

\subsection{Distillation Example for Type I Fractons: Cluster-cube}

Our next example suffices as a distillation of Type-I FTO from an SSPT model. The parent model is the $d=2$ cluster model and the target model is a new model we refer to as the cluster-cube model, which we describe in detail below.

\subsubsection{Parent Model}

\begin{figure}

\centering

\includegraphics[scale=.5]{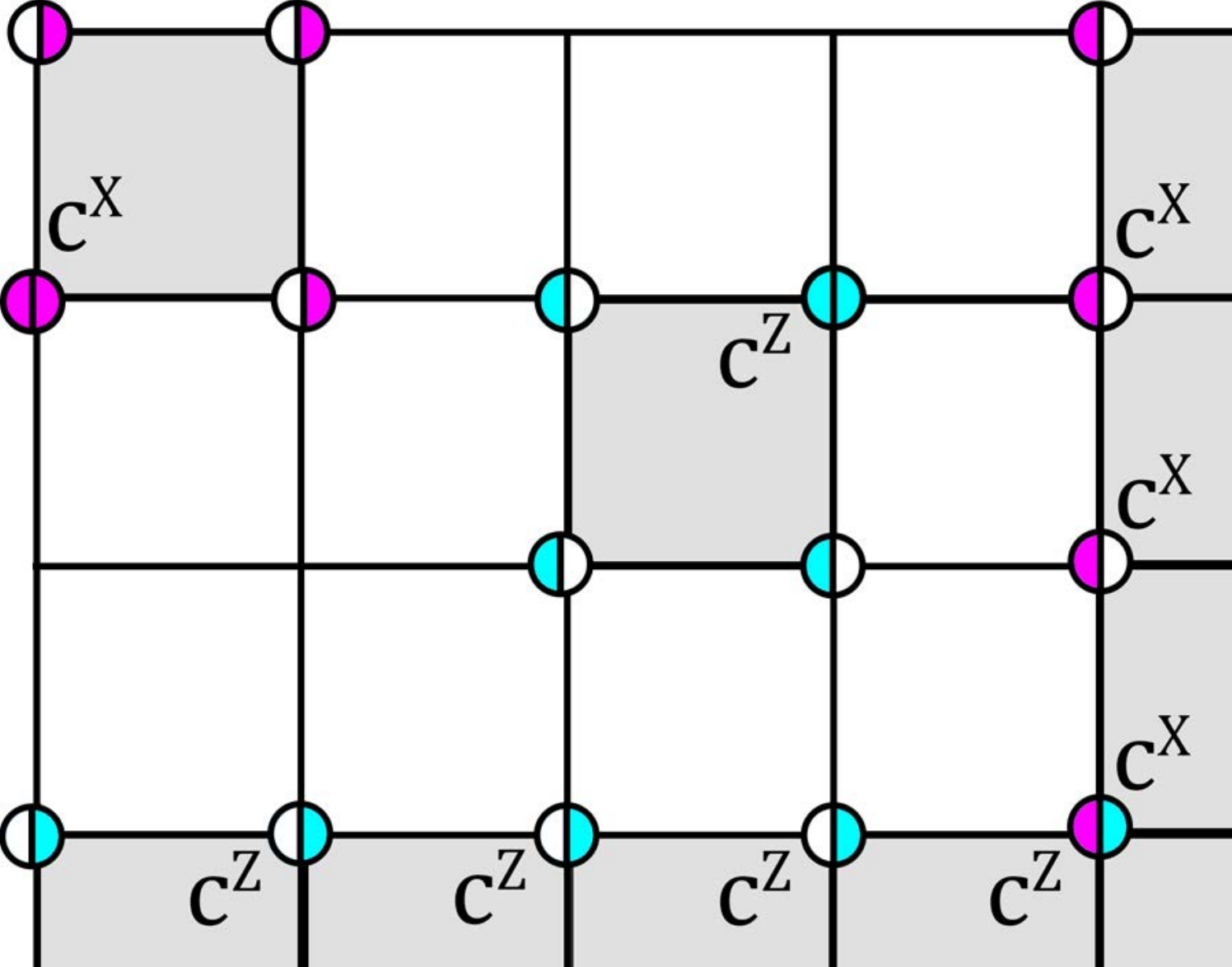}

\caption{Depiction of the stabilizer operators and subsystem symmetries of the cluster model.}\label{fig:clusdef}

\end{figure}

 The parent model is composed of a CSS variant of the cluster model as introduced in Ref.~\cite{Schmitz2019b}. The Hilbert space is that of $N=2L^2$ qubits with two qubits associated with each vertex. The stabilizer set is given by
\begin{align}
S_{Cl} =\{c^X_v, c^Z_v: \text{ vertex } v\},
\end{align}
where $c^X_v, c^Z_v$ are defined in Fig.~\ref{fig:clusdef}.\footnote{To see this equivalent to the cluster model, One can follow Ref.~\cite{Williamson2018} where the usual cluster model has been coarse-grained from one qubit per vertex to two qubits per vertex. Then by applying a Hadamard gate to every second qubit in a vertex unit cell, one finds our version.} To understand this model, consider applying $z^1_v \id_v^2$ for some vertex $v$. This excites only $c^X_v$, and likewise, $\id_v^1 x_v^2$ only excites $c_v^Z$.  Using general combinations of these operators, this implies that we can achieve any syndrome/excitation configuration. This further implies the constraint space is trivial by Theorem~\ref{prop:stabcon}, and because the number of stabilizers is equal to the number of qubits, it must be that $d_\ell =0$ and $ R_{Cl}=0$. Thus, the cluster model contains no LRE, and it would seem there is nothing of interest. However, one can {\bf assert} that there are special members of $\mc G_{Cl}$ which are referred to as subsystem symmetries (SS). Consider the product of all $c_v^X$ along any rigid line in a coordinate direction, and likewise for $c_v^Z$. The resulting operators are
\begin{subequations}\label{eq:sscluster}
\begin{align}
s_{e_1}^X = \prod_{e_2} (x\,\id)_{(e_1, e_2)}=\prod_{e_2} c^X_{(e_1, e_2)} \in \mc G_{cl} ,\\
s_{e_2}^X = \prod_{e_1} (x \,\id)_{(e_1,e_2)}= \prod_{e_1} c^X_{(e_1,e_2)}\in \mc G_{cl}, \\
s_{e_1}^Z = \prod_{e_2} (\id\, z)_{(e_1,e_2)}= \prod_{e_2} c^Z_{(e_1,e_2)}\in \mc G_{cl}, \\
s_{e_2}^Z = \prod_{e_1} (\id\, z)_{(e_1,e_2)}=\prod_{e_1} c^Z_{(e_1,e_2)} \in \mc G_{cl},
\end{align}
\end{subequations}
where $(e_1,e_2) \in \mb Z_L^{2}$ indexes a vertex of the graph, as shown in Fig.~\ref{fig:clusdef}. As argued in Ref.~\cite{Schmitz2019b}, these stabilizer group members gain significance if we limit the operators which act on our system to a subspace of the Pauli space $\mc R \subseteq \mc P$ containing all operators which commute with the subsystem symmetries. As the subsystem symmetries are in $\mc G_{Cl}$, $\mc R$ must contain all of $\mc G_{Cl}$, in which case it is reasonable to only consider $\mc R/\mc G_{Cl}\simeq \mc R'$, where $\mc R'$ is the set of all operators generated by $x^1_v \id_v^2$ and $\id_v^1 z_v^2$ for all $v$. Note this  excludes  operators generated by $z^1_v \id_v^2$ and $\id_v^1 x_v^2$, the exact operators we used to generate an arbitrary syndrome. We are now free to ignore any part of the stabilizers which commute with all $\mc R'$ since they no long affect the syndromes as generated by the function $\psi$ after we limit ourselves to $\mc R'$. The subspace that commutes with all of $\mc R'$ is just $\mc R'$ as it is a maximal, mutually commuting set. As a result, all subsystem symmetries in Eqs.\eqref{eq:sscluster} effectively become constraints and as a result imply effective conservation laws via Theorem~\ref{prop:stabcon}. Looking back to Fig.~\ref{fig:clusdef}, one can see that removing operators from $\mc R'$ makes the cluster model into two copies of the Ising plaquette model. The Ising plaquette model is a classical model with one qubit (or classical spin) per vertex of the square lattice and stabilizers given by
\begin{align}
S_{IPM}= \{\prod_{v \in p} z_v : p \text{ plaquettes }\}.
\end{align}
Action by $x_v$ creates a $\mb Z_2$  quadrupole and in general, excitations must satisfy $\mb Z_2$ conservation laws along all rigid $d=1$ subsystems in the coordinate directions, exactly the conservation laws implied by Eqs.~\eqref{eq:sscluster} when all operators from $\mc R'$ are removed. This model is also the paradigm example of fractons in two dimensions \cite{Yan2019} (though not FTO strictly speaking). This process of restricting to subspaces of operators is exactly the idea behind a gauge substructure as introduced in Section \ref{sec:gs}. 

\subsubsection{Target Model}\label{sec:ccdesc}

We refer to the target model as the cluster-cube model and to the best of our knowledge, this is its first introduction in the literature. The Hilbert space is that of $N=2L^3$ qubits arranged on the cubic lattice such that there are two qubits associated to each vertex. The stabilizer set is given by
\begin{align}
S_{CC}= \{cc_v^X, cc_v^Z: \text{ vertices } v\}, 
\end{align}
where $cc_v^X, cc_v^Z$ are defined in Fig.~\ref{fig:ccubedef}. We focus on the $cc_v^X$ terms, where all the same properties are present for the $cc_v^Z$ terms if one applies the clear duality between these two types of stabilizers. If we apply the operator $z_v^1 \id_v^2$ at some vertex $v$, one excites four cube stabilizers whose centers form an ``upward-facing'' right-angled tetrahedron (henceforth just referred to as a UF tetrahedron). Likewise, applying the operator $z_v^1 z_v^2$ at $v$ generates a tetrahedral pattern which is symmetric to the first by inversion in the $[111]$ direction or a ``downward-facing'' (DF) tetrahedron. As a result, one can generate the syndrome such that all eight cubes about $v$ are excited. Applying the same operator to an adjacent vertex is equivalent to pulling apart two bound states of four excitations. One can continue in this fashion until the four-excitation bound states are annihilated and so the resulting string operator must commute with all $\mc G_{CC}$. It should be clear that this string operator is not a member of $\mc G_{CC}$ and is therefore in $\mc P_\ell(CC)$. To characterize such string operators as logical operators, we need to determine their equivalence classes. Consider taking the product of cube stabilizers along a ridge line in a coordinate direction--similar to that of the subsystem symmetries for the cluster model. If this one dimensional subsystem is, for example, along the $[100]$ direction, then the resulting operator is equivalent to the product of two of our logical operators which both lie in the plane perpendicular to the $[011]$ direction. This implies all logical string operators in the $[011]$ plane are equivalent in $\mc P_\ell (CC)$, and likewise for stings in the $[101]$ and $[110]$ planes. The ``dual'' logical operators which anti-commute with the string logical operators are diagonal string operators which are discussed in detail in Section \ref{sec:distCC}. Furthermore, as we can take pairs of string operators in the $[011], [101]$ and $[110]$ planes, move them apart via the $d=1$ subsystem product of stabilizers and eventual cancel them out so as to remove all support, this implies that any set of cubes forming a $[110], [101]$ or $[011]$ plane is a topological constraints and implies a planar $\mb Z_2$ charge conservation similar to the X-cube model. Also like the X-cube model,  $[110], [101]$ and $[011]$ constraints intersect at a single cube stabilizer, implying fractonic behavior similar to X-cube. As the four-cube bound state is mobile and are found at the end of string operators, this model fits the definition of a Type I fracton model. 

\begin{figure}

\centering

\includegraphics[scale=.75]{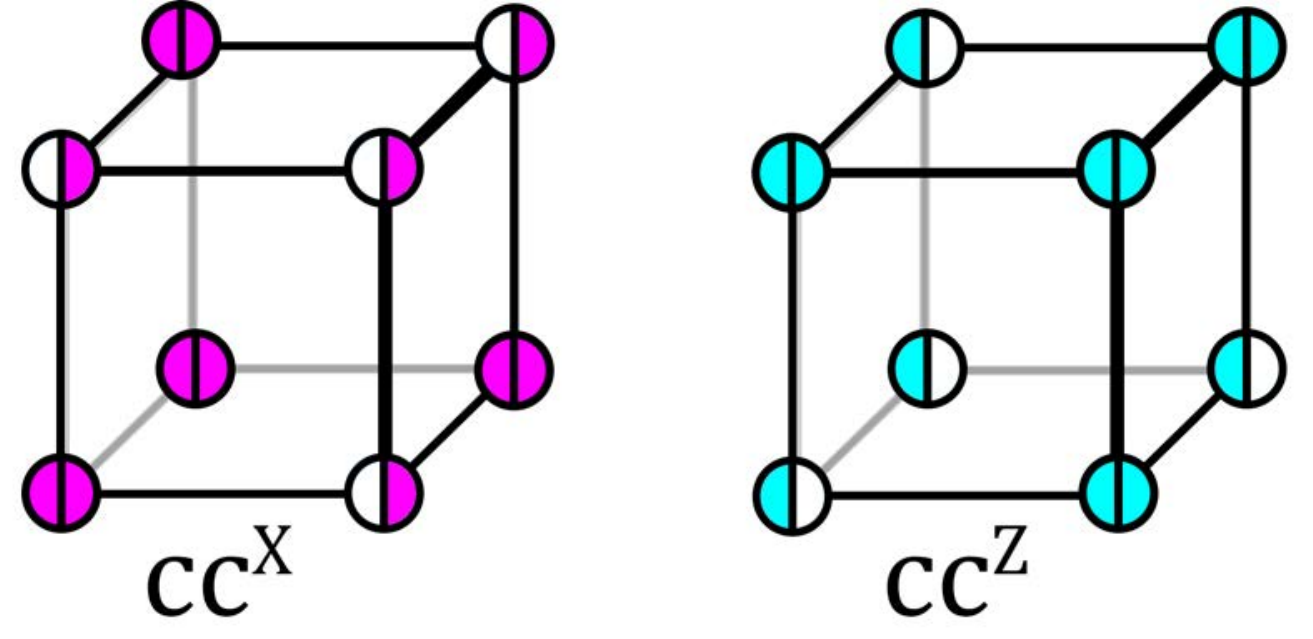}

\caption{Depiction of the stabilizer operators of the cluster cube model.}\label{fig:ccubedef}

\end{figure}

One might be tempted to classify this model as equivalent to X-cube, but string operators are not the only logical operators. If we start from either one of the tetrahedral configurations, we can apply the same operator to adjacent vertices, thereby un-exciting three of the original tetrahedron of cubes and exciting four cubes of a larger tetrahedron as shown if Fig. \ref{fig:fracbuild}. This process can be continued such that we excite the corners of larger and larger tetrahedra and if the system size is such that $L=2^a$ for some integer $a$, then these corners annihilate with one another after $a$ generations. The resulting operator which commutes with all of $\mc G_{CC}$ can be described as a three-dimensional, right-angled version of the Sierpinski fractal. We argue below that this cannot be equivalent to any combination of string operators, nor is it a product of stabilizers. However, not all translated versions of such fractal logical operators are distinct in $\mc P_\ell(CC)$. Consider that a cube stabilizer can be thought of as an UF and DF tetrahedron which are ``glued'' together. The UF tetrahedron is  formed from only $x_v^1 x_v^2$ operators and the DF tetrahedron is formed from only $\id_v^1 x_v^2$ operators. In the same manner as our process for forming the fractal logical operators, we can remove all support of a given type, say $x_v^1 x_v^2$, out to the corners of an ever larger UF tetrahedron and at the $a^{th}$ generation, those four points overlap and all such support is removed. We consider this operator to be a fractal SS operator and is equivalent to the product of four fractal logical operators. To see this, consider our fractal product of stabilizers at some generation less than $a$. Now imagine we remove the four  $x_v^1 x_v^2$ corners and consider the syndrome the resulting operator creates. Naturally, this operator generates the same syndrome as the four  $x_v^1 x_v^2$ corners we removed and can be described as four elementary DF tetrahedral configurations (among the the $cc_v^Z$ operators) centered on vertices of an UF tetrahedron which then annihilate with each other at the $a^{th}$ generation. This implies our fractal SS must be equivalent to four of the UF fractal logical operators arranged in a DF tetrahedron configuration. Furthermore as this is a stabilizer group element, these four fractal logical are equivalent to the identity in $\mc P_\ell(CC)$. Such fractal SS operators also allow us to understand the additional fractal constraints which are implied by the existence of the fractal logical operators. If we represent a fractal logical operator by a single base-point, then a UF fractal SS operator is equivalent to four logical operators whose base-points form an elementary DF tetrahedron. Again, we can combine SS fractals in an analogous DF fractal such that the four logical operator base-points are moved along the corners of a fractal and eventually annihilate with one another. The resulting product of SS fractals, which are themselves products of stabilizers must be equivalent to the identity, and thus this collection of stabilizers is a constraint which enforces the fractal hopping of the excitations.

At this time, we make no attempt at counting the number of independent logical operators, i.e. calculate the GSD for this model, as the exact number is not important for the remaining discussion and may be a complicated function of $L$. We only note  that $\log_2$GSD scales as $\sim L$ which is a common feature of fractonic models. The string operators imply that $R_{CC}= L$ which implies this model contains LRE.

\begin{figure}

\centering

\includegraphics[scale=.22]{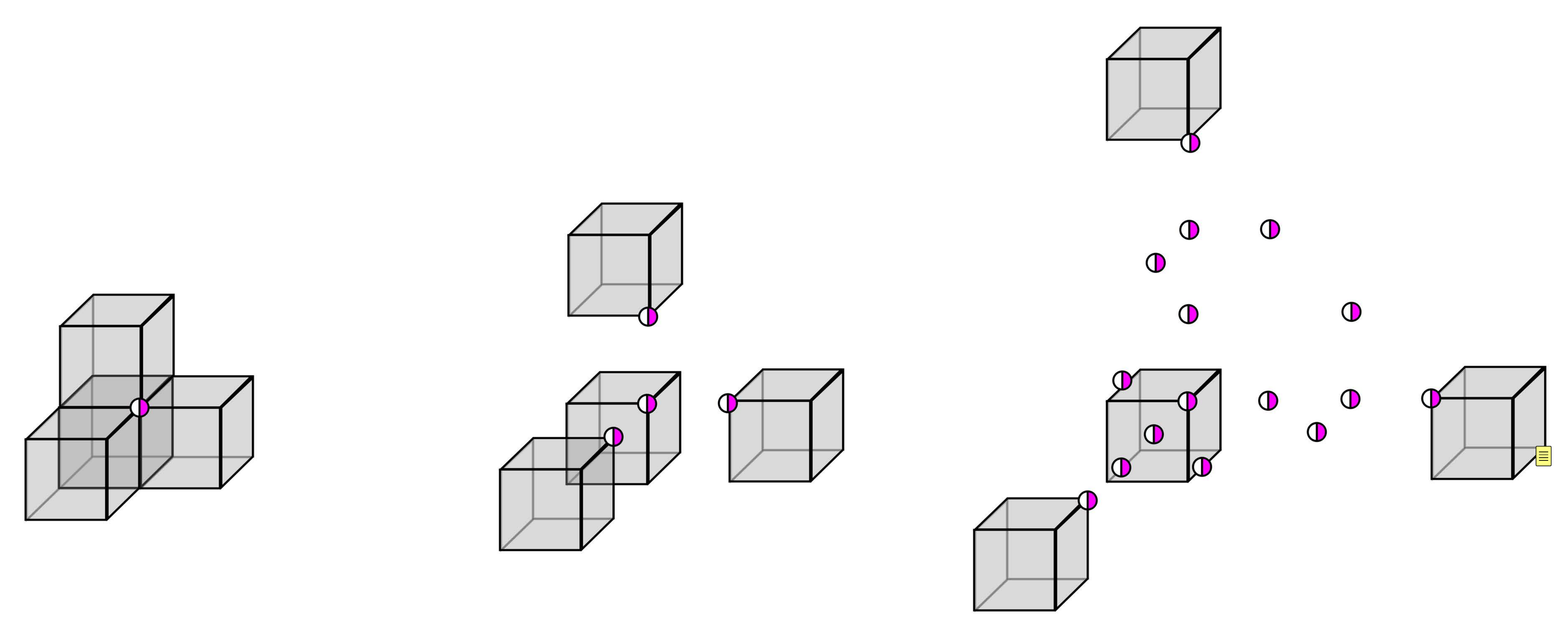}

\caption{Depiction of the process used to build the fractal logical operators.}\label{fig:fracbuild}

\end{figure}

\subsubsection{Base Model}\label{sec:ccbase}

\begin{figure}

\centering

\includegraphics[scale=.4]{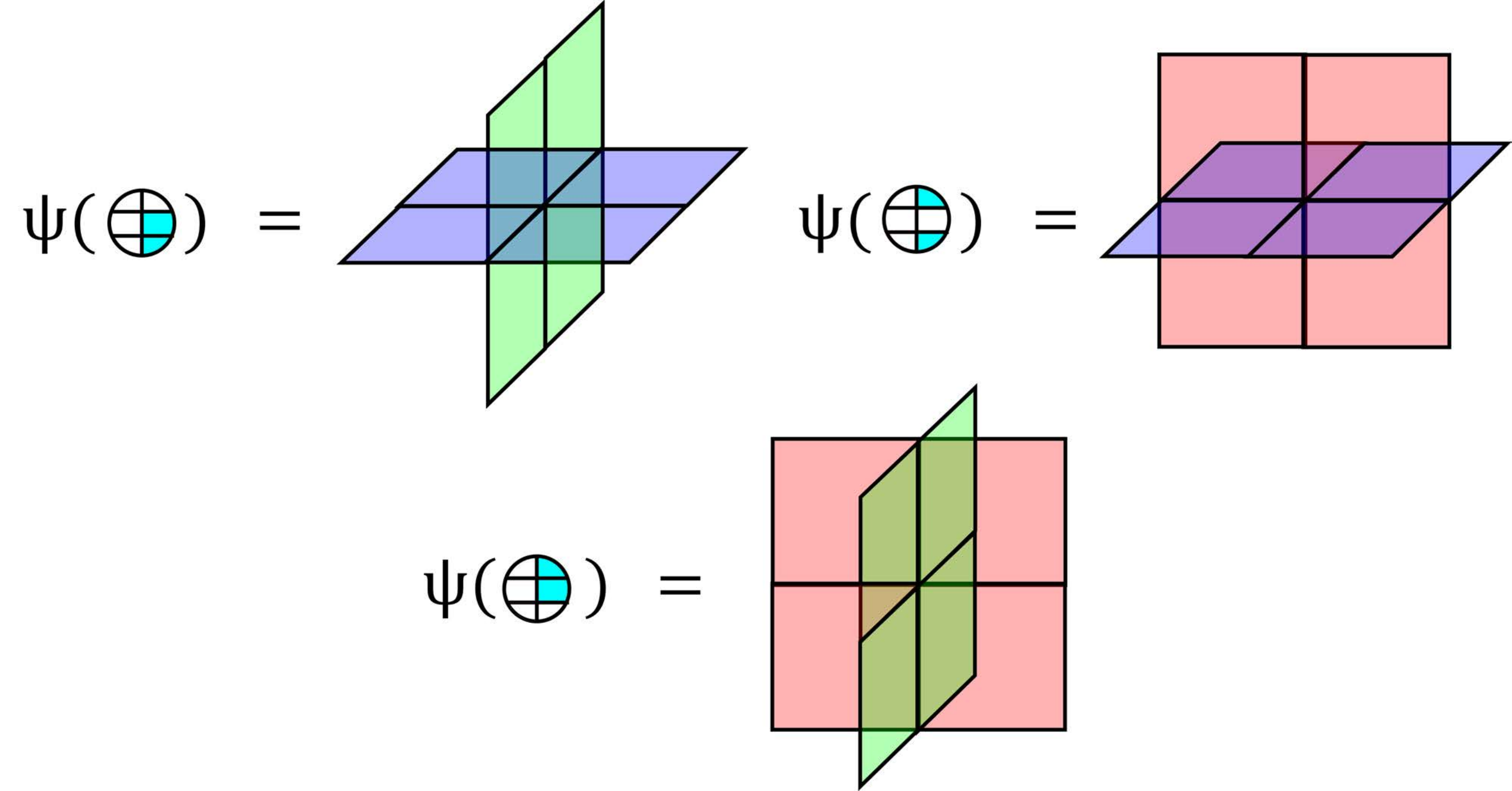}

\caption{X-type loop configurations generated by the Z-type members of $S_{c\hyph \triangle}$. The X-type members generate analogous loop configurations.}\label{fig:2loop}

\end{figure}

Once we combine the three stacks of cluster models, the resulting Hilbert space can be characterized as  $N=6L^3$ qubit where each vertex represents a six qubit unit cell. It is best to further break this up into two unit cells of three qubits. For ease of notation, we then define the following three-qubit operators,
\begin{subequations}\label{eq:3qdef}
\begin{align}
(xxx)^0 =& x^1 x^2 x^3,\\
(xxx)^1 =& \id^1 x^2 x^3, \\
(xxx)^2=& x^1 \id^2 x^3, \\
(xxx)^3=& x^1 x^2 \id^3, \\
(zzz)^0 =& z^1 z^2 z^3,\\
(zzz)^1 =& \id^1 z^2 z^3, \\
(zzz)^2=& z^1 \id^2 z^3, \\
(zzz)^3=& z^1 z^2 \id^3. 
\end{align}
\end{subequations}
Note that $ (xxx)^1(xxx)^2(xxx)^3 =\id$, $(zzz)^1(zzz)^2(zzz)^3 =\id$, and the commutation relations can be summarized as $\lambda\left((xxx)^i,(zzz)^j\right) =1$ for $i\neq j$ and $i,j= 1,2,3$, $\lambda\left( (xxx)^0, (zzz)^0\right) =1$ and all other pairs commute.

For reasons we discuss below, there are multiple possible base models, but we focus on one as it is the most ``local''. The stabilizers are given by
\begin{align}
S_{c \hyph\triangle}= \{ (xxx)^i_v\otimes \id_v, \id_v \otimes (zzz)_v^i: v \text{ vertices and } i=1,2,3 \}.
\end{align}       
Just as with the X-cube example, this perturbation model is trivial with $R_{c\hyph\triangle} =1$  and simple product eigenstates in the local vertex Hilbert spaces. The name $c\hyph\triangle$ is explained below.

\subsubsection{Perturbation Analysis}

To understand the structure of the stacking, we organize the layers so the $x_v^1x_v^2$  base points of $c_v^X$ in each layer meets at the ``origin point'' of the elementary cube. This maintains a three-fold rotation symmetry about the $[111]$ axis. As a consequence, the $z_v^1z_v^2$ base points of $c_v^Z$ in each layer meet at the origin antipode of the elementary cube (see Fig. \ref{fig:ccbuild}). 

Just as with the X-cube example, the perturbation analysis of the cluster cube requires us to determine which product of parent model stabilizers commute with the base model. For $\mc A_{Cl_{\text{stacks}}}=\wp\left(S_{cl}\right)^{3L}$, let $\mc D^X$ be the same subspace generated by the set of six X-type stabilizers which form a cube and likewise for $\mc D^Z$ using Z-type stabilizers. X-type p-strings are then members of $\left(\mc D^X\right)^{\perp_\omega}$ and Z-type p-strings are members of $\left(\mc D^Z\right)^{\perp_\omega}$.  Following the same procedure as X-cube, we consider the space $\im \psi_{\text{parent}}\phi_{\text{base}}$ which is partially generated by the set $\psi_{\text{parent}}\left(   \id_v \otimes (zzz)_v^i\right)$. From the definitions in Eqs.~\eqref{eq:3qdef}, one finds that $\psi_{\text{parent}}\left(   \id_v \otimes (zzz)_v^i\right)$ contains two small X-type p-strings about the two edges coordinated to $v$ and extending in the $\hat e_i$ direction as depicted in Fig.~\ref{fig:2loop}. Likewise, $\psi_{\text{parent}}\left( (xxx)^i_v\otimes \id_v \right)$ contains similar Z-type p-strings. Note these pairs of p-strings are rigid along this line, i.e. no operators of the base model can ``turn'' one p-string relative to the other in any plane. This already suggests that the set of  terms which survive the perturbation, specifically those from $\left( \im \psi_{\text{parent}}\phi_{\text{base}}\right)^{\perp_\omega}$ includes the subspace $\mc D^X \oplus \mc D^Z$. Again, other products can be found in $\left( \im \psi_{\text{parent}}\phi_{\text{base}}\right)^{\perp_\omega}$, even more so than in the X-cube example, and furthermore, we lack any constraints. In particular, $\left( \im \psi_{\text{parent}}\phi_{\text{base}}\right)^{\perp_\omega}$ contains all sets of configurations forming a rigid $d=1$ line in any plane --exactly those forming the subsystem symmetries of the cluster model layers, and so these products survive and are not trivial. Such operators  only appear at $L^{th}$ order in the perturbation and are exponentially suppressed in $L$, however, we show in Section \ref{sec:distCC} that the appearance of such operators ``breaks'' the symmetry of the ground space and fixes which target model ground state we obtain.

So again, the lowest order to survive the perturbation is sixth order and this contributes terms to the effective Hamiltonian given by the product of stabilizers which form a cube and only products of these cubes survive for higher orders, except for the subsystem symmetries which appear at $L^{th}$ order. This alone does not guarantee that the effective Hamiltonian corresponds to the cluster-cube model. If we take the product of cluster model operators which form the cube as in Fig.~\ref{fig:ccbuild}, we find that in terms of the operators defined in Eqs.~\eqref{eq:3qdef}, the resulting product is the cluster cube operators written in terms of the $(xxx)^0$ or$(zzz)^0$ operators times ($+$ in $\mc P$) a triangle operator, $t_c^X$ or $t_c^Z$, which is written in terms of the $(xxx)^i$ or $(zzz)^i$ operators for $i=1,2,3$. Even more specific, they are written in terms of the base model operators which implies $t_c^X, t_c^Z \in \mc G_{\text{base}}$. So as a result of the infinite perturbation which fixes the degrees of freedom for the base model, $t_c^X, t_c^Z \simeq \id_{\mc H}$. This implies our cube operators become that of the cluster cube on effective qubit degrees of freedom defined by the local Pauli algebra generated by $(xxx)^0$ and $(zzz)^0$.

\begin{figure}

\centering

\includegraphics[scale=.4]{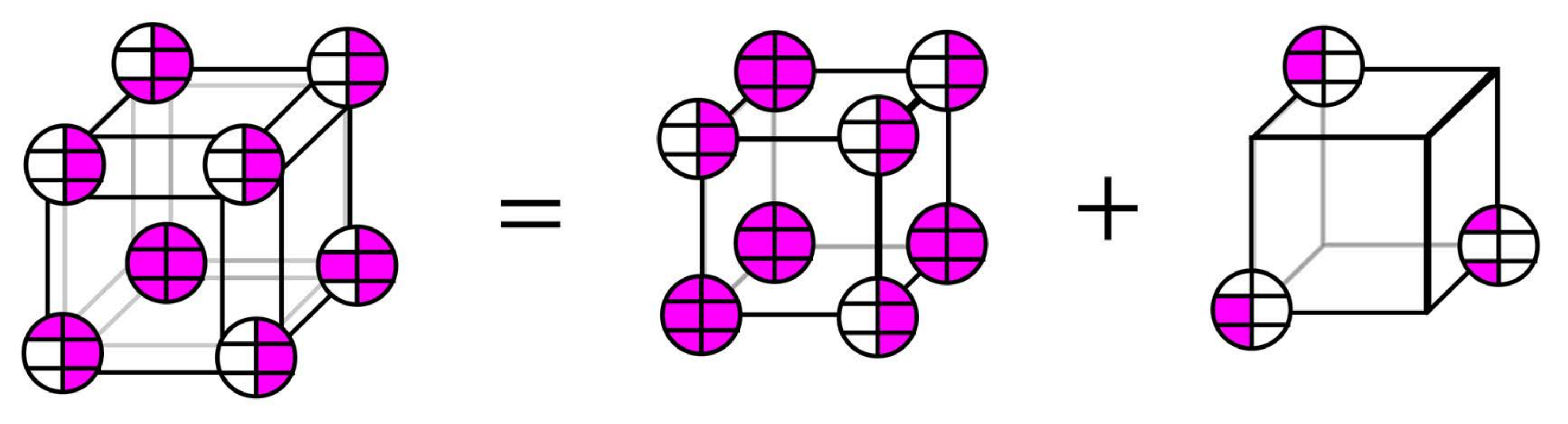}

\caption{Demonstration of how the product of $X$-type cluster model plaquettes forming a cube is broken down to a sum of effective degrees of freedom and members of $G_{c\hyph \triangle}$. }\label{fig:ccbuild}

\end{figure}
 
\subsection{Distillation Example for Type II Fractons: Haah's Cubic Code}

The final case suffices as an example of distillation of Type-II fractons from an SSPT. The parent model is a new SSPT model we deem the quasi-cluster model which is described below and the target model is Haah's cubic code \cite{Haah2011, HaahThesis}. 

\subsubsection{Parent Model}

\begin{figure}

\centering

\includegraphics[scale=.5]{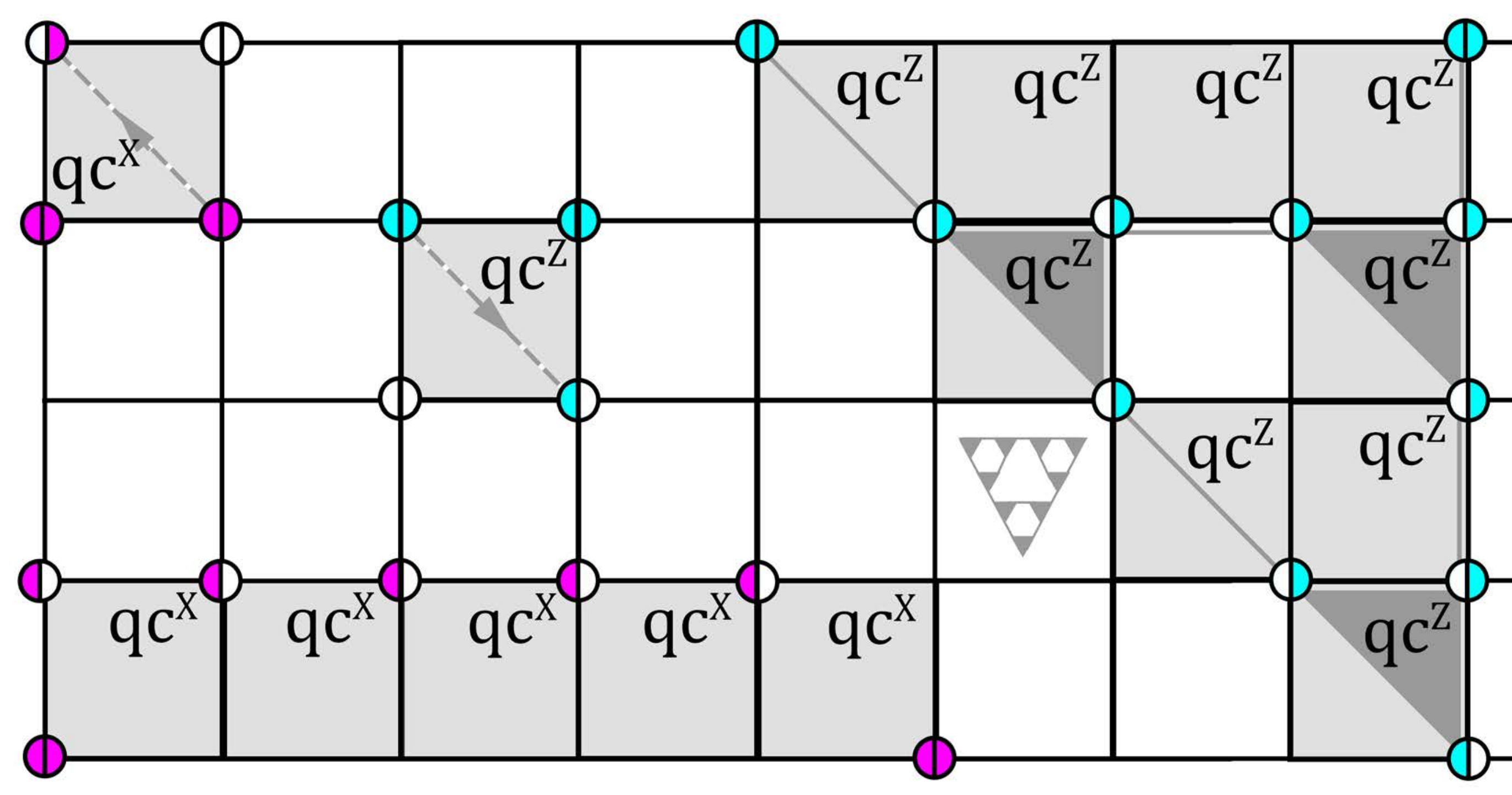}

\caption{Depiction of the stabilizer operators as well as the subsystem symmetries of the quasi-cluster model.}\label{fig:qclusdef}

\end{figure}

The parent model for this example is a $d=2$ SSPT model which we refer to as the quasi-cluster model. To the best of our knowledge, this is the first introduction of this model in the literature. The Hilbert space for the quasi-cluster model is given by $N= 2L^2$ with the stabilizer set given by
\begin{align}
S_{qCl} =\{ qc_v^X, qc_v^Z: v \text{ vertices }\},
\end{align}
which are defined in Fig.~\ref{fig:qclusdef}.\footnote{One can equivalently define this model by reflection of the stabilizers about the main diagonal. Either version can be used for this example.}  We apply the same analysis as with the cluster model, whereby we consider acting with $z_v^1 z_v^2$ for any $v$. This excites only $qc_v^X$. Likewise, acting with $x_v^1 x_v^2$ excites only $qc_v^Z$. By the same logic as used for the cluster model, this implies $d_\ell=0$, $R_{qCl} =0$ and thus this model is SRE. Furthermore, we consider the $d=1$ subsystem symmetries given by 
\begin{subequations}\label{eq:ssdc1}
\begin{align}
s_{e_2}^X = \prod_{e_1} (\id\, x)_{(e_1,e_2)}= \prod_{e_1} qc^X_{(e_1,e_2)}\in \mc G_{qCl}, \\
s_{e_2}^Z = \prod_{e_1} (z \, \id)_{(e_1,e_2)}=\prod_{e_1} qc^Z_{(e_1,e_2)} \in \mc G_{qCl}.
\end{align}
\end{subequations}
We again follow the same logic by considering restricting to $\mc R' \simeq \mc R/ \mc G_{qCl}$, where $\mc R$ is all Pauli operators which commute with the subsystem symmetries. $\mc R'$ can be described as all operators generated by $z_v^1 \id_v^2$ and $\id_v^1 x_v^2$ which forms a maximal, mutually commuting subset of $\mc P$ and just as before, modding out by $\mc R'$ implies our subsystem symmetries are effectively constraints and our stabilizer code becomes two independent stacks of classical Ising chains in the $\hat e_2$ direction.

There is another subsystem symmetry which, although not as clean, essentially has the same properties. For the $qc_v^X$ terms, one recognizes that if we ignore the first qubit at every vertex, these terms have the form of the Newman-Moore model \cite{Newman1999}. This is a classical model of $L^2$ qubits (or classical spins) arranged on the triangular lattice with one qubit  per vertex and the stabilizer set given by
\begin{align}
S_{NM}= \{ \prod_{v \in t}z_v : t \text{  upward triangle } \}.
\end{align}
In our case, we have just skewed the lattice so that the elementary triangles form right angles. For an $L= 2^a$ sized lattice, the product of triangle stabilizers in the form of an $a^{th}$-generation Sierpinski fractal only has support on the three corners of the fractal which now overlap on a single point. This almost suffices to form a constraint and even though it does not do so exactly, it does enforce a local version of Theorem~\ref{prop:stabcon}. This results in the usual fractonic properties of the Newman-Moore model which requires that excitation must be created locally in threes and can only be moved along the ends of a fractal. In the case of the quasi-cluster model, we can form the same fractal to obtain the subsystem symmetries
\begin{subequations}
\begin{align}
s_v^X =\id_v^1 x_v^2 \prod_{v' \in \widetilde{\text{Sierpinski}}_v} x_{v'}^1 \id_{v'}^2 =  \prod_{v' \in\text{Sierpinski}_v} qc_{v'}^X \in \mc G_{qCl}\\
s_v^Z =z_v^1 \id_v^2 \prod_{v' \in \widetilde{\text{Sierpinski}}'_v} \id_{v'}^1 z_{v'}^2 =  \prod_{v' \in\text{Sierpinski}'_v} qc_{v'}^Z \in \mc G_{qCl},
\end{align}
\end{subequations}
where $\text{Sierpinski}_v$, is the usual Sierpinski fractal, and $\widetilde{\text{Sierpinski}}_v$ is the alternative Sierpinski fractal as shown on the right side of  Fig.~\ref{fig:qclusdef}. The primed versions are given by reflecting about the $\hat{e}_1=-\hat{e}_2$ line. Because the base point of  $s_v^X, s_v^Z$ is not a part of the fractal, we cannot apply  the same logic as the other subsystem symmetries. Still, we can restrict to the set of Pauli operators generated by $x_v^1 \id_v^2$ and $\id_v^1 z_v^2$, in which case our stabilizers become two independent copies of the Newman-Moore model and the fractal subsystem maps to the product of Newman-Moore stabilizers and the local fractal charge conservation is enforced.

\subsubsection{Target Model}\label{sec:haahdef}

\begin{figure}

\centering

\includegraphics[scale=.75]{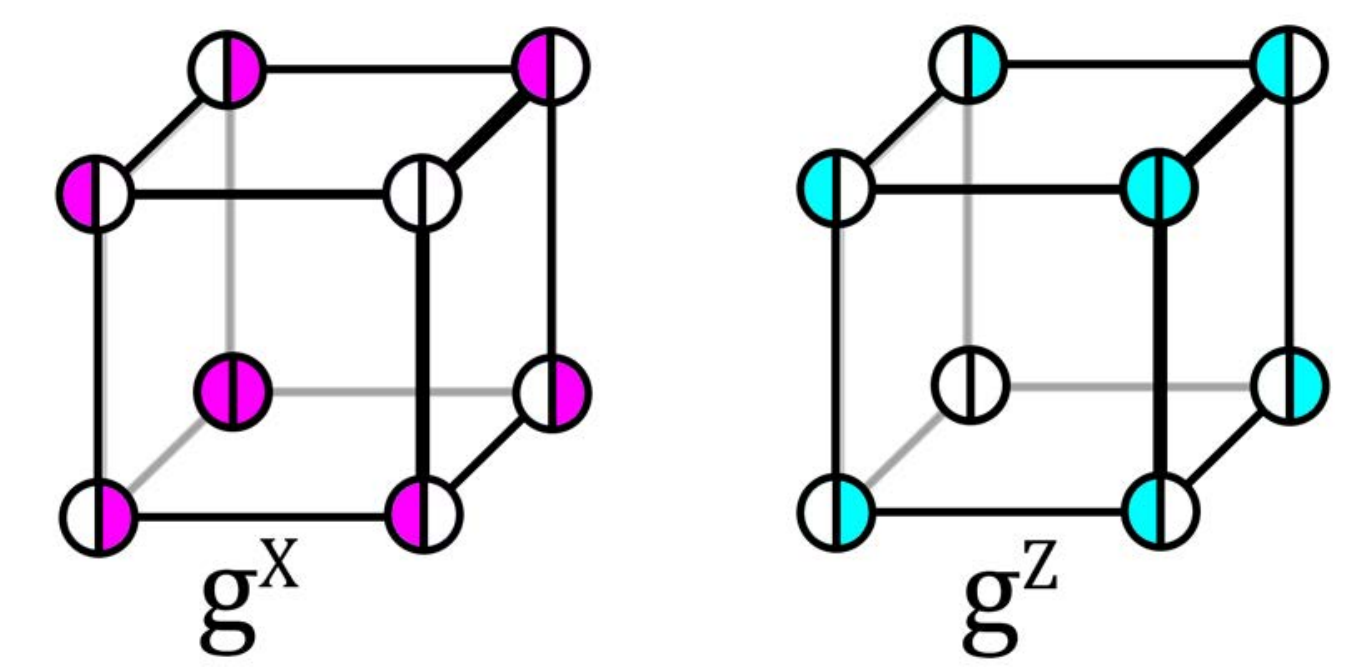}

\caption{Depiction of the stabilizer operators of Haah's cubic code model.}\label{fig:haahdef}

\end{figure}

The target model is Haah's cubic code--or just Haah's code--which is the paradigmatic example of Type-II fractons. The Hilbert space is given by $N=2L^3$ qubit arranged on the cubic lattice such that each vertex represents a two qubit unit cell. The stabilizers are given by 
\begin{align}
S_{\text{Haah}}=\{ g_c^X, g_c^Z: c \text{ cubes } \}, 
\end{align}
where $g_c^X$ and $g_c^Z$ are defined in Fig.~\ref{fig:haahdef}. We primarily study this model by considering the $g_v^X$ operators, but analogous statements hold for the $g_v^Z$ operators by the obvious symmetry between these operators. Applying $ x_v^1 \id_v^2$ generates the same UF tetrahedron excitation pattern as Fig. \ref{fig:fracbuild}, so excitation again move on the ends of a fractal operator and  for $L=2^a$ we find fractal logical operators. Applying $\id_v^1 x_v^2$ instead generates a regular tetrahedron pattern and a similar process results in a different fractal logical operator with a regular tetrahedron geometry. By design, Haah's code contains no string logical operators. To the best of our knowledge, the constraints of Haah's code are not discussed in the literature except for a maximum of 7 found in \cite{Schmitz2019a} which represent all constraints which are periodic in the $[111]$ direction. This clearly cannot be sufficient to characterize all constraint or else excitations would be free to move in the $[111]$ direction, which is not the case. However like the cluster-cube example, we can find a sub-extensive number of constraints that enforce this fractonic mobility. Looking at Fig. \ref{fig:haahdef}, we see that just as with the cluster-cube, we can view $g_v^X$ as the product of two tetrahedron, an UF right-angle tetrahedron formed using $\id_v^1 x_v^2$ operators and a regular tetrahedron using $ x_v^1 \id_v^2$ operators. So for $L=2^a$, we can again multiply the $g_v^X$ operators in a fractal pattern, thereby removing all $x_v^1 \id_v^2$ support, or with a different fractal, all $\id_v^1 x_v^2$ support. By the same arguments as the cluster-cube model, we again find the fractal subsystem symmetry composed of four logical operators, and just as before, these SS can be multiplied in a fractal way such that the resulting product maps to the identity, thus constituting a constraint. The number of independent constraints of this form must scale with the system size. It is known that for $L=2^a$ \cite{Haah2011,HaahThesis}, $d_\ell =4L-2$ which must also be the dimension of the constraint space. 
As there are no string operators for this code, the fractal logical operators  determine the code distance. Counting the number of single qubit Pauli's needed to form these fractals, one finds $R_{\text{Haah}} \leq 4^{(a-1)} =\frac{1}{4} L^2$. Though these fractal logical operators are deformable and this bound may not be tight, it should be clear that $R_{\text{Haah}}\gtrsim L$, so this model contains LRE \cite{Haah2011}.              


\subsubsection{Base Model}

\begin{figure}

\centering

\begin{tabular}{c}
\subfloat[\label{fig:hexdefa}]{\includegraphics[scale=.45]{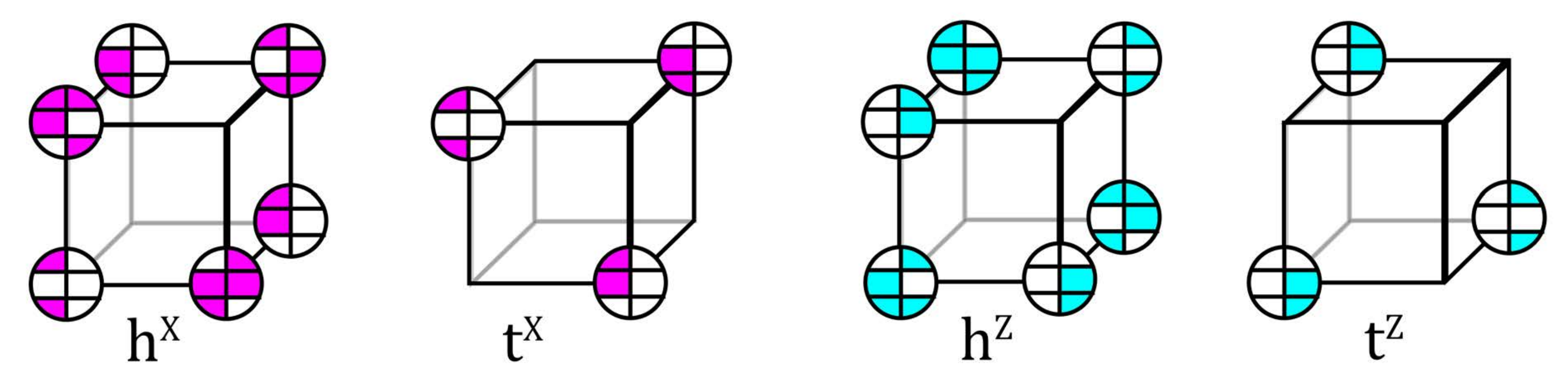}} \\
\subfloat[\label{fig:tritc}]{\includegraphics[scale=.45]{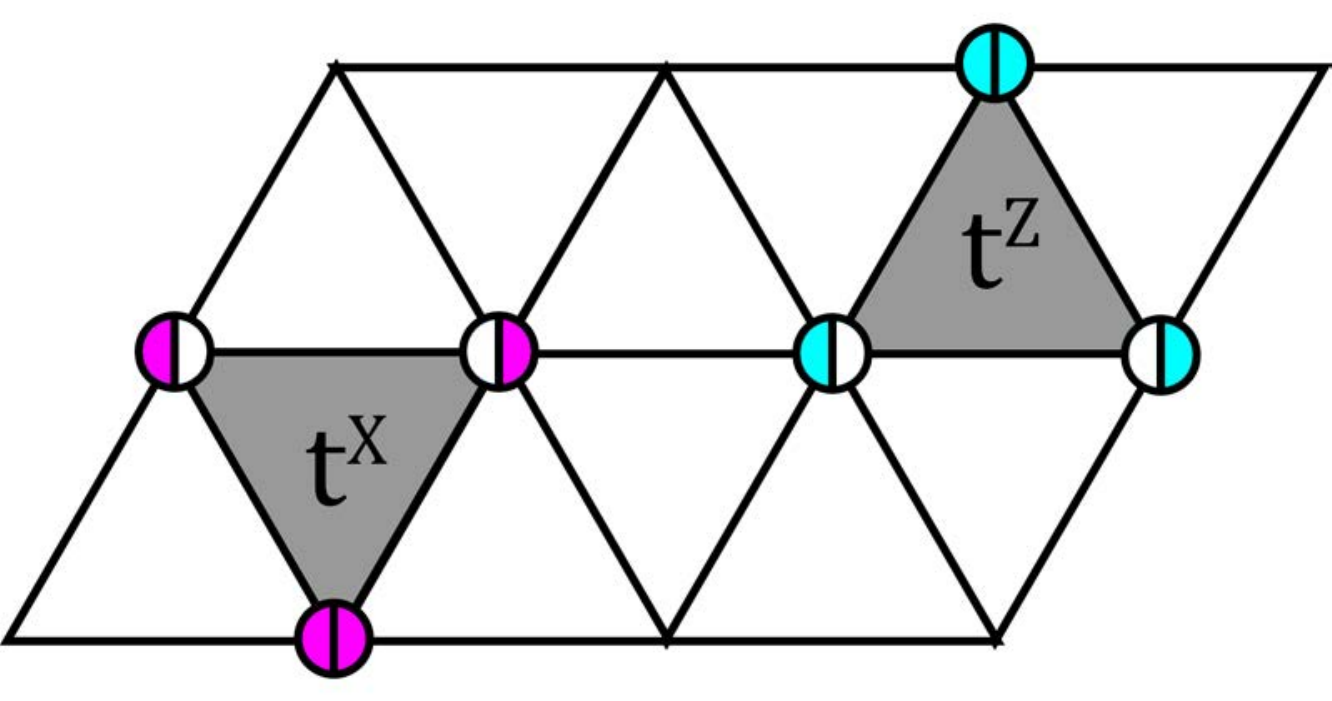}}
\end{tabular}

\caption{(a) Depiction of the stabilizer operators contained in $S_{c\hyph\text{hex}}$. (b) Definition of the stabilizer operators of a version of the toric code on the triangular lattice.}\label{fig:hexdef}

\end{figure}

Just as with the cluster-cube example, we form the three stacks of the quasi-cluster model which yields a Hilbert space of $N=6L^3$ qubits arranged on the cubic lattice such that there are six qubits associated with each vertex. We also divide each of these unit cells into two three qubit unit cells. The stabilizer set of the base model is
\begin{align}
S_{c \hyph \text{hex}} = \{h_c^{X}, t_c^{X}, h_c^{Z}, t_c^{Z} : c \text{ cubes }\},
\end{align}
where $h_c^{X}, t_c^{X}, h_c^{Z}, t_c^{Z}$ are defined in Fig.~\ref{fig:hexdefa}. We refer to $h_c^X, h_c^Z$ as hexagon operators, but it may be more appropriate to visualize them as two triangle operators similar to $t_c^X, t_c^Z$ stacked in the $[111]$ direction. We recognize that these are rather complicated operators and arguably more so that the target model. This is a consequence of the construction of Section \ref{sec:sretolre} which only guarantees the perturbation model is no ``less local'' than the target model. However, we are allowed a great deal of freedom, even removing the condition that it is a stabilizer Hamiltonian in order to simply the base model. We discuss this in Section \ref{sec:haahsimple}.

This model contains four layer constraints for every $[111]$ layer, one for each stabilizer type. It is equivalent to stacks of a variant of the $d=2$ toric code on the triangular lattice which is depicted in Fig. \ref{fig:tritc}.  This implies that $R_{c\hyph\text{hex}}=L$ and the model contains LRE.

\subsubsection{Perturbation Analysis}

\begin{figure}

\centering

\includegraphics[scale=.45]{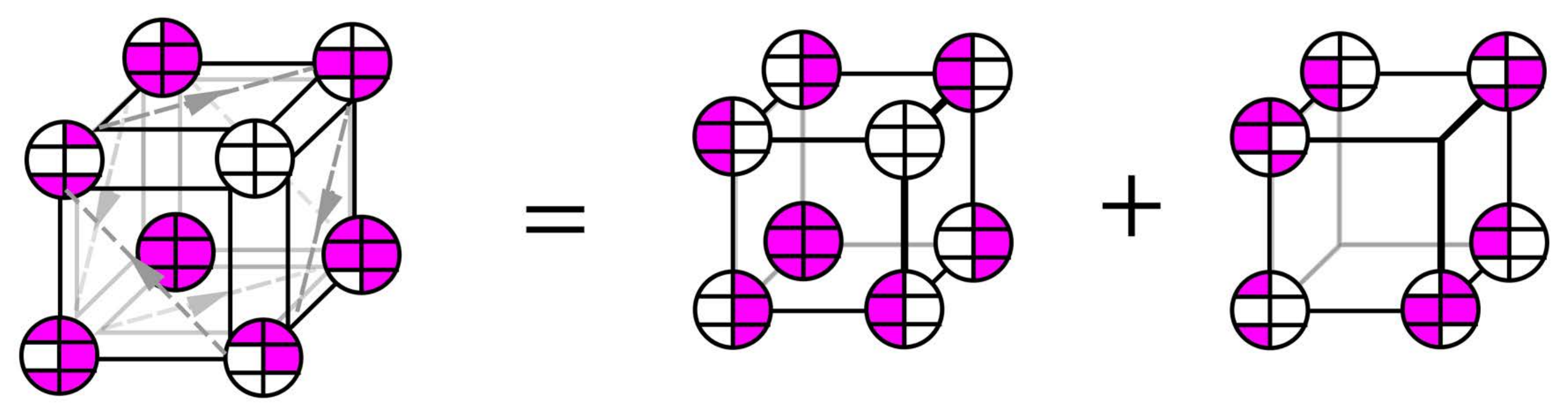}

\caption{Demonstration of how the product of $X$-type quasi-cluster model plaquettes forming a cube is broken down to a sum of effective degrees of freedom and members of $S_{c \hyph \text{hex}}$. A similar relation holds for the analogous $Z$-type operators. }\label{fig:haahbuild}

\end{figure}

To understand the structure of the stacking, we organize the layers so  the $x_v^1x_v^2$ right-angle base-points of $qc_v^X$ in each layer meets at the origin point of the elementary cube and the arrows in Fig. \ref{fig:qclusdef} meet tip-to-tail where stabilizers between different layers meet (see Fig. \ref{fig:haahbuild}). This maintains a three-fold rotation symmetry about the $[111]$ axis. As a consequence, the $z_v^1z_v^2$ right-angle base-points of $c_v^Z$ in each layer meet at the origin antipode of the elementary cube and the direction of the associated arrows is reversed. 

Much of the analysis from the cluster-cube example carries over. In particular, we take the product of the six quasi-cluster plaquettes which form a cube and rewrite this operator by separating  the $(xxx)^0, (zzz)^0$ degrees of freedom from the $(xxx)^j, (zzz)^j$ degrees of freedom as shown in Fig. \ref{fig:haahbuild}. From  this, one can see this product of quasi-cluster operators is equivalent to the product of the $g_c^X, g_c^Z$ operators using the $(xxx)^0, (zzz)^0$ degrees of freedom times ($+$ in $\mc P$) the $h_c^X, h_c^Z$ operators using the $(xxx)^j, (zzz)^j$ degrees of freedom. Since the base model operators are completely  written in the $(xxx)^j, (zzz)^j$ degrees of freedom, then this sixth order perturbation term survives if and only if $h_c^X, h_c^Z$ commutes with the entire base model, which by construction, it does. Furthermore, the strong perturbation is such that $h_c^X \simeq h_c^Z \simeq \id_{\mc H}$, implying these sixth order terms are equivalent to the terms of Haah's code. All that's left is to argue no lower order terms survive the perturbation. As before we consider the subspace of all terms which survive the perturbation $\left( \im \psi_{\text{parent}}\phi_{\text{base}}\right)^{\perp_\omega}$. Members of $ \im \psi_{\text{parent}}\phi_{\text{base}}$ are generated by the p-string configurations shown in Fig. \ref{fig:haahloops}. Due to the complexity of these p-strings, it is not immediately clear what other configurations, besides members of $\mc D^X \oplus \mc D^Z$, are orthogonal to these p-strings, but it is clear that they must be greater than sixth-order and scale with the size of the system. This is discussed in greater detail in Section \ref{sec:disthaah}.

\begin{figure}

\centering

\includegraphics[scale=.3]{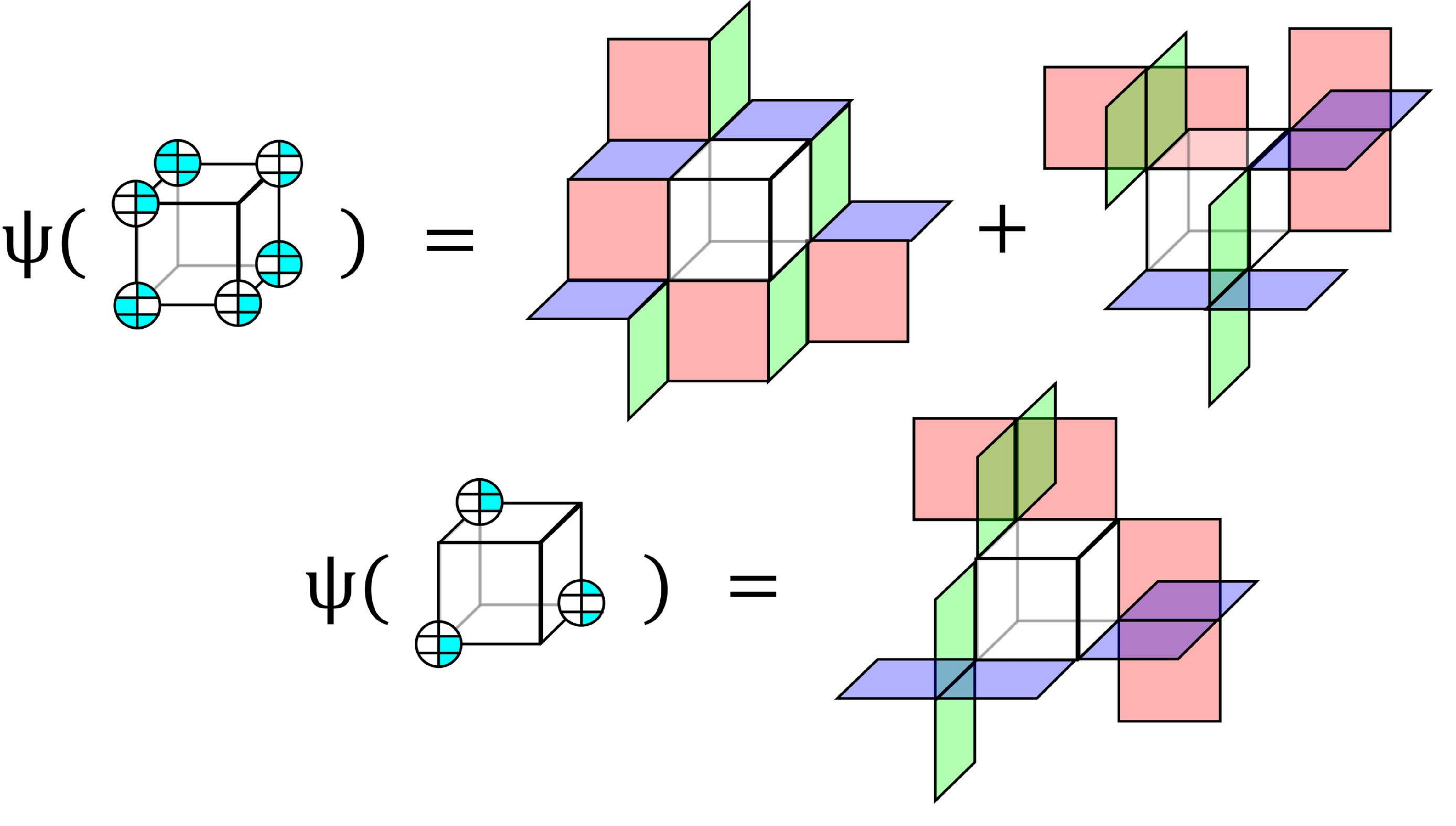}

\caption{X-type loop configurations generated by the Z-type members of $S_{c\hyph \text{hex}}$. The X-type members generate analogous loop configurations.}\label{fig:haahloops}

\end{figure}

\subsection{General Properties  of Distillation}\label{sec:genprop}

Already from the last section,  we see some of the distinctions between condensation and distillation. For the X-cube, specifically for the cube stabilizers, we implicitly defined effective local degrees of freedom given by $z_e^1 z_e^2$. The product of the plaquettes which form a cube can be written exactly in terms of these degrees of freedom with no additional factors. The same is true for the $d=3$ toric code where the product of three vertex terms can be written exactly in terms of the $x_e^1 x_e^2$ degrees of freedom (a similar explanation for the remaining terms of the X-cube and $d=3$ toric code have to wait unit Section~\ref{sec:condXC}). In contrast for the cluster-cube, the product of the cluster model operators forming a cube cannot, simultaneously for all such cubes, be written in terms of some local degrees of freedom such that the results are the cluster-cube. Instead once we choose the $(xxx)^0, (zzz)^0$ degrees of freedom, there is some ``junk'' left over in the form of the operators $t_c^X, t_c^Z \in S_{\triangle}$ which generates a group $\mc G_\triangle$ expressed only in the $(xxx)^i, (zzz)^i$ degrees of freedom. Members of this group mutually commute and so this group represents some degrees of freedom which must be fixed by the infinite-strength perturbation. In principle, one could use $S_\triangle$ as the base model, whereby one would find the terms corresponding to the cluster cube. Still, this does not fix an extensive number of degrees of freedom and we risk other low order terms surviving the perturbation (see Section \ref{sec:subex} for examples of such operators). So, we require the base model to be  a{ \it completion} stabilizer set,  $S_{\triangle} \to S_{c\hyph \triangle}$, such that all member are formed using the $(xxx)^i, (zzz)^i$ degrees of freedom (i.e. they automatically commute with the target model degrees of freedom), $\mc G_\triangle \subseteq \mc G_{c\hyph \triangle}$ and $\dim \mc G_{c\hyph \triangle}$ is of the order of the number of $(xxx)^i, (zzz)^i$ degrees of freedom. There is generally no unique choice for this completion, so we choose the completion such that the members of $S_{c\hyph \triangle}$ are of smaller support, i.e. ``more local'' than $ S_{\triangle}$, where this exists. For the cluster-cube, we can use the fact that $t_c^X$ and $t_c^Z$ commute trivially within the unit cells, allowing us to choose the base model presented in Section~\ref{sec:ccbase}. Finally, since all members $g \in \mc G_{c\hyph\triangle}$ commute with all $\phi(A)\in \phi[ \mc D^X \oplus \mc D^Z]$, we apply Eq.\eqref{eq:stabgauge} to find,
\begin{align}\label{eq:psicomp}
0 = \lambda(\phi_{\text{parent}}(A), g) = \omega(A, \psi_{\text{parent}}(g)),
\end{align}
This implies $\psi_{\text{parent}}(g)$ is always a p-string, regardless of the choice of completion, such that $\left( \im \psi_{\text{parent}}\phi_{c \hyph \triangle}]\right)^{\perp_\omega}\supseteq \mc \mc D^X \oplus \mc D^Z$, i.e. the set of terms which survives the perturbation includes all cube terms and we can generally expect any other surviving terms to be of the order $\sim L$ or greater, albeit this has to be checked. All the same general features carry over for Haah's code and thus these two examples represent a distillation process. Still there are significant differences accounting for the fact that one target model is Type-I and the other is Type-II. We explore there distinctions below.

\section{Role of  Gauge Substructures in the Theory of Condensation and Distillation} \label{sec:part2}

In the following section, we look to explore the similarity and differences between distillation and condensation in a more formal way. To do this, we use the notion of a {\it linear gauge structure} which includes all stabilizer codes. We further introduce here a powerful extension of this idea we deem a {\it gauge substructure} which is used to describe these layer constructions, though the idea is more general than producing 3D fracton models from layers of 2D models. We start by reviewing the notion of a linear gauge structure, then defining a gauge substructure. The layer constructions presented here are shown to be examples of gauge substructures, and we use this language to compare the relevant features. 

The next four sub-sections are considerably more technical, so the uninterested reader may wish to skip to Section \ref{sec:subex}, where the results are applied to the examples and the following four sub-section might be used only for reference. 

\subsection{Definition  of a Linear Gauge Structure}\label{sec:lgsdef}

Ideas and results of this section can be found in Ref. \cite{Schmitz2019a}. Consider two vector spaces, $\mc A$ and $\mc F$, both over some field $\mb F$. We refer to $\mc A$ as the {\it potential space} and $\mc F$ as the {\it field space}. Each is equipped with a  linear, non-degenerate form $\Omega: \mc A \to \mc A^*$ and $\Lambda: \mc F \to \mc F^*$ where in general $\mc V^*$ is the dual space of $\mc V$ or more formally 
\begin{align}
\mc V^* = \{ \left(f:\mc V \to \mb F\right) : f \text { is (anti-)linear and bounded}\}.
\end{align}
Non-degeneracy of $\Omega$ requires that for any $A \in \mc A$,  $\Omega (A) =0^*$ if and only if $A=0$, where $0^*$ is the zero map in $\mc A^*$. We require the same holds for $\Lambda$. Non-degeneracy implies that $\Omega$ and $\Lambda$ are  injective. If they are also surjective, then we say they are invertible.

Let $\phi :\mc A \to \mc F$ and $\psi: \mc F \to \mc A$ be some linear maps. We define a linear gauge structure as follows:

\begin{definition}\label{def:GS}
An $\mb F$-linear gauge structure $GS =\left(( \mc A, \Omega, \phi), (\mc F, \Lambda, \psi)\right)$ satisfies the following:
\begin{enumerate}
\item $\Lambda$ is invertible, and
\item \begin{align}\label{eq:def1}
\phi^\star \Lambda = \Omega \psi,
\end{align}
or the following diagram commutes \footnote{To understand this diagram, we recall that the pullback is an involution for which $\phi^\star: \mc F^* \to \mc A^*$ such that $f \mapsto f\phi$.}:
\begin{equation}
\begin{tikzcd}
\mc F \arrow[r, "\psi"] \arrow[d, "\Lambda"'] & \mc A \arrow[d, "\Omega"] \\
\mc F^* \arrow[r, "\phi^\star"'] & \mc A^*
\end{tikzcd}
\end{equation}
\end{enumerate}
\end{definition}
 We can put \eqref{eq:def1} into the form familiar from Eq. \eqref{eq:stabgauge} by defining $\omega(A,B) =\left( \Omega(B) \right)(A)$ and $\lambda(F, G) = \left(\Lambda(G)\right) (F)$ with which our definition can be written for all $A \in \mc A$ and $F\in \mc F$ as
\begin{align}\label{eq:def2}
\lambda(\phi(A), F) = \omega(A, \psi(F)).
\end{align}
Put this way, we recognize $\psi$ as a generalized adjoint of $\phi$ with respects to $\lambda$ and $\omega$, which we denote as $\psi^\dagger =\phi$. Also as a consequence of non-degeneracy, $\psi$ is unique to $\phi$, though existence is not guaranteed unless $\Omega$ is also invertible. In all cases considered here, $\Omega$ is invertible and the existence of $\psi$ is guaranteed and given by $\psi= \Omega^{-1} \phi^\star\Lambda$. The most important consequence of a gauge structure is the so-called {\it Braiding Law for Excitations} or BrLE rules given by:

\begin{theorem}[BrLE Rules] \label{th:fund}
For any $\mb F$-linear gauge structure \break $GS =\left(( \mc A, \Omega, \phi), (\mc F, \Lambda, \psi)\right)$,  
\begin{subequations}\label{eq:conss}
\begin{align}
(\ker \phi)^{\perp_\Omega} =& \im \psi, \label{eq:cons}\\
(\im \phi)^{\perp_\Lambda} =& \ker \psi \label{eq:consalt},
\end{align}
\end{subequations}
where $(\ker \phi)^{\perp_\Omega}$ is the subspace $\mc B$ for which $\ker \phi$ is the degenerate subspace of $\Omega \,\iota_{\mc B}$ and likewise for $(\im \phi)^{\perp_\Lambda}$.
\end{theorem}

From Eq.~\eqref{eq:stabgauge}, we see that all stabilizer codes fit the definition of an $\mb F_2$-linear gauge structure and Proposition \ref{prop:stabcon} is just Eq. \ref{eq:cons} as applied to the stabilizer code. 

All stabilizer codes also fit the definition of a {\it symplectic } gauge structure which satisfies 
\begin{align}
\phi^\star \Lambda \phi = 0^*,
\end{align}
or $\lambda(\phi(A), \phi(B))=0$ for all $A, B\in \mc A$. As a consequence, $\im \phi \subseteq \ker \psi$ which allows us to define the gauge homology as
\begin{align}
\mc H(\text{GS})= \ker \psi/\im \phi,
\end{align}
which we recognize as the logical subspace of the stabilizer code represented by GS. As such, we define a {\it trivial} symplectic gauge structure as one such that the following short sequence is exact:
\begin{equation}
\begin{tikzcd}
0 \arrow[r] & \mc A \arrow[r, "\phi"] & \mc F \arrow[r, "\psi"] & \mc A \arrow[r] &0.
\end{tikzcd}
\end{equation}
A trivial symplectic gauge structure is one such that $\mc H(\text{GS})$ as well as the BrLE rules are trivial.

\subsection{Definition of a Gauge Substructure}\label{sec:gs}

Suppose we are given a gauge structure GS$= \left(\left(\mc A, \Omega, \phi\right), \left(\mc F, \Lambda, \psi\right)\right)$ and we wish to restrict the field space to some subspace $\mc R \subseteq \mc F$. That is, we somehow only allow members of $\mc R$ to act on our system. Then our new syndrome map is given by $\psi_{\mc R} =\psi \iota_{\mc R}$, where $ \iota_{\mc R}$ is the inclusion map of $\mc R \to \mc F$. In order to maintain the gauge structure after this restriction, i.e. insure an analogous version of Eq.~\eqref{eq:def1}, we must modify $\phi$ to $\phi_{\mc R}= (\psi \iota_{\mc R})^\dagger= \iota_{\mc R}^\dagger \phi$. That is, we first apply $\phi$ and then apply this so-called adjoint restriction. Just as the inclusion map effectively modifies the domain of a function, the adjoint restriction must modify the co-domain. Thus $\iota_{\mc R}^\dagger : \mc F \to \left[\mc F\right]_{\mc R}$, where $\left[\mc F\right]_{\mc R}$ is the collection of equivalence classes defined for all $F \in \mc F$ as
\begin{align}
[F]_{\mc R}= \{ G\in \mc F: \lambda(G, R)+ \lambda(F, R)=0 \text{ for all } R\in \mc R\}.
\end{align}
It should be clear that $[\mc F]_{\mc R}= \mc F /(\mc R)^{\perp_\Lambda}$, so $\iota_{\mc R}^\dagger$ is a quotient map. We also take for granted that $[\mc F]_{\mc R} \simeq \mc R$. If $\Lambda$ is an inner product, this is obvious as there is a unique decomposition of the space as $\mc F = \mc R \oplus \mc R^{\perp_\Lambda}$. However other cases are not so clear. Appendix \ref{apx:F2vec} discusses this question for $\mb F_2$ following results from Ref. \cite{Bernhard2009}. When there does exist a linear bijective map $\alpha:[\mc F]_{\mc R} \to \mc R$, we note that $\alpha \iota_{\mc R}^\dagger:\mc F \to \mc R$ can be but is not always a projection map (which we denote with $\pi_{\mc R}: \mc F \to \mc R$) and is highly dependent on $\Lambda$ and $\mc R$. For a given $\mc R$, if there exists an $\alpha$ such that $\alpha\iota_{\mc R}^\dagger = \pi_{\mc R}$, then we shall refer to this as a {\it projective} restriction and always use such an $\alpha$ when it exists. If $\Lambda$ is an inner product, all $\mc R$ are projective which is not the case for a symplectic $\Lambda$. An example of a non-projective restriction for stabilizer codes (where $\Lambda$ is symplectic) is the case when $\mc R$ consists of mutually commuting operators. In that case, all $f \in \mc R \subseteq \mc R^{\perp_\Lambda}$  are members of the equivalence class of the identity (the zero element of the vector space), i.e. $f \in [\id_\mc H]$ and so all bijective linear maps $\alpha:[\mc F]_{\mc R} \to \mc R$ must be such that $\alpha\iota^\dagger_{\mc R}(f) \simeq 0$. Clearly, this is not projective. It is very important that all examples from Section~\ref{sec:pert} are projective, as is demonstrated below. The fact that we must modify $\phi \to \phi_{\mc R} =\alpha\iota_{\mc R}^\dagger \phi$  is a reflect of the fact that $\Lambda \iota_{\mc R}$ is no longer a non-degenerate form which was a requirement for a gauge structure. Thus we must modify this to $\Lambda_{\mc R} = \left( \alpha^{-1}\right)^\star \Lambda \iota_{\mc R}$ to recover non-degeneracy and invertibility.\footnote{We have committed an intentional notational mistake here as $\left( \alpha^{-1}\right)^\star \Lambda \iota_{\mc R}$ is technically meaningless as written. We actually mean to replace $\Lambda\iota_{\mc R}$ in this expression with $\tilde \Lambda \iota_{\mc R}: \mc R \to \left([\mc F]_\mc R\right)^*$ such that $\tilde{\lambda}\left([F], G \right)= \lambda(F,G)$, which is singled-valued and well-defined by the definition of $[\mc F]_\mc R$. We continue to use the overloaded notation for $\Lambda$ and $\Omega$ where the actual meaning should be obvious in the context.} 

As we have recovered all the conditions of Definition \ref{def:GS}, we define:
\begin{definition}
The gauge structure $\text{GS}_{\mc R}= \left(\left(\mc A, \Omega, \phi_{\mc R}\right), \left(\mc R, \Lambda_{\mc R}, \psi_{\mc R}\right)\right)$ is referred to as the \emph{ right gauge substructure of GS with respects to the field space restriction to $\mc R$} and the process as the \emph{right restriction of GS to $\mc R$}.
\end{definition}
Likewise, we can also restrict to a subspace $\mc L \subseteq \mc A$ via the same process such that ${}_{\mc L}\phi = \phi \iota_{\mc L}$, ${}_{\mc L}\psi=  \beta \iota_{\mc L}^\dagger \psi$ and  $\Omega_{\mc L} = \left(\beta^{-1}\right)^\star \Omega \iota_{\mc L}$, where $\beta: \mc A/(\mc L)^{\perp_\Omega} \to \mc L$ is a bijective linear map when it exists, and define:
\begin{definition}
 The gauge structure ${}_{\mc L} \text{GS}= \left(\left(\mc L, \Omega_{\mc L}, {}_{\mc L}\phi\right), \left(\mc F, \Lambda,{}_{\mc L}\psi\right)\right)$ is referred to as the \emph{left gauge substructure of GS with respects to the potential space restriction to $\mc L$} and the process as the \emph{left restriction of GS to $\mc L$}.
\end{definition}
Again, this is a projective restriction when $\beta\iota_{\mc L}^\dagger =\pi_{\mc L}$. In the case of $\mb F_2$, Appendix \ref{apx:F2vec}, Proposition \ref{app:proj} gives a necessary and sufficient condition for a projective left restriction which is that $\mc L \cap \mc L^{\perp_\Omega}$ is trivial, i.e. $(\iota_{\mc L})^\star\Omega\iota_{\mc L}$ is non-degenerate as a form in $\mc L$. The left restrictions for our examples are only projective for lattices such that $L$ is odd (see Appendix \ref{apx:F2vec}).

\subsection{Direct Sum of Gauge Structures and Their Relation to Substructures}\label{sec:dsum}

In the category of $\mb F$-linear gauge structures, we define isomorphic equivalence via:

\begin{definition}
 For two $\mb F$-linear gauge structures, $\text{GS}_1=\left(\left(\mc A_1, \Omega_1, \phi_1\right), \left(\mc F_1, \Lambda_1, \psi_1\right)\right)$ and $\text{GS}_2=\left(\left(\mc A_2, \Omega_2, \phi_2\right), \left(\mc F_2, \Lambda_2, \psi_2\right)\right)$, $\text{GS}_1$ is equivalent to $\text{GS}_2$ if and only if there exists linear isomorphisms $\alpha: \mc F_2 \to \mc F_1$ and $\beta:\mc A_2 \to \mc A_1$ such that 
 \begin{enumerate}
 \item $\Omega_2 = \beta^\star \Omega_1 \beta$,
 \item $\Lambda_2 =  \alpha^\star \Lambda_1 \alpha$ and
 \item $\phi_2 = \alpha^{-1} \phi_1 \beta$.
 \end{enumerate}
\end{definition}

As for $\psi_1$ and $\psi_2$, we automatically get
\begin{align}
\phi_2^\star \Lambda_2=  \left(\alpha^{-1} \phi_1 \beta\right)^\star  \alpha^\star \Lambda_1 \alpha= \beta^\star \phi_1^\star\Lambda_1 \alpha= \beta^\star \Omega_1 \psi_1 \alpha= \beta^* \Omega \beta \left(\beta^{-1} \psi_1 \alpha\right)= \Omega_2\left(\beta^{-1} \psi_1 \alpha\right).
\end{align}
So by uniqueness, $\psi_2 = \beta^{-1} \psi_1 \alpha$.

Also in the category of $\mb F$-linear gauge structures, we can define a direct sum:

\begin{definition}
Let GS$= \left(\left(\mc A, \Omega, \phi\right), \left(\mc F, \Lambda, \psi\right)\right)$ and $\text{GS}_{1(2)}= \left(\left(\mc A_{1(2)}, \Omega_{1(2)}, \phi_{1(2)}\right), \left(\mc F_{1(2)}, \Lambda_{1(2)}, \psi_{1(2)}\right)\right)$ be $\mb F$-linear gauge structures. Then GS is the \emph{direct sum} of $\text{GS}_1$ and $\text{GS}_2$, written as $\text{GS} = \text{GS}_1 \oplus \text{GS}_2$, if and only if the following are true:
\begin{enumerate}
\item $\mc A = \mc A_1 \oplus \mc A_2$,
\item $\mc F = \mc F_1 \oplus \mc F_2$,
\item $\Omega = \Omega_1 \oplus \Omega_2$,
\item $\Lambda = \Lambda_1 \oplus \Lambda_2$ and 
\item $\phi = \phi_1 \oplus \phi_2$.
\end{enumerate} 

Furthermore, $\text{GS}_1$ is said to be a \emph{divisor} of GS if and only if there exists a $\text{GS}_2$ such that $\text{GS} \simeq \text{GS}_1 \oplus \text{GS}_2$, in which case we can write $\text{GS}_1 \simeq \text{GS}/\text{GS}_2$. 
\end{definition}
 
 Note as a consequence of this definition, one automatically gets $\psi = \psi_1 \oplus \psi_2$ and that $\iota_{\mc A_2}^\star \Omega \iota_{\mc A_1} =0^*$ and $\iota_{\mc F_2}^\star \Lambda \iota_{\mc F_1} =0^*$.
 
 We now prove some important results.
 
 \begin{proposition}\label{prop:symp}
For $\text{GS} \simeq \text{GS}_1 \oplus \text{GS}_2$, if any two of the gauge structures are symplectic, then the third is also symplectic.
 \end{proposition}

\begin{proof}
Let GS$= \left(\left(\mc A, \Omega, \phi\right), \left(\mc F, \Lambda, \psi\right)\right)$ and $\text{GS}_{1(2)}= \left(\left(\mc A_{1(2)}, \Omega_{1(2)}, \phi_{1(2)}\right), \left(\mc F_{1(2)}, \Lambda_{1(2)}, \psi_{1(2)}\right)\right)$ be $\mb F$-linear gauge structures such that $\text{GS} \simeq \text{GS}_1 \oplus \text{GS}_2$. Suppose $\text{GS}_1$ and $\text{GS}_2$ are symplectic. Then clearly by the definition of the direct sum, GS is symplectic. 

Suppose GS and one of the other two (w.l.o.g. $\text{GS}_1$) are symplectic. This implies
\begin{align}
0^*= \phi^\star \Lambda \phi \simeq \phi_1^\star \Lambda_1 \phi_1 \oplus \phi_2^\star \Lambda_2 \phi_2 \simeq\phi_2^\star \Lambda_2 \phi_2.
\end{align}
Therefore, $\text{GS}_2$ is symplectic. 
\end{proof}

We now come to the defining characteristic which differentiates a restriction being characterized as a condensation or a distillation:

\begin{definition}\label{def:cond}
Let GS$=\left(\left(\mc A, \Omega, \phi\right), \left(\mc F, \Lambda, \psi\right)\right)$ be a symplectic $\mb F$-linear gauge structure and $\mc L\subseteq \mc A$ and $\mc R\subseteq \mc F$ be subspaces. We refer to the restriction ${}_{\mc L}\text{GS}_{\mc R}$ as a \emph{condensation} if and only if ${}_{\mc L}\text{GS}_{\mc R}$ is symplectic and a divisor of GS, i.e. there exists a symplectic $\mb F$-linear gauge structure $\text{GS}_2$ such that ${}_{\mc L}\text{GS}_{\mc R}\simeq \text{GS} / \text{GS}_2$. If ${}_{\mc L}\text{GS}_{\mc R}$ is symplectic and not a divisor of GS, then we refer to the restriction as a \emph{distillation}.
\end{definition}

This says that a restriction of GS from $\mc F \to \mc R$ and $\mc A \to \mc L$ is identical to modding out or ``condensing'' some parts of GS, namely $\text{GS}_2$ which we argue must be $\text{GS}_2 \simeq{}_{\mc L^{\perp_\Omega}}\text{GS}_{\mc R^{\perp_\Lambda}}$. We can use this and Proposition~\ref{prop:symp} to eliminate the possibility that a restriction is a condensation by showing ${}_{\mc L^{\perp_\Omega}}\text{GS}_{\mc R^{\perp_\Lambda}}$ is not symplectic. However, this is only necessary, not sufficient to show a restriction is a condensation. We can give a necessary and sufficient  criterion in the following proposition:

\begin{proposition}\label{subad}
For GS$=\left(\left(\mc A, \Omega, \phi\right), \left(\mc F, \Lambda, \psi\right)\right)$, $\mc L \subseteq \mc A$ and $\mc R \subseteq \mc F$, ${}_{\mc L}\!\text{GS}_{\mc R}$ is a divisor of GS if and only if $\phi[\mc L]$ is isomorphic to a subspace of $\mc R$. 
\end{proposition}

\begin{proof}
Suppose ${}_{\mc L}\text{GS}_{\mc R}$ is a divisor of GS. Then by definition, there exists a \break $\text{GS}_{2}= \left(\left(\mc A_2, \Omega_{2}, \phi_{2}\right), \left(\mc F_2, \Lambda_{2}, \psi_{2}\right)\right)$ such that $\text{GS} \simeq {}_{\mc L}\text{GS}_{\mc R} \oplus \text{GS}_2$. Let $\alpha : \mc F \to  \mc R \oplus \mc F_2$ and $\beta: \mc A \to \mc L \oplus \mc A_2$ be the maps which define equality. As $\mc L \subseteq \mc A\simeq \mc L \oplus \mc A/\mc L$, we can assume w.l.o.g. that $\beta\simeq \id_{\mc L} \oplus \beta'$ for some linear bijective $\beta': \mc A/ \mc L \to \mc A_2$. So $\phi[\mc L] = \alpha^{-1}({}_{\mc L} \phi_{\mc R} \oplus \phi_2) \beta [\mc L] \simeq \alpha^{-1}\left[\im {}_{\mc L} \phi_{\mc R}\right] \subseteq \alpha^{-1}[\mc R] \simeq \mc R$.

Now suppose that $\phi[\mc L] \subseteq \alpha[\mc R]$ for some linear isomorphism $\alpha: \mc R \oplus \mc F_2 \to \mc F$ and some appropriate space $\mc F_2$. Trivially, one has that $\mc A \simeq \mc L \oplus \mc A/\mc L$. Furthermore, since $\mc L \simeq \mc A/ \mc L^{\perp_\Omega}$, then $\mc A \simeq \mc L^{\perp_\Omega} \oplus \mc A/\mc L^{\perp_\Omega} \simeq \mc L^{\perp_\Omega} \oplus \mc L$. This implies $\mc A/ \mc L \simeq \mc L^{\perp_\Omega}$. Likewise, we conclude that $\mc F/\alpha[\mc R] \simeq \mc F/\mc R \simeq \mc R^{\perp_\Lambda}$. Thus we define $\Omega_2 \simeq \iota_{\mc L^{\perp_\Omega}}^\star \Omega \iota_{\mc L^{\perp_\Omega}}$ and $\Lambda_2 \simeq \iota_{\mc R^{\perp_\Lambda}}^\star \Lambda \iota_{\mc R^{\perp_\Lambda}}$. Therefore, we have $\Omega \simeq \Omega_{\mc L} \oplus \Omega_2$ and $\Lambda \simeq \Lambda_{\mc R} \oplus \Lambda_2$. Note that the above equivalences implicitly define isomorphisms used to map $ \mc L \oplus \mc A/\mc L \to \mc L \oplus \mc L^{\perp_{\Omega}}$ and $\mc R \oplus \mc F/\alpha[\mc R] \to \mc R \oplus \mc R^{\perp_\Lambda}$. Thus if we construct a map $\tilde \phi_2: \mc A/ \mc L \to \mc F/ \alpha[\mc R]$, we define $\phi_2$ such that $\tilde \phi_2 \to \phi_2$ under these isomorphisms. For all $A,B \in \mc A$ and using the appropriate equivalence classed for  $\mc A/ \mc L$ and $\mc F/ \alpha[\mc R]$, note that if $[A] =[B]$ then $B \in [A]$, which implies $\phi(B) \in \phi\left[[A]\right]= \phi(A) + \phi[\mc L] \subseteq \phi(A) +\alpha[\mc R] = [\phi(A)]$ where we have used our hypothesis. Thus $[A] =[B]$ implies $[\phi(A)] = [\phi(B)]$. So we can define $\tilde \phi_2: \mc A/\mc L \to \mc F/\alpha[\mc R]$ as $[A] \mapsto [\phi(A)]$, which implies $\phi \simeq {}_{\mc L}\phi_\mc R \oplus \phi_2$. Therefore if we define $\text{GS}_{2}= \left(\left(\mc L^{\perp_{\Omega}}, \Omega_{2}, \phi_{2}\right), \left(\mc R^{\perp_\Lambda}, \Lambda_{2}, \psi_{2}\right)\right)$, we have that $\text{GS} \simeq {}_{\mc L}\text{GS}_{\mc R} \oplus \text{GS}_2$, which is to say ${}_{\mc L}\text{GS}_{\mc R}$ is a divisor of GS.
\end{proof}

From the proof, one gets the corollary,
\begin{corollary}\label{cor}
If ${}_{\mc L}\text{GS}_{\mc R}$ is a divisor of GS, then $\text{GS}/{}_{\mc L}\!\text{GS}_{\mc R} \simeq {}_{\mc L^{\perp_\Omega}} \! \text{GS}_{\mc R^{\perp_\Lambda}}$ and is unique up to an equivalence transformation.
\end{corollary}
$\text{GS} \simeq {}_{\mc L}\!\text{GS}_{\mc R} \oplus {}_{\mc L^{\perp_\Omega}} \! \text{GS}_{\mc R^{\perp_\Lambda}}$ is a statement that one can get ${}_{\mc L}\!\text{GS}_{\mc R}$ by starting with GS and then condensing ${}_{\mc L^{\perp_\Omega}} \! \text{GS}_{\mc R^{\perp_\Lambda}}$ in the usual sense of the term. That is, $\mc R^{\perp_\Omega} = [\id]_{\mc R}$ i.e. we are forcing no excitations of ${}_{\mc L^{\perp_\Omega}} \! \text{GS}_{\mc R^{\perp_\Lambda}}$ to form by forcing any operator which generates them to effectively act as the identity. The condition for condensation is just the statement that there a consistent sense in which we achieve the condensation of some excitations by suppressing some set of physical operators/ degrees of freedom. On the other hand for distillation as we will see, the excitations we ``distill'' are not uniquely associated to the operators which we suppress.

When ${}_{\mc L} \text{GS}_{\mc R}$ is a divisor of $\text{GS}$, one trivially has that 
\begin{proposition}\label{inher}
For GS$=\left(\left(\mc A, \Omega, \phi\right), \left(\mc F, \Lambda, \psi\right)\right)$, $\mc L \subseteq \mc A$ and $\mc R \subseteq \mc F$, if GS and ${}_{\mc L}\text{GS}_{\mc R}$ are symplectic and ${}_{\mc L}\text{GS}_{\mc R}$ is a divisor of GS, then %
\begin{align}
\mc H(\text{GS}) \simeq \mc H({}_{\mc L}\!\text{GS}_{\mc R}) \oplus \mc H( {}_{\mc L^{\perp_\Omega}} \! \text{GS}_{\mc R^{\perp_\Lambda}})
\end{align}
\end{proposition}

In particular, the dimension of $\mc H(\text{GS})$ must be equal to the sum of the dimensions of $\mc H({}_{\mc L}\!\text{GS}_{\mc R})$ and $\mc H( {}_{\mc L^{\perp_\Omega}} \! \text{GS}_{\mc R^{\perp_\Lambda}})$. So if we find that $\dim \mc H(\text{GS})< \dim \mc H({}_{\mc L} \text{GS}_{\mc R})$, we can conclude that ${}_{\mc L}\text{GS}_{\mc R}$ is not a divisor of $\text{GS}$. So by this alone, we see that the cluster-cube and Haah's code examples must be a distillation as we ``distill'' LRE from SRE.

\subsection{General Properties of Distilling LRE from SRE}\label{sec:sretolre}

In this section, we tabulate some of the general features of distilling LRE from SRE (SRE $\to$ LRE). SRE $\to$ LRE is given by the following scenario: we start with some trivial symplectic gauge structure, GS, which we again refer to as the parent model, where trivial is as defined in Section~\ref{sec:gs}. We then consider a ``local'' right restriction. By local, we mean that we can partition the qubits (or more  generally, local degrees of freedom) into unit cells, and we restrict the local Pauli spaces in each unit cell to a subspace such that $\mc R = \bigoplus_i \mc R_i \subseteq \mc F = \bigoplus_i \mc F_i$. We assume that this right restriction is projective. We then look to find a left restriction $\mc L$ which recovers the symplectic structure of the resulting gauge substructure, which we refer to as the target gauge structure. Again, we assume this is a projective left restriction though this assumption is relaxed in our examples for even $L$. We also want this left restriction to be as ``big as possible'' in some sense. So we require the left restriction to be {\it maximal} according to the following definition:
\begin{definition}\label{def:max}
A left restriction $\mc L$ is maximal relative to a right restriction $\mc R$ if and only if for all $A \in \mc A$, if $\phi(A) \in \mc R$, then $A \in \mc L$.
\end{definition}
This requirement prevents ``trivial'' logical operators. That is as the right restriction is projective, ${}_{\mc L}\phi_{\mc R}= \pi_{\mc R} \phi \iota_{\mc L}$. So for all $B\in \mc L$ and $A \in \mc A$, if $\phi(A)\in \mc R$, we have 
\begin{align}
\lambda(\phi(A), {}_{\mc L} \phi_{\mc R} (B)) = \lambda( \phi(A), \pi_{\mc R} \phi(B)) = \lambda(\phi(A), \phi(B)) = 0,
\end{align}
 where we have used $\lambda(\phi(A), \pi_{\mc R^{\perp_\Lambda}} \phi(B))=0$ since $\phi(A) \in \mc R$, a trick we use repeatedly below. So if $A \notin \mc L$, then $A$ trivially represents a logical operator and we do not want this to be the source of our LRE. For our examples, we must relax this condition a bit, but as we discuss in Appendix \ref{apx:F2vec}, this does not change any of the results of this section when applied to those examples.  

Note there is no guarantee this process results in LRE, i.e. a non-trivial gauge structure, though as we've seen, this is the case for our examples. As stated above, if we achieve LRE, the process is necessarily a distillation and not a condensation.

The general form of the right restrictions we consider here is the {\it canonical local restriction} (CLR) defined for a unit cell of $n$ qubits as

\begin{align}
\mc R^n_{\text{CLR}}= \{ \id_{\mb C_2^{\otimes n}},\, \prod_i^n x_i,\, \prod_i^n y_i,\, \prod_i^n z_i\}.
\end{align}
The equivalence classes of $[\mc F]_{\mc R_{\text{CLR}}}$ have the convenient general form
\begin{align} \label{eq:CLReqcl}
[P]= \{ \prod_i^n q_i: \sum_i \text{type}(q_i)= P\},
\end{align}
where $P=\{\id, X,Y,Z\}$ is a  Pauli type,  $q_i \in \{ \id_i, x_i, y_i, z_i\}$ for qubit $i$ and the sum over type is a representation of the single qubit algebra. To see this is the case, consider that $[\id]$ is the set of all Pauli operators which commute with $\mc R_{\text{CLR}}$. If one considers $\lambda (p, \prod_i^n x_i)$ for any n-qubit Pauli $p$, this is the same as the modulo 2 Hamming weight of the $\{0,1\}^ n$ representation of all Z-type operators making up $p$. Likewise is true for $\lambda (p, \prod_i^n z_i)$ in terms of the X-type operators making up $p$. Thus a Pauli is in $[\id]$ if and only if its total Hamming weight is $0 \mod 2$ which is equivalent to the sum of its Pauli-types being the identity. The other equivalence classes follow immediately. CLR is projective only if $n$ is odd, in which case, $\mc R^n_{\text{CLR}}$ defines a single effective qubit. For $n=3$, we recognize $\mc R^3_{\text{CLR}}$ as being associated with the $(xxx)^0, (zzz)^0$ degree of freedom.

The primary question we wish to answer is what is the origin of the LRE and non-triviality of the target in terms of the parent? That means, we have to describe both the non-triviality of $\ker {}_{\mc L}\phi_{\mc R}$ and the gauge homology $\mc H\left( {}_{\mc L}\text{GS}_{\mc R}\right)$ via $\ker {}_{\mc L}\psi_{\mc R}$.  Both kernels are generally characterized by the following:
\begin{proposition}\label{sec:genker}
For SRE$\to$LRE,
\begin{subequations}
\begin{align}
\ker {}_{\mc L}\phi_{\mc R} = \phi^{-1} \left[\mc R^{\perp_\Lambda} \right]\cap \mc L, \label{eq:ker1} \\
\ker {}_{\mc L} \psi_{\mc R} = \psi^{-1}\left[\mc L^{\perp_\Omega}\right ]\cap \mc R. \label{eq:ker2}
\end{align}
\end{subequations}
\end{proposition}
\begin{proof}
($\subseteq \text{(a)}$) Let $A \in \ker {}_{\mc L}\phi_{\mc R} \subseteq \mc L$. Then $0= {}_{\mc L}\phi_{\mc R}(A)= \pi_{\mc R} \phi(A)$. This implies $\phi(A) \in \mc R^{\perp_\Lambda}$ and $A \in  \phi^{-1} \left[\mc R^{\perp_\Lambda} \right]\cap \mc L$.

($\supseteq  \text{(a)}$) Now let $A \in  \phi^{-1} \left[\mc R^{\perp_\Lambda} \right]\cap \mc L$. This implies there exists an $F \in \mc R^{\perp_\Lambda}$ such that $\phi(A)= F$. Applying the projector onto $\mc R$ we have that ${}_{\mc L}\phi_{\mc R}(A) =\pi_{\mc R} \phi(A) =\pi_{\mc R}(F) =0$. Therefore $A \in \ker {}_{\mc L}\phi_{\mc R}$.

($\subseteq \text{(b)}$) Let $F \in \ker {}_{\mc L} \psi_{\mc R}\subseteq \mc R$. By the BrLE  rules, this implies that for all $A \in \mc L$, 
\begin{align}\label{eq:temp1}
0=\lambda( {}_{\mc L}\phi_{\mc R}(A), F)= \lambda( \pi_{\mc R}\phi(A), F) =  \lambda( \phi(A), F)= \omega(A, \psi(F)).
\end{align}
Thus, $\psi(F) \in \mc L^{\perp_\Omega}$, and $F \in \psi^{-1}\left[\mc L^{\perp_\Omega}\right ]\cap \mc R$.

($\supseteq \text{(b)}$) Let $F \in  \psi^{-1}\left[\mc L^{\perp_\Omega}\right ]\cap \mc R$. This implies $\psi(F) \in \mc L^{\perp_\Omega}$. One then follows Eq.\eqref{eq:temp1} backward and concludes that $F\in \ker {}_{\mc L} \psi_{\mc R}$ by the BrLE rules.
\end{proof}
For Eq. \eqref{eq:ker1}, constraints in the target model are a consequence of members of $\mc L$ which land in $\mc R^{\perp_\Lambda}$. Now, this result only relies on the right restriction being projective, however, the triviality of the parent model implies for all $A \in \ker {}_{\mc L}\phi_{\mc R}$, $\phi(A)=0$ if and only if $A=0$. So all non-trivial constraint in the target model must come from parts $\mc L$ which map into non-trivial members of $\mc R^{\perp_{\Lambda}}$ (see examples below). So fractonic behavior is a consequence of a sub-extensive number of members of $\mc L$ landing in $\mc R^{\perp_\Lambda}$.

Eq. \eqref{eq:ker2} confirms that all members of $\ker {}_{\mc L}\psi_{\mc R}$ generate excitations in $\mc L^{\perp_\Omega}$, i.e. generate p-strings in our examples. However, Eq. \eqref{eq:ker2} does not satisfactorily address the primary question of SRE$\to$LRE which is the source of the LRE, i.e. the non-trivial gauge homology, in terms of the underlying SRE model. To partially address this, we define the projector $\pi_{\mc R^{\perp_\Lambda}} =\id_{\mc F} - \pi_{\mc R}$. We then note that for all $A,B \in \mc A$,
\begin{align}
0 = \lambda(\phi(A), \phi(B))= \lambda(\pi_{\mc R}\phi(A), \pi_{\mc R} \phi(B)) +  \lambda(\pi_{\mc R^{\perp_{\Lambda}}}\phi(A), \pi_{\mc R^{\perp_\Lambda}} \phi(B)),
\end{align}
which implies
\begin{align}\label{eq:stiltojunk}
\lambda(\pi_{\mc R}\phi(A), \pi_{\mc R} \phi(B)) =   \lambda(\pi_{\mc R^{\perp_{\Lambda}}}\phi(A), \pi_{\mc R^{\perp_\Lambda}} \phi(B)).
\end{align}
When we restrict this relation to $A, B \in \mc L$ where by construction, the left-hand side is zero, we recognize this as the statement from the perturbation analysis that the ``junk'' left over after we separate out the $\mc R$ degrees of freedom forms its own stabilizer group. Thus, we define the junk group (or subspace) as 
\begin{align}
\mc G_{\text{junk}}= \im \pi_{\mc R^{\perp_\Lambda}} \phi \iota_{\mc L},
\end{align}
and $\dim \mc G_{\text{junk}}$ as the {\it distillation cost}. This represents the number of degrees of freedom that we necessarily throw away in order to achieve the distillation of LRE. All members of $\mc L$ which map into $\mc R$ decreases the distillation cost, so if the restriction results in a condensation, then $\mc G_{\text{junk}}$ is trivial and there is no distillation cost. Furthermore, every constraint of the target model increases the size of $\mc G_{\text{junk}}$, so the more fractonic behavior and LRE, the larger the distillation cost. Then a completion of the junk group, $\mc G_{c\hyph\text{junk}}$, is any mutually commuting subset of $\mc G_{\text{junk}}^{\perp_\Lambda}/\mc R$ in $\mc R^{\perp_\lambda}$. \footnote{we actually mean the perpendicular compliment with respects to $\Lambda_{\mc R^{\perp_\Lambda}}$, i.e. $\mc G_{\text{junk}}^{\perp_{\Lambda_{\mc R^{\perp_\Lambda}}}} \simeq \mc G_{\text{junk}}^{\perp_\Lambda}/\mc R$. We use the second notation for readability.}

An important subspace of $\mc A$ for SRE$\to$LRE is $\mc A_{\text{distill}}=\psi\left[\im{}_{\mc L} \phi_{\mc R}\right] \subseteq \mc L^{\perp_\Omega}$, which by Eq. \eqref{eq:stiltojunk} is also $\mc A_{\text{distill}} = \psi[\mc G_{\text{junk}}]$. This set allows us to partially characterize the target gauge homology via the following:
\begin{proposition}\label{prop:LREsource}
For SRE$\to$LRE,
\begin{align}
 \ker {}_{\mc L} \psi_{\mc R}\supseteq \pi_{\mc R}\phi\left[ \left(\mc A_{\text{\emph{distill}}}\right)^{\perp_\Omega}\right].
\end{align}
Furthermore, $F \in \pi_{\mc R}\phi\left[ \left(\mc A_{\text{\emph{distill}}}\right)^{\perp_\Omega}\right]$ if and only if there exists an $F^\perp \in \mc R^{\perp_\Lambda}$ such that $F+ F^\perp \in \ker \psi$ or equivalently $\psi(F)= \psi(F^\perp)$.
\end{proposition}
\begin{proof}
($\supseteq$) Let $F \in \pi_{\mc R}\phi\left[ \left(\mc A_{\text{distill}}\right)^{\perp_\Omega}\right]$ which implies there exists a  $D \in \left(\mc A_{\text{distill}}\right)^{\perp_\Omega}$ such that $\pi_{\mc R} \phi(D) = F$. Let $A \in \mc L$ and consider
\begin{align}
\lambda\left(F, {}_{\mc L} \phi_{\mc R}(A))\right) =& \lambda\left(\pi_{\mc R}\phi(D), {}_{\mc L} \phi_{\mc R}(A))\right) = \lambda\left(\pi_{\mc R}\phi(D), \pi_{\mc R} \phi(A))\right)\nonumber \\
=&  \lambda\left(\phi(D), \pi_{\mc R} \phi(A))\right) = \omega\left(D,  \psi \pi_{\mc R} \phi(A)\right) \nonumber \\
=& 0,
\end{align}
where in the second line we use $\phi(D)= \pi_{\mc R}(D) + \pi_{\mc R^{\perp_\Lambda}}(D)$, and the third line is a consequence of the hypothesis for $D$ and the fact that $ \psi\pi_{\mc R} \phi(A) \in \mc A_{\text{distill}}$. Therefore, $F \in \ker {}_{\mc L} \psi_{\mc R}$.

As for the second claim, Let $F \in \ker {}_{\mc L} \psi_{\mc R} \subseteq \mc R$. Now suppose there exists an $F^\perp \in \mc R^{\perp_\Lambda}$, such that $F+F^\perp \in \ker \psi$. Since the parent model is trivial, this implies there exists a $D \in \mc A$ such that $F+F^\perp = \phi(D)$ and $\pi_{\mc R} \phi(D) =F$. Now let $B \in \mc A_{\text{distill}}$ which implies there exists a $G \in \im {}_{\mc L}\phi_{\mc R}$ such that $\psi(G)=B$. Furthermore, there exists an $A \in \mc L$ such that $G =\pi_{\mc R}\phi(A)$. Now consider  
\begin{align}
\omega(D, B)=& \omega(D, \psi(G)) = \lambda( \phi(D), G)= \lambda(\pi_{\mc R}\phi(D),  \pi_{\mc R}\phi(A)) =\lambda( F, {}_{\mc L}\phi_{\mc R}(A)) \nonumber \\
 =& 0,
\end{align}
where we have again used  Eqs.~\eqref{eq:def2} and \eqref{eq:stiltojunk}, $\phi(D)= \pi_{\mc R}(D) + \pi_{\mc R^{\perp_\Lambda}}(D)$,  and the hypothesis for $F$. As this is true for all such $B$, this implies $D \in \left(\mc A_{\text{distill}}\right)^{\perp_\Omega}$ and thus  $F= \pi_{\mc R} \phi(D) \in \pi_{\mc R}\phi\left[ \left(\mc A_{\text{distill}}\right)^{\perp_\Omega}\right]$. As the converse is trivially true by the BrLE rules, we therefore have that $F \in \ker {}_{\mc L} \psi_{\mc R}\cap \pi_{\mc R}\phi\left[ \left(\mc A_{\text{distill}}\right)^{\perp_\Omega}\right]=  \pi_{\mc R}\phi\left[ \left(\mc A_{\text{distill}}\right)^{\perp_\Omega}\right]$ if and only if there exists an $F^\perp \in \mc R^{\perp_\Lambda}$ such that $F+ F^\perp \in \ker \psi$ which is trivially equivalent to $\psi(F)=\psi(F^\perp)$.
\end{proof}

As $\mc A_{\text{distill}}\subseteq \mc L^{\perp_\Omega}$, then $\mc L \subseteq \left(\mc A_{\text{distill}}\right)^{\perp_\Omega}$. We can see that any member $A_\ell \in  \left(\mc A_{\text{distill}}\right)^{\perp_\Omega}/\mc L$ maps onto a logical operator under the action of $\pi_{\mc R} \phi$ (where it may map onto the trivial logical operator if it maps to zero). For our examples, such operators correspond to subsystem symmetries. In fact, this suffices as a definition of a subsystem symmetry for our purposes, and so we refer to such operators as {\it distilled SS logical operators}.  Note that for all $D\in \mc A$ such that $ \phi(D) \in \mc R^{\perp_\Lambda}$, the symplectic nature of the parent implies that $\phi(D) \in \mc G_{\text{junk}}^{\perp_\Lambda}/ \mc R$ as well as $D \in \psi[\mc G_{\text{junk}}]^{\perp_\Omega}=\left(\mc A_{\text{distill}}\right)^{\perp_\Omega} $. When $D\in \mc L$, again these represent constraints of the target model. When $D\notin \mc L$, we refer to these as {\it null distilled SS configurations}.

Proposition \ref{prop:LREsource} does not cover all logical operators  as some are {\it accidental} in the sense that they do not descend from some subsystem symmetries. Such an accidental logical operator corresponds an operator in $\mc R$ which maps onto a member $\mc L^{\perp_\Omega}$ under $\psi$ whose pre-image contains no member of $\mc R^{\perp_\Lambda}$. More concretely, 
\begin{definition}
$F \in \ker {}_{\mc L}\psi_{\mc R}$ is \emph{accidental} if and only if $\psi(F) \notin \psi[\mc G_{\text{\emph{junk}}}^{\perp_\Lambda}/\mc R]$. \footnote{Using $\mc G_{\text{junk}}^{\perp_\Lambda}/\mc R$ instead of all $\mc R^{\perp_\Lambda}$ is no restriction as for all $F \in \ker {}_{\mc L}\psi_{\mc R}$, $\psi(F) \in \psi[\mc G_{\text{junk}}^{\perp_\Lambda}/\mc R]$ if and only if  $\psi(F) \in \psi[\mc R^{\perp_\Lambda}]$.}
\end{definition}
 That is for the case of our examples, these operators in $\mc R$ just happen to generate p-strings which no operator in $\mc G_{\text{junk}}^{\perp_\Lambda}/ \mc R$ also generates.

We show below that both accidental and distilled SS logical operators can be found in our examples. However, we can also show that whenever we have an accidental logical operator, it anti-commutes with at least one distilled SS operator. This implies two things: the first is that for SRE$\to$LRE to distill any LRE, it must distill some subsystem symmetries. Second, the ground space of the target is spanned by the simultaneous eigenvectors of some set of only distilled SS operators. 

To prove this claim, we start with a lemma:
\begin{lemma}\label{lemmajunk}
For SRE$\to$LRE, $\psi\left[\mc G_{\text{\emph{junk}}}^{\perp_\Lambda}/ \mc R\right] \cap \psi\left[\mc G_{\text{\emph{junk}}}^{\perp_\Lambda}/\mc R\right]^{\perp_\Omega}$ is trivial.
\end{lemma} 

\begin{proof}
Let $A \in \psi\left[\mc G_{\text{junk}}^{\perp_\Lambda}/\mc R\right] \cap \psi\left[\mc G_{\text{junk}}^{\perp_\Lambda}/\mc R\right]^{\perp_\Omega}$, which implies for all $B \in \psi\left[\mc G_{\text{junk}}^{\perp_\Lambda}/ \mc R\right] $, $\omega(A,B)=0$.  By definition, there exists an $F,G \in \mc G_{\text{junk}}^{\perp_\Lambda}/\mc R$ such that $A = \psi(F)$ and $B= \psi(G)$. Thus we have 
\begin{align}
0= \omega(A,B) = \omega(\psi(F), \psi(G))= \lambda( \phi\psi(F), G) = \lambda(\pi_{\mc R^{\perp_\Lambda}} \phi \psi(F), G).
\end{align}
 As this is true for all $G \in \mc G_{\text{junk}}^{\perp_\Lambda}/ \mc R$, this implies $\pi_{\mc R^{\perp_\Lambda}}\phi \psi(F) = \pi_{\mc R^{\perp_\Lambda}}\phi(A) \in  \left(\mc G_{\text{junk}}^{\perp_\Lambda} / \mc R\right)^{\perp_\Lambda} /\mc R =  \mc G_{\text{junk}}$ as $\mc G_{\text{junk}}$ is a mutually commuting set in $\mc R^{\perp_\Lambda}$ (see Appendix \ref{apx:F2vec}). By definition of $\mc G_{\text{junk}}$, this implies that there exists a $B \in \mc L$, such that $\pi_{\mc R^{\perp_\Lambda}} \phi(A) = \pi_{\mc R^{\perp_\Lambda}} \phi(B)$ or $\pi_{\mc R^{\perp_\Lambda}} \phi(A+B) =0$. We thus have $\phi(A+B) \in \mc R$ which by the maximal condition implies $A+B \in \mc L$ and $A \in \mc L$.  But  as all of  $\psi\left[\mc G_{\text{junk}}^{\perp_\Lambda}/ \mc R\right]\subseteq \mc L^{\perp_\Omega}$, we have that $A \in \mc L \cap \mc L^{\perp_\Omega}$ which because $\mc L$ is projective, implies $A=0$. Therefore, $ \psi\left[\mc G_{\text{junk}}^{\perp_\Lambda}/\mc R\right] \cap \psi\left[\mc G_{\text{junk}}^{\perp_\Lambda}/\mc R\right]^{\perp_\Omega}$ is trivial.
\end{proof}

By Proposition \ref{app:proj} in Appendix \ref{apx:F2vec}, we have thus shown that $\psi\left[\mc G_{\text{junk}}^{\perp_\Lambda}/ \mc R\right]$ is a projective space.  As a corollary, if $A \notin \psi\left[\mc G_{\text{junk}}^{\perp_\Lambda}/\mc R\right]$ then there exists a $D \in \psi\left[\mc G_{\text{junk}}^{\perp_\Lambda}/\mc R\right]^{\perp_\Omega}$ such that $\omega(A,D)\neq 0$, for if this were not true, then $A \in \left(\psi\left[\mc G_{\text{junk}}^{\perp_\Lambda}/ \mc R\right]^{\perp_\Omega}\right)^{\perp_\Omega}= \psi\left[\mc G_{\text{junk}}^{\perp_\Lambda}/\mc R\right]$, by the discussion in Appendix \ref{apx:F2vec}. As this is a contradiction, the corollary holds. We use Lemma \ref{lemmajunk} and its corollary to prove the following proposition:

\begin{proposition}\label{prop:acc}
For SRE$\to$LRE, $F\in \ker {}_{\mc L} \psi_{\mc R}$ is accidental, if and only if there exists a distilled SS operator $G \in \pi_{\mc R} \phi\left[\psi\left[\mc G_{\text{\emph{junk}}}^{\perp_\Lambda}/\mc R\right]^{\perp_\Omega}\right] \subseteq \pi_{\mc R} \phi\left[\mc A_{\text{\emph{distill}}}^{\perp_\Omega}\right]$ such that $\lambda(F,G) \neq 0$. 
\end{proposition}

\begin{proof}

($\Leftarrow$) Let $F \in \ker {}_{\mc L} \psi_{\mc R}$ and suppose there exists a $G \in \pi_{\mc R}\phi\left[\psi\left[\mc G_{\text{junk}}^{\perp_\Lambda}/ \mc R\right]^{\perp_\Omega}\right]$ such that $\lambda(F,G) \neq 0$. This implies there exists a $D \in \psi\left[\mc G_{\text{junk}}^{\perp_\Lambda}/\mc R\right]^{\perp_\Omega}$ such that $G = \pi_{\mc R}\phi(D)$ and that
\begin{align}
0 \neq \lambda(F, G)= \lambda(F, \pi_{\mc R} \phi(D))= \omega(\psi(F), D).
\end{align}
This implies that $\psi(F) \notin \left(\psi\left[\mc G_{\text{junk}}^{\perp_\Lambda}/\mc R\right]^{\perp_\Omega}\right)^{\perp_\Omega} = \psi\left[\mc G_{\text{junk}}^{\perp_\Lambda}/\mc R \right]$, using Lemma \ref{lemmajunk}. Therefore by definition, $F$ is accidental.

($\Rightarrow$)Let $F\in \ker {}_{\mc L} \psi_{\mc R}$ be accidental. By definition, this implies $\psi(F) \notin \psi\left[\mc G_{\text{junk}}^{\perp_\Lambda}/ \mc R\right]$. By the corollary to Lemma \ref{lemmajunk}, this implies there exists a $ D \in \psi\left[\mc G_{\text{junk}}^{\perp_\Lambda}/ \mc R\right]^{\perp_\Omega} \subseteq \psi\left[\mc G_{\text{junk}}\right]^{\perp_\Omega} = \mc A_{\text{distill}}^{\perp_\Omega}$, such that 
\begin{align}
0 \neq \omega(D,\psi(F))= \lambda( \phi(D), F) = \lambda(\pi_{\mc R} \phi(D), F).  
\end{align}
As $D \notin \mc L$ by the hypothesis that $F \in \ker {}_{\mc L} \psi_{\mc R}$, we therefore have that there exists a distilled SS operator $G=\pi_{\mc R} \phi(D)\in \pi_{\mc R} \phi\left[\psi\left[\mc G_{\text{junk}}^{\perp_\Lambda}/\mc R\right]^{\perp_\Omega}\right] \subseteq \pi_{\mc R} \phi\left[\mc A_{\text{distill}}^{\perp_\Omega}\right]$ such that $\lambda(F,G) \neq 0$.
\end{proof}

As a last point, we want to understand how the choice of completion of the junk group $\mc G_{c\hyph\text{junk}}$ determines the ground state which results from the perturbation. In the case of condensation, the logical operators of the parent are essentially those of the target and co-target, so if one starts in one ground state and condenses via the strong perturbation, the resulting target ground state should coincide with the parent ground state. For the case of SRE$\to$LRE, however, there is only one parent ground state, and yet the target has many ground states. It is reasonable to think that which ground state we end up in is dependent on our choice of completion. If we let $p_{c\hyph\text{junk}}^{(0)}$ be the projection operator formed from the members of $\mc G_{c\hyph\text{junk}}$ and Eq. \ref{eq:proj}, the ground state is (see Appendix \ref{apx:PT}),
\begin{align}\label{eq:newgs}
\ket{\psi_{\text{parent}}(0)} \to  p_{c\hyph\text{junk}}^{(0)}\ket{\psi_{\text{parent}}(0)} = \ket{\psi_{\mc L} (0)} \otimes \left(p^{(0)}_{c\hyph\text{junk}} \ket{\psi_{\mc L^{\perp_\Omega}}(0)}\right),
\end{align}
where we have split the parent ground state into the product state between $\mc L$ and $\mc L^{\perp_\Omega}$. \break $p^{(0)}_{c\hyph\text{junk}} \ket{\psi_{\mc L^{\perp_\Omega}}(0)}$ represents an equal superposition over a subspace of p-string in the case of our examples. Now we consider acting on this state with some logical operator $F \in \ker {}_{\mc L} \psi_{\mc R}$ to determine the ground state. As $\psi(F)\in \mc L^{\perp_\Omega}$, this only acts non-trivially on the second factor. It also commutes with all of $\mc G_{c\hyph\text{junk}}$ and thus it commutes with $p_{c\hyph\text{junk}}^{(0)}$. So the action of the logical operator on our ground state is given by  $p^{(0)}_{c\hyph\text{junk}} \ket{\psi_{\mc L^{\perp_\Omega}}(F)}$. This is equal to our original state, i.e. the ground state is a $+1$ eigenstate of $F$ if and only if $\psi(F) \in \psi[\mc G_{c\hyph\text{junk}}]$. This relates back to distilled SS operators in  $\psi[\mc G_{c\hyph \text{junk}}]^{\perp_\Omega} \subseteq \mc A_{\text{distill}}^{\perp_\Omega}$ via the following proposition:
\begin{proposition}\label{prop:loops}
for SRE$\to$LRE,  if $ \mc G_{c\hyph \text{\emph{junk}}} =\left( \mc G_{c\hyph \text{\emph{junk}}}\right)^{\perp_\Lambda}/ \mc R$ i.e. it is a maximal mutually commuting set in $\mc R^{\perp_\Lambda}$, then for all $ F \in \mc R$, $F \in \pi_{\mc R} \phi\left[\psi[\mc G_{c\hyph \text{\emph{junk}}}]^{\perp_\Omega}\right]$ if and only if $\psi(F) \in \psi[\mc G_{c\hyph\text{\emph{junk}}}]$.
\end{proposition}

\begin{proof}
($\Rightarrow$) Let $ \mc G_{c\hyph \text{junk}} =\left( \mc G_{c\hyph \text{junk}}\right)^{\perp_\Lambda}/ \mc R$ and $F \in \pi_{\mc R} \phi\left[\psi[\mc G_{c\hyph \text{junk}}]^{\perp_\Omega}\right] \subseteq \mc R$. This implies there exists a $D \in \psi[\mc G_{c\hyph \text{junk}}]^{\perp_\Omega}$  such that $F= \pi_{\mc R} \phi(D)$. Furthermore for all $G \in \mc G_{c\hyph\text{junk}}$
\begin{align}
0 = \omega( D, \psi(G)) = \lambda(\phi(D), G)= \lambda(\pi_{\mc R^{\perp_\Lambda}} \phi(D),G ).
\end{align}
This implies $\pi_{\mc R^{\perp_\Lambda}} \phi(D) \in\left( \mc G_{c\hyph\text{junk}}\right)^{\perp_\Lambda}/\mc R= \mc G_{c\hyph\text{junk}}$ by our hypothesis. Thus $\psi(F) = \psi(\pi_{\mc R} \phi(D)) = \psi(\pi_{\mc R^{\perp_\Lambda}} \phi(D)) \in \psi[\mc G_{c\hyph\text{junk}}]$.

($\Leftarrow$) Let  $ \mc G_{c\hyph \text{junk}} \subseteq \left( \mc G_{c\hyph \text{junk}}\right)^{\perp_\Lambda}/\mc R$ and $F \in \mc R$ such that $\psi(F) \in \psi[\mc G_{c\hyph\text{junk}}]$. This implies there exists an $F^{\perp} \in \mc G_{c\hyph\text{junk}}\subseteq \mc R^{\perp_\Lambda}$ such that $\psi(F + F^{\perp}) =0$. By the triviality of the parent model, this implies there exists an $A \in \mc A$ such that $F+ F^\perp = \phi(A)$ and $ \pi_{\mc R^{\perp_\Lambda}}\phi(A) = F^\perp$. Now let $D \in \psi[\mc G_{c\hyph\text{junk}}]$, which implies there exists a $G \in \mc G_{c\hyph\text{junk}}$ such that $D= \psi(G)$. Then consider
\begin{align}
\omega(A,D)= \omega(A, \psi(G))= \lambda(\phi(A), G)= \lambda( \pi_{\mc R^{\perp_\Lambda}}\phi(A), G)= \lambda(F^\perp, G) =0,
\end{align}
where the last equality uses our hypothesis. As this is true for all such $D$, we have that $A \in  \psi[\mc G_{c\hyph\text{junk}}]^{\perp_\Omega}$ and therefore, $F = \pi_{\mc R} \phi(A) \in  \pi_{\mc R} \phi\left[\psi[\mc G_{c\hyph \text{junk}}]^{\perp_\Omega}\right]$.
\end{proof}

Note that the proof of  $\Leftarrow$ only requires that $ \mc G_{c\hyph \text{junk}}$ is mutually commuting in $\mc R^{\perp_\Lambda}$ and so this direction of the conditional is true even if  $\mc G_{c\hyph \text{junk}}$ is not a maximal mutually commuting set in $\mc R^{\perp_\Lambda}$. 

We now want to make the claim that if $ \mc G_{c\hyph \text{junk}}$ is a maximal mutually commuting set in $\mc R^{\perp_\Lambda}$, then the resulting ground state in Eq.~\eqref{eq:newgs} is completely specified as the simultaneous $+1$ eigenstate of all members of  $\pi_{\mc R} \phi\left[\psi[\mc G_{c\hyph \text{junk}}]^{\perp_\Omega}\right]$, i.e. it is the singular ground state characterized as the eigenstate of the distilled SS operators in this set. Proposition ~\ref{prop:loops} already tells us that the ground state is an eigenstate of all $\pi_{\mc R} \phi\left[\psi[\mc G_{c\hyph \text{junk}}]^{\perp_\Omega}\right]$, so all we need is that this set is itself a maximal mutually commuting set in $\mc R$.  Before doing so, we must prove the following lemma:
\begin{lemma}\label{perpperp}
For SRE$\to$LRE, if the completion can be characterized as a gauge structure $GS'= \left((\mc A', \phi', \Omega'), (\mc F, \psi', \Lambda)\right)$, such that $\im \phi' = \mc G_{c\hyph\emph{\text{junk}}}$, then $\left(\psi[\mc G_{c\hyph \text{\emph{junk}}}]^{\perp_\Omega}\right)^{\perp_\Omega} = \psi[\mc G_{c\hyph \text{\emph{junk}}}]$.
\end{lemma}
\begin{proof}
Suppose there exists a gauge structure $GS'= \left((\mc A', \phi', \Omega'), (\mc F, \psi', \Lambda)\right)$, such that $\im \phi' = \mc G_{c\hyph\text{junk}}$. This allows us to generate a {\it composite} gauge structure, $\text{GS}\circ \text{ GS}'=\left( (\mc A', \psi\phi', \Omega'), (\mc A, \psi' \phi, \Omega)\right)$. This fits the definition of a gauge structure as $\Omega$ and $\Omega'$ are invertible, and for all $A \in \mc A$ and $A' \in \mc A'$, 
\begin{align}
\omega(\psi \phi'(A'), A) = \lambda(\phi'(A), \phi(A))= \omega'(A', \psi' \phi(A)),
\end{align}
which we recognize as the second condition on a gauge structure. We also recognize that the composite is also {\it transposable}, i.e. $\text{GS}'\circ \text{ GS}=\left( (\mc A, \psi'\phi, \Omega), (\mc A', \psi \phi', \Omega')\right)$ is also a gauge structure.  This implies the BrLE rules hold for both $\text{GS}\circ \text{ GS}'$ and $\text{GS}'\circ \text{ GS}$. Of the four, we focus on the two BrLE rules,
\begin{subequations}
\begin{align}
(\im \psi \phi')^{\perp_\Omega}= \ker \psi' \phi \\
(\ker \psi' \phi)^{\perp_\Omega} = \im \psi \phi'.
\end{align}
\end{subequations}
Therefore when combined, we have 
\begin{align}
 \left(\psi [\mc G_{c \hyph\text{junk}}]^{\perp_\Omega}\right)^{\perp_\Omega} =\left( (\im \psi \phi')^{\perp_\Omega} \right)^{\perp_\Omega} = \im \psi \phi' = \psi [\mc G_{c \hyph\text{junk}}].
\end{align}
\end{proof}
In all our examples, this is the case as the completion can also be described as a stabilizer code. Thus in general,  $ \left(\psi [\mc G_{c \hyph\text{junk}}]^{\perp_\Omega}\right)^{\perp_\Omega} =\psi [\mc G_{c \hyph\text{junk}}]$ holds. This is but one use for composite and transposable gauge structures which we look to explore more in future work.

With Lemma \ref{perpperp}, we now prove the following:
\begin{proposition}\label{prop:comp}
For SRE$\to$LRE,  if $ \mc G_{c\hyph \text{\emph{junk}}} =\left( \mc G_{c\hyph \text{\emph{junk}}}\right)^{\perp_\Lambda}/\mc R$ i.e. it is a maximal mutually commuting set in $\mc R^{\perp_\Lambda}$ and corresponds to a gauge structure completion, then \break $ \pi_{\mc R} \phi\left[\psi[\mc G_{c\hyph \text{\emph{junk}}}]^{\perp_\Omega}\right] \subseteq \ker {}_{\mc L}\psi_{\mc R}$ is a maximal mutually commuting set in $\mc R$.
\end{proposition}

\begin{proof}
 Let $\mc G_{c\hyph \text{junk}} =\left( \mc G_{c\hyph \text{junk}}\right)^{\perp_\Lambda}/\mc R$ correspond to a gauge structure completion. Furthermore, let $G \in \mc R$ be such that for all $F \in \pi_{\mc R} \phi\left[\psi[\mc G_{c\hyph \text{junk}}]^{\perp_\Omega}\right]$, $\lambda(F,G)=0$. Let $D \in \psi[\mc G_{c\hyph \text{junk}}]^{\perp_\Omega}$ which implies
\begin{align}
\omega(\psi(G), D) = \omega(G, \phi(D))= \lambda(G, \pi_{\mc R} \phi(D)) =0.
\end{align}
As $\pi_{\mc R} \phi(D) \in\pi_{\mc R} \phi\left[\psi[\mc G_{c\hyph \text{junk}}]^{\perp_\Omega}\right]$. This implies $\psi(G) \in\left(\psi[\mc G_{c\hyph \text{junk}}]^{\perp_\Omega}\right)^{\perp_\Omega} = \psi[\mc G_{c\hyph \text{junk}}]$ by Lemma \ref{perpperp}, which implies there exists a $G^{\perp} \in \mc G_{c\hyph\text{junk}}$ such that $\psi(G + G^{\perp}) =0$. As the parent model is trivial, this implies there exists an $A\in \mc A$ such that $\phi(A)= G+ G^{\perp}$. Note that $\pi_{\mc R} \phi(A) = G$ and $\pi_{\mc R^{\perp_\Lambda}} \phi(A) = G^\perp$. This implies
\begin{align}
\psi(G) = \psi \pi_{\mc R} \phi(A) = \psi \pi_{\mc R^{\perp_\Lambda}} \phi(A) = \psi(G^\perp) \in \psi[\mc G_{c\hyph \text{junk}}].
\end{align}
by Proposition~\ref{prop:loops}, this implies $G \in \pi_{\mc R} \phi\left[\psi[\mc G_{c\hyph \text{junk}}]^{\perp_\Omega}\right]$ and therefore $\pi_{\mc R} \phi\left[\psi[\mc G_{c\hyph \text{junk}}]^{\perp_\Omega}\right]$ is a maximal mutually commuting set.
\end{proof}

\subsection{Understand the Examples from the Gauge Substructure Perspective}\label{sec:subex}

The results of the last four sections are technical but we can boil them down to the following: For the case of condensation, we can distinguish an example as being a condensation by finding that all stabilizer products representing the target model land in a specific subspace of the Pauli space representing a reduced number of local degrees of freedom. This implies that a co-target exists, and all LRE of the target is inherited from the parent.
 
For the case of distillation, specifically SRE$\to$LRE, most of the LRE found in the target model descends or is distilled from subsystem symmetries of the parent (though they may not always be confined to any single layer). Specifically, these subsystem symmetries are those which are orthogonal to (overlap an even number of times with) the p-strings generated by the product of operators which form the target stabilizers. After projection, some LRE may not come from this process i.e. are accidental, but there must always be a sufficient amount of distilled SS LRE to specify a stabilizer basis for the ground space. As for the resulting ground state using a general completion, one considers all p-strings which the terms of the completion generate. Then all the SS which are orthogonal to (overlap an even number of times with) those p-strings, once promoted to operators in $\mc R$, are the distilled SS logical operators which determine the resulting ground state. Specifically, the ground state is the simultaneous $+1$ eigenstate of those operators. Furthermore if we think of this as a kind of state preparation, Proposition \ref{prop:LREsource} tells us that no completion can be used for state preparation of an eigenstate of an accidental logical operator.

\subsubsection{Condensation Example: X-cube}\label{sec:condXC}

In the perturbation analysis of  $(TC_2)_{\text{layers}} \to XC$ in Section \ref{sec:pert1}, we claimed that obtaining the X-cube model from stacks of $d=2$ toric code is a condensation process rather than a distillation. To prove this claim according to Def. \ref{def:cond}, we must show that Proposition~\ref{subad} is satisfied. Our unit cell is the two qubits associated with each edge of the cubic lattice. We do not use CLR as this is not projective for a two qubit unit cell, but instead we restrict these local unit cells to 
\begin{align}
\mc R_{\text{local}}= \{ \id_{\mb C_2^{\otimes 2}}, z^1 z^2\} \oplus \{ \id_{\mb C_2^{\otimes 2}}, x^1 \id^2\}= \mc R_{\text{local}}^Z \oplus \mc R_{\text{local}}^X,
\end{align}
so that $\mc R = \bigoplus_e (\mc R_{\text{local}})_e= \bigoplus_e (\mc R^Z_{\text{local}})_e \oplus \bigoplus_e (\mc R^X_{\text{local}})_e$. The choice of $x^1 \id^2$ as opposed to  $ \id^1 x^2$  is arbitrary. We now consider the left restriction. We already know that we can take $\mc D^Z$ as the Z-type left restriction. Analogous to the ``divergence''-like subspace for the Z-type restriction, we shall consider the ``curl''-like subspace for the X-type restriction defined by $\overline{\mc C}^X= \wp(S_{TC_2}^X)^{3L^3}/\overline{\mc D}^X$, i.e. the space of X-type operators mod those from $\overline{\mc D}^X$. We have already discussed how $\phi[\mc D^Z] \subseteq \bigoplus_e (\mc R^Z_{\text{local}})_e$. As for $\overline{\mc C}^X$, we note that $\phi[\overline{\mc D}^X] \subseteq \bigoplus_e \left(\{ \id_{\mb C_2^{\otimes 2}}, x_1 x_2\}\right)_e = \bigoplus_e (\tilde{\mc R}^X_{\text{local}})_e$. So for all $A^X +\overline{\mc D}^X \in \overline{\mc C}^X$, we have
\begin{align}
 \phi[A^X +\mc D^X]=& \phi(A^X) +\phi[\mc D^X]\nonumber \\
\subseteq& \phi(A^X) +  \bigoplus_e (\tilde{\mc R}^X_{\text{local}})_e \in \mc P^X/ \left( \bigoplus_e (\tilde{\mc R}^X_{\text{local}})_e\right) \simeq  \bigoplus_e (\mc R^X_{\text{local}})_e.
\end{align}
So all of $\overline{\mc C}^X$ maps under the image of $\phi$ into a subspace that is isomorphic to $ \bigoplus_e (\mc R^X_{\text{local}})_e$ in a way that is reminiscent of the proof of Proposition~\ref{subad}. The actual image is contained in $\mc P^X/ \left( \bigotimes_e (\tilde{\mc R}^X_{\text{local}})_e\right)$ which we recognize as the equivalence classes for which $x_e^1 \simeq x_e^2$. This is why the exact choice of $\mc R_{\text{local}}^X$ does not matter.

So in total, we have decomposed the space $\mc A_{TC_2 \text{ layer}}$ according to
\begin{align}
\mc A_{TC_2 \text{ layer}} \simeq& \mc D^Z \oplus \mc A^Z/\mc D^Z \oplus \overline{\mc D}^X \oplus \mc A^X/ \overline{\mc D}^X\nonumber \\
 \simeq& \left(\mc D^Z \oplus \overline{\mc C}^X \right) \oplus \left( \overline{\mc D}^X \oplus \mc C^Z\right) \simeq \mc A_{XC} \oplus \mc A_{TC_3},
\end{align}
up to redundancies due to local constraints and where $\mc C^Z = \wp(S_{TC_2}^Z)^{3L^3}/\mc D^Z$. A similar decomposition exists for $\mc F_{TC_2 \text { layers}} \simeq \mc F_{XC} \oplus \mc F_{TC_3}$. This completes the proof of Eq.~\eqref{eq:tcdecomp}. So in terms of the perturbation analysis, the Hilbert space of the $TC_2$ layers were already such that $\mc H_{TC_2 \text{ layers}} \simeq \mc H_{XC} \oplus \mc H_{TC_3}$ and any stabilizer state decomposes as $\ket{\psi_{TC_2 \text{ layers}}(f)} \simeq \ket{ \psi_{XC} (f_\mc R)} \otimes \ket{\psi_{TC_3}(f_{\mc R^{\perp_\Lambda}})}$. The only distinction is the energy of such a state is more naturally written in terms of $\psi_{TC_2 \text{ layers}}(f)$ i.e. the $TC_2$ excitations rather than some complicated combination of the XC and $TC_3$ excitations. But importantly, the ground spaces coincide. The effect of the infinite perturbation by the base model is such that the energetics of the $XC$ are decoupled from that of $TC_3$ via deconfinement of loop excitations such that the energy of an open p-string is no-longer proportional to the length, but rather is a constant due to the open ends. Fractons are then a biproduct of the interplay of deconfinement and non-trivial braiding/conservation laws in the underlying layers  (again, see Ref. \cite{Ma2017, Vijay2017}).  This is given by the fact that for the ground space
\begin{align}
\ket{\psi_{TC_2 \text{ layers}}(0)} \to \ket{ \psi_{XC} (0_{\mc R})} \otimes \left( p^{(0)}_{ZZ} \ket{\psi_{TC_3}(0_{\mc R^{\perp_\Lambda}})}\right),
\end{align}
where $ p^{(0)}_{ZZ} \ket{\psi_{TC_3}(0_{\mc R^{\perp_\Lambda}})}$ represents the condensation of all p-strings not encircling a cycle of the 3-torus. So in total, a condensation does not change our perspective on the Hilbert space, i.e. how we factor it into distinct excitations, but rather changes the energetics of those excitations. 

\subsubsection{Type-I Distillation Example: Cluster-cube}\label{sec:distCC}

Here we describe the SRE$\to$LRE process which takes the $d=2$ cluster model into the cluster-cube model. In this case $\mc L= \mc D^X \oplus \mc D^Z$ is the left restriction and $\mc R= \bigoplus_{v}\left(\mc R_{\text{CLR}}^3 \oplus \mc R_{\text{CLR}}^3\right)_v$ is the right restriction, i.e. we apply two copies of three qubit CLR corresponding one qubit for each direction, for the six qubits associated to each vertex $v$. We have already shown the cluster-cube example is not a condensation and thus a distillation by definition. However, we now also have the tools to understand the source of the fractonic behavior and LRE. On both accounts, we can see that the Type-I behavior descends directly from 1D subsystem symmetries in the layer planes. Operators for two different completions, $S_{c\hyph \triangle}$ and $S_{c'\hyph \triangle}$ are shown in Fig. \ref{fig:ccbuild2}.

\begin{figure}

\centering

\includegraphics[scale=.3]{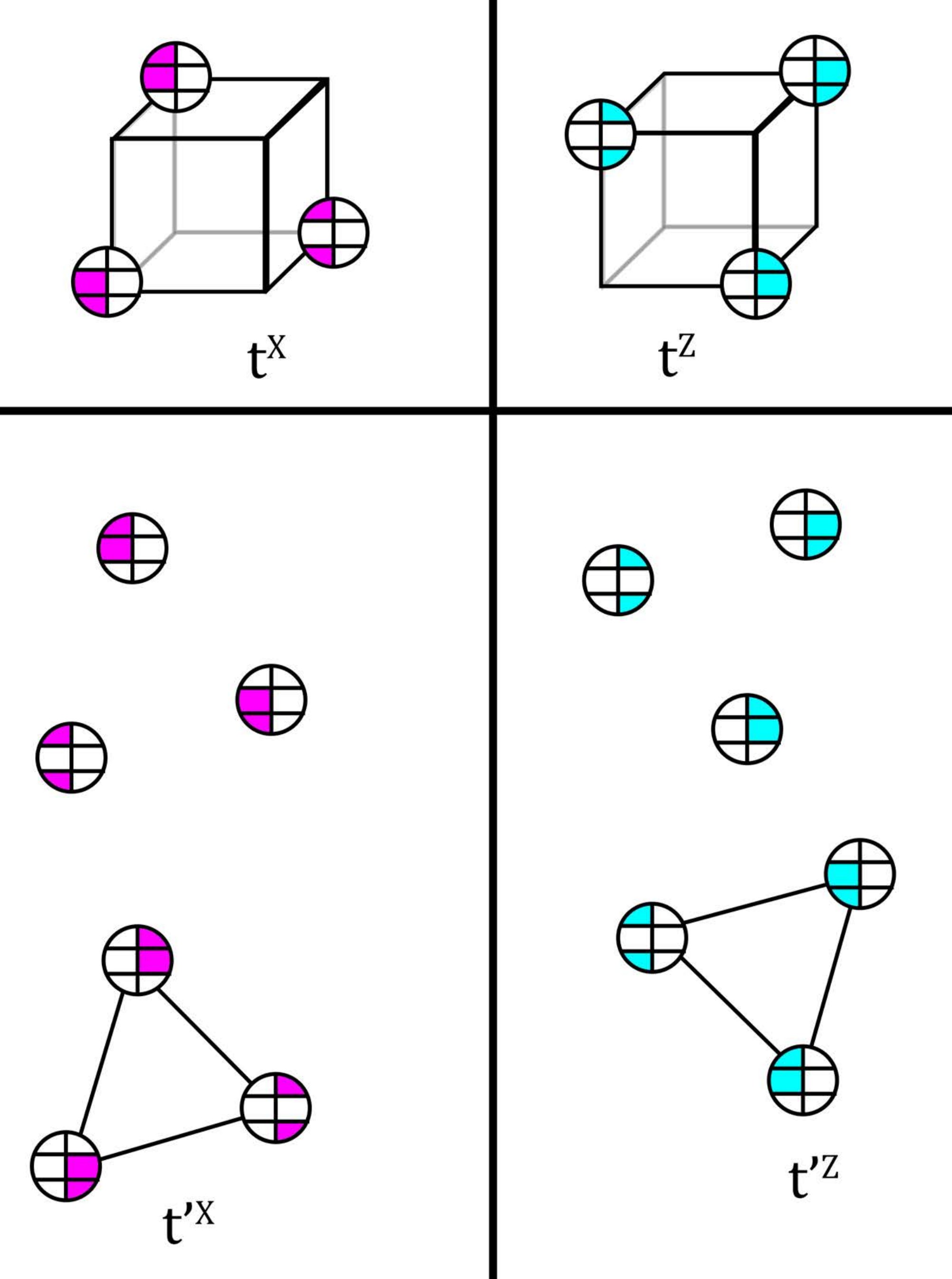}

\caption{(top) Members $t^X, t^Z \in S_{\triangle}$. (bottom) Operators for the two primary completions, Including the $t'^X, t'^Z \in S_{c'\hyph \triangle}$.}\label{fig:ccbuild2}
\end{figure}

\begin{figure}

\centering

\includegraphics[scale=.2]{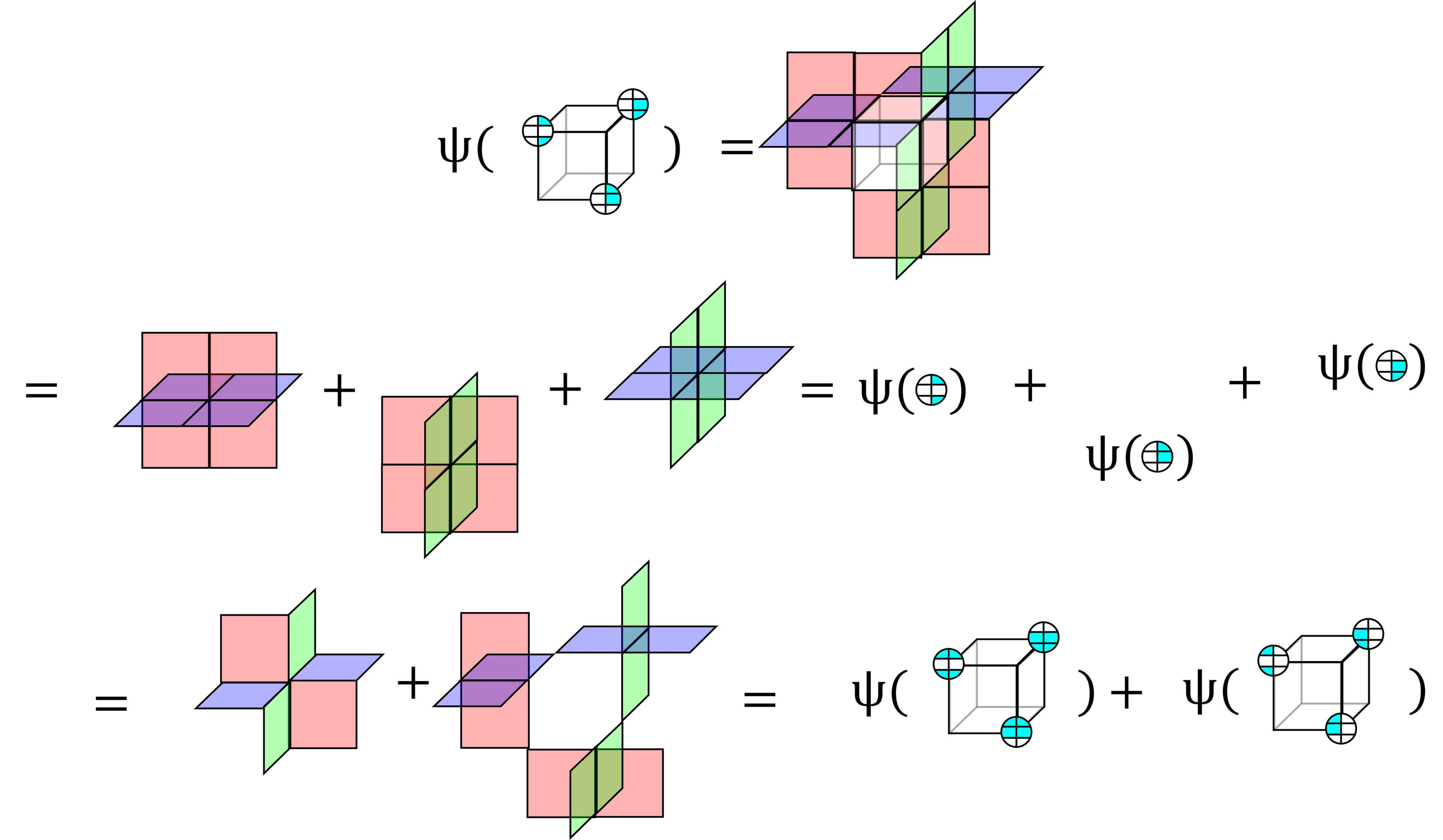}

\caption{X-type p-string configuration generated by a Z-type member of $\mc G_{\text{junk}}$ which is then broken down into the sum of X-type p-strings generated by members of $S_{c \hyph \triangle}$ and $S_{c'\hyph \triangle}$. The X-type members of $\mc G_{\text{junk}}$ generate analogous Z-type p-string configurations.}\label{fig:ccdistloops}

\end{figure}

\begin{figure}

\centering

\begin{tabular}{c c}
\subfloat[$\mc A_{c \hyph \triangle}$ \label{fig:ccSSovera}]{\includegraphics[scale=.15]{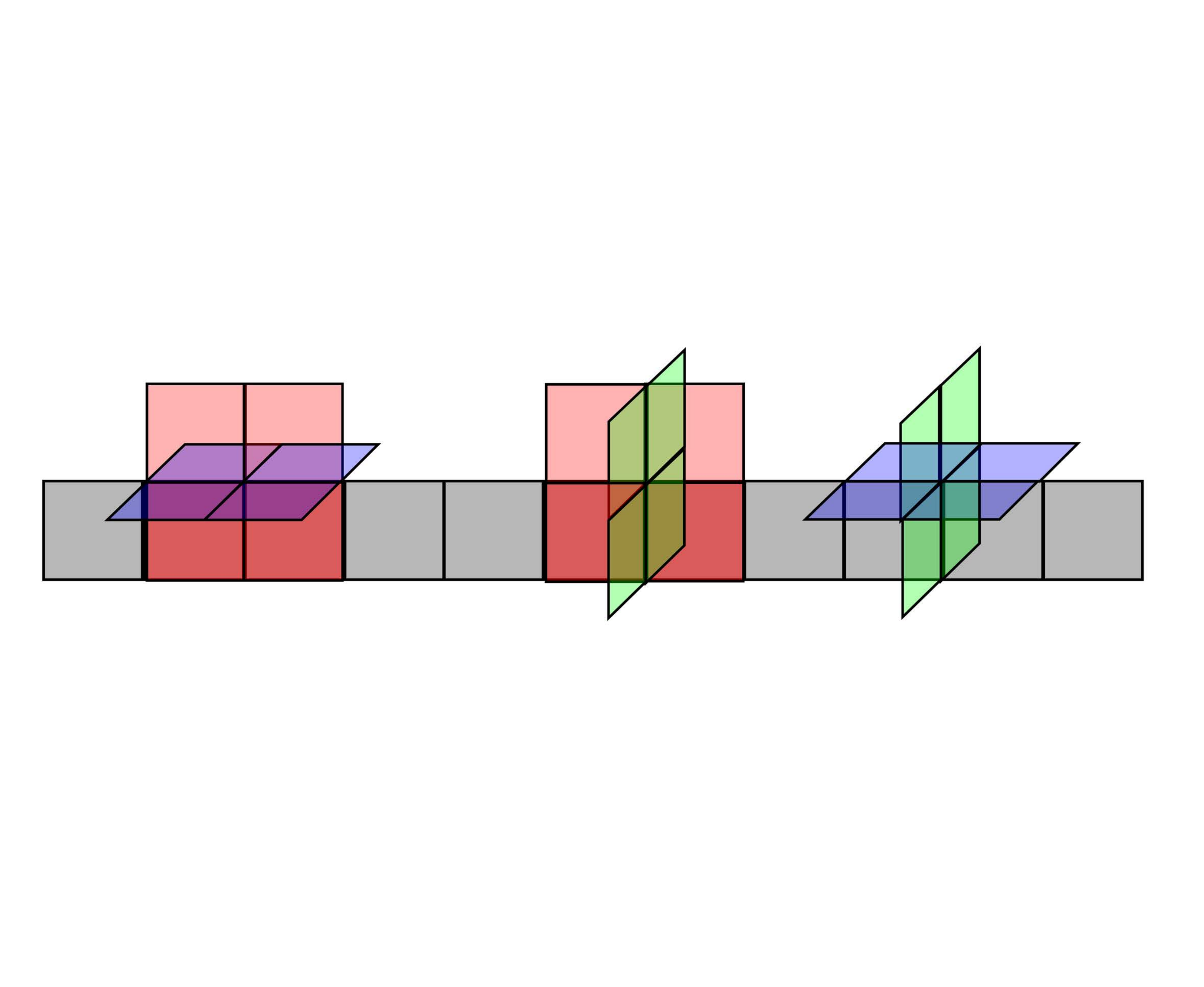}}
\subfloat[$\mc A_{c' \hyph \triangle}$ \label{fig:ccSSoverb}]{\includegraphics[scale=.15]{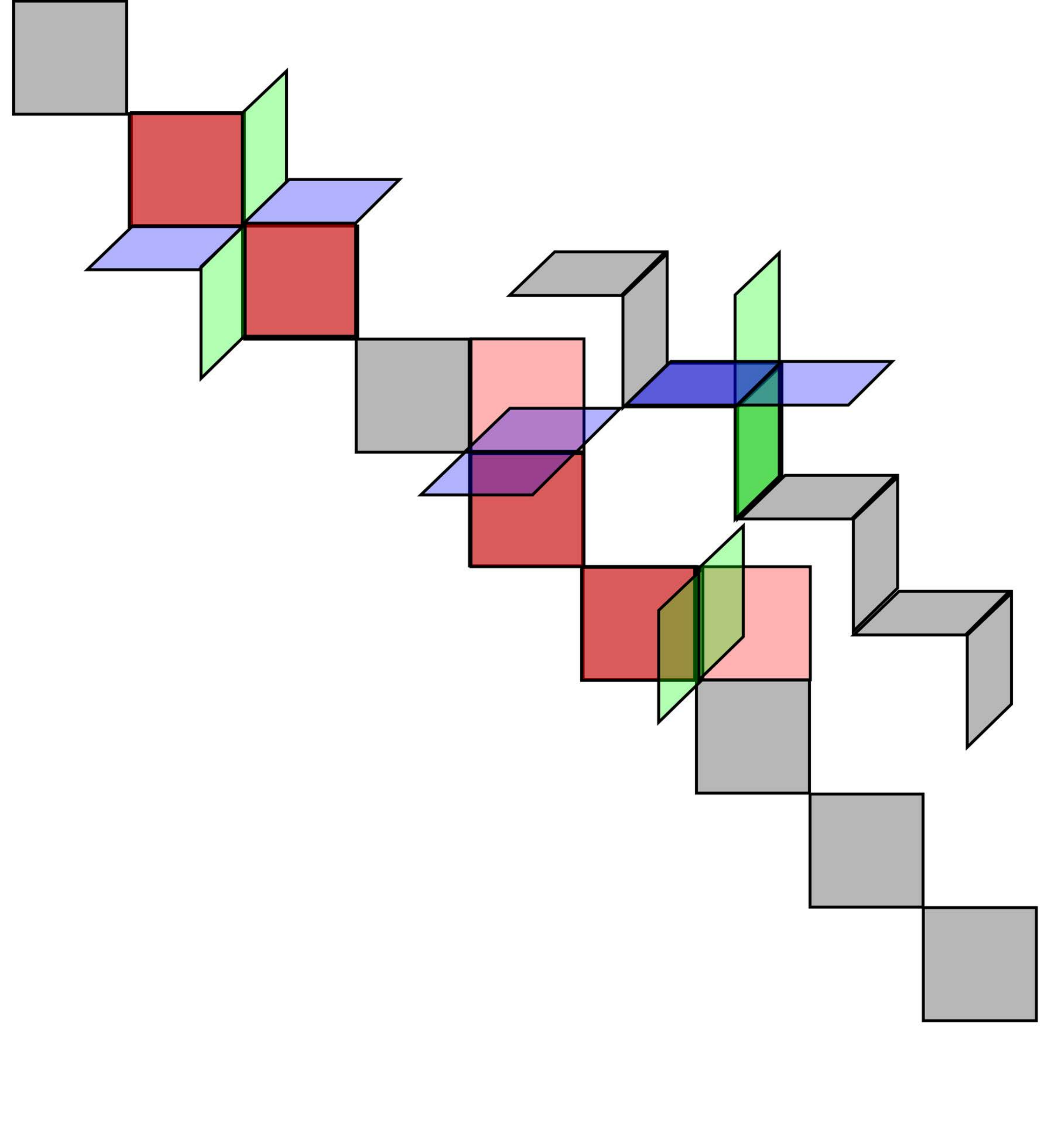}}
\end{tabular}

\caption{Demonstration how the p-string configurations of each completion overlaps an even number of times with the subsystem symmetries (gray). Thus these subsystem symmetries are promoted to logical operators under the map $\pi_{\mc R}\phi$.}\label{fig:ccSSover}

\end{figure}

We start by consider members of $\mc A_{\text{distill}}= \mc A_{\triangle}$ as generated by the p-string configuration in Fig. \ref{fig:ccdistloops}, which is also split into the sum of simpler p-string configurations, two different ways. Splitting it this way is not arbitrary, but corresponds to our two different completions of $\mc G_{\text{junk}}=\mc G_\triangle$. The first corresponds to $S_{c \hyph \triangle}$ as discussed in Section~\ref{sec:ccbase} and generates a maximal mutually commuting subspace in $\mc R^{\perp_\Lambda}$. The second is an alternative completion $S_{c'\hyph \triangle} \supseteq S_{\triangle}$ which includes the original triangle operators as well as another set of triangle operators, $t'^X, t'^Z$ which collectively forms two stacks of our triangular lattice $TC_2$ variant in the $[111]$ direction as can be seen in Fig. \ref{fig:ccbuild2}. Because this set is topologically ordered and has non-trivial logical operators, the subspace it generates in $\mc R^{\perp_\Lambda}$ is not maximal mutually commuting without including some of the $TC_2$  string logical operators. However, such strings are irrelevant as discussed below.

In the case of $\psi[\mc G_{c\hyph\triangle}] =\mc A_{c\hyph\triangle}$  (Fig. \ref{fig:ccSSovera}), we have already discussed how $ (\mc A_{c\hyph\triangle})^{\perp_\Omega}$ contains the $d=1$ subsystem symmetries from Eqs. \eqref{eq:sscluster} in all three directions, which are then promoted to string logical operators under the action of $\pi_\mc R \phi$. As $\mc G_{c\hyph \triangle}$ is maximal mutually commuting, by Proposition \ref{prop:loops}, we know there exists a p-string in $ \mc A_{c\hyph \triangle}$ corresponding to this distilled SS operator under the action of $\psi$. This looks like four p-strings wrapping the torus in one direction and is generated by side-by-side pairs of small p-strings as shown in Fig. \ref{fig:ccstrloop}.  Furthermore, these distilled SS operators are not unique once we mod out $\mc D^X \oplus \mc D^Z$ from $(\mc A_{c\hyph\triangle})^{\perp_\Omega}$. In particular, one can put four subsystem symmetries together to form a product of cubes along a line. Under the action of $\phi$,  the intersection of subsystems along two hinges maps into $\mc R^{\perp_\Lambda}$ and are thus killed by $\pi_{\mc R}$, while the other two hinges become the pair of strings. As discussed before, these can be spread apart and eventually cancel out to form a constraint which is responsible for fractonic behavior. As we can see, this planar constraint is represented in $\mc R^{\perp_\Lambda}$ by the layers of subsystem symmetries formed along the ``stair-stepping'' hinges as demonstrated in Fig. \ref{fig:ccSStoconst}.

We now consider $\psi[\mc G_{c'\hyph\triangle}] =\mc A_{c'\hyph\triangle}$  (Fig. \ref{fig:ccSSoverb}). In this case, one can see that $\left(\mc A_{c'\hyph\triangle}\right)^{\perp_\Omega}$ contains the sum of  plaquettes in a plane which are connected corner-to-corner and extend diagonally as well as the diagonal ``stair-stepping'' sum. The in-plane diagonal subsystem symmetry is then promoted to another distilled SS operators that anti-commutes with the distilled SS logical operator of Eqs. \eqref{eq:sscluster} perpendicular to the plane from the other subsystem symmetry. In the other case, the stair-stepping sum of plaquettes maps into $\mc R^{\perp_\Lambda}$ and is thus a null distilled SS configuration. We can also add two back-to-back diagonal subsystem symmetries and two back-to-back stair-stepping subsystem symmetries to form a stair-stepping sum of cubes, and thus two back-to-back distilled SS logical operator are equivalent under the action of $\pi_{\mc R} \phi$ once we mod out $\mc D^X \oplus \mc D^Z$ from $(\mc A_{c'\hyph\triangle})^{\perp_\Omega}$.  Analogous to Fig. \ref{fig:ccSStoconst}, stacking these stair-stepping cubes along a diagonal plane is equivalent to moving these diagonal distilled SS operators apart to  cancel out and create the same planar constraint. We can also find the member in $\mc A_{c'\hyph\triangle}$ which this distilled SS maps to under $\psi$ a la Proposition \ref{prop:loops} as shown in Fig. \ref{fig:ccstrloop2}. This ends up being true even though $\mc G_{c'\hyph\triangle}$ is not maximal. Interestingly, the p-strings generated by the $TC_2$  sting logical operators, are exactly the stair-stepping subsystem symmetries which are mapped to zero under $\pi_{\mc R} \phi$. This seem to be why Proposition \ref{prop:loops} still holds.  We can also find a local null distilled SS configuration as shown in Fig. \ref{fig:ccnull}. As this belongs to $\mc G_{\triangle}^{\perp_\Lambda}/ \mc R$, we can decompose this into a product of completion operators which is also shown in Fig. \ref{fig:ccnull}. By the scaling of dimensions, such local configurations must exist as the distillation cost is extensively less than the number of discarded degrees of freedom i.e. $\dim \mc G_{\triangle} \sim 2L^3$, whereas $\dim \mc R^{\perp_\Lambda} \sim 4 L^3$. So these operators could be further condensed from this model via Proposition \ref{subad} and would show up as terms in our effective Hamiltonian if we only used $S_\triangle$ for the base model instead of a completion.

\begin{figure}

\centering

\includegraphics[scale=.25]{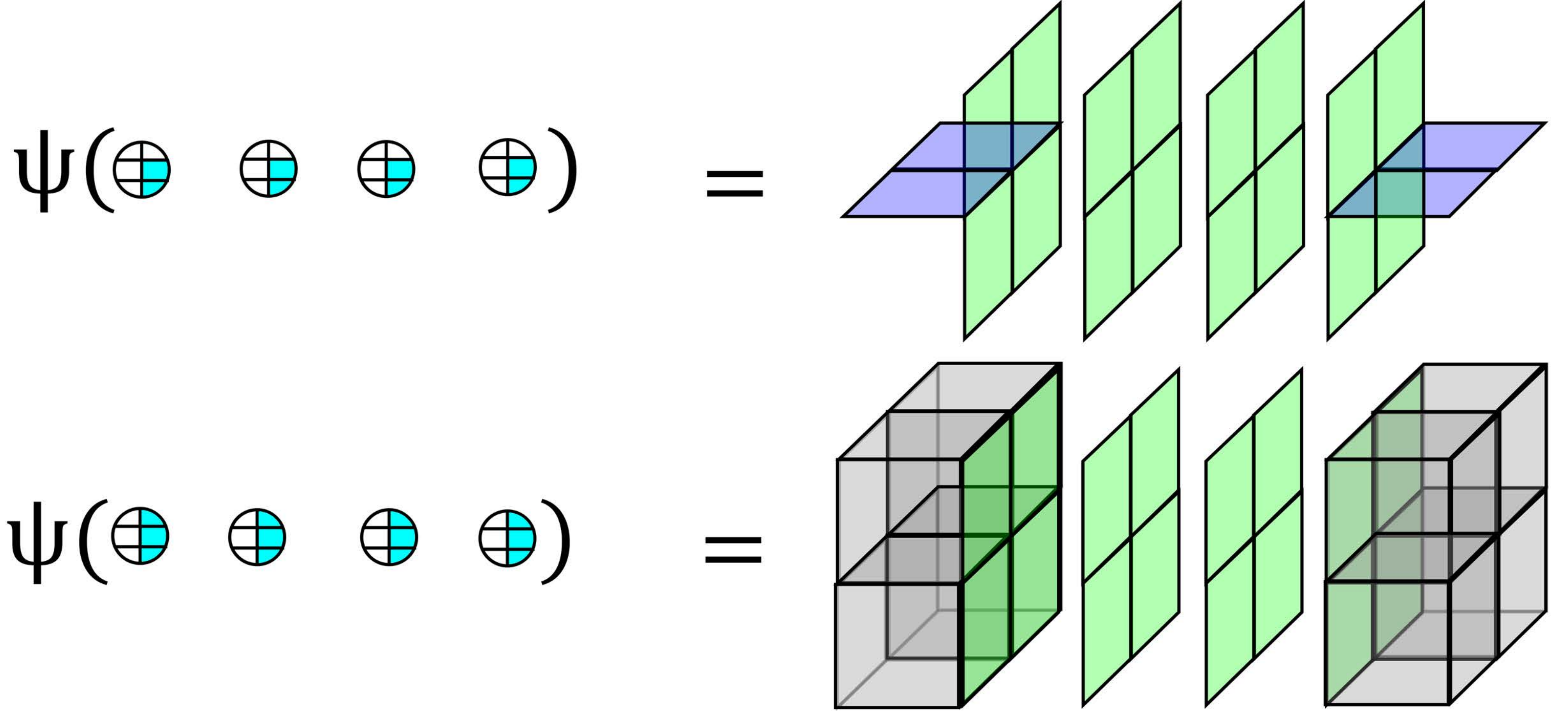}

\caption{(top)X-type p-string configurations generated by a line-like member of $\mc G_{c\hyph \triangle}$. (bottom) Plaquette configuration generated by part of the distilled SS operator. The gray cubes represent the cluster-cube excitation which is reminiscent of the method used in Section \ref{sec:ccdesc} to first describe the logical string operators. The X-type distilled SS operator generates analogous p-string configurations.}\label{fig:ccstrloop}

\end{figure}

\begin{figure}

\centering

\includegraphics[scale=.2]{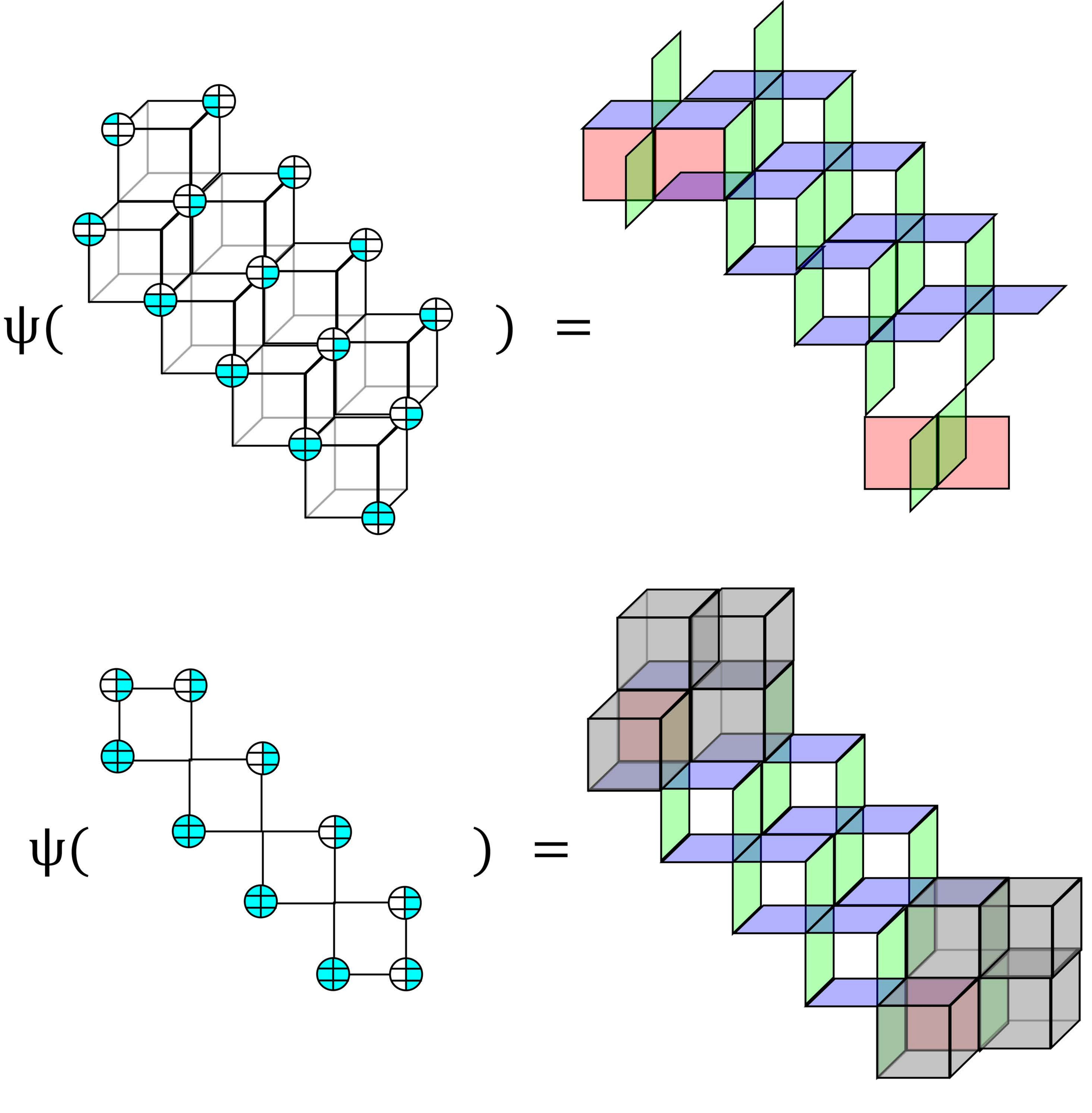}

\caption{(top)X-type p-string configurations generated by a line-like member of $\mc G_{c'\hyph \triangle}$. (bottom) Plaquette configuration generated by part of the diagonal distilled SS operator. The gray cubes represent the cluster-cube excitations generated by this part of the distilled SS operator. The X-type distilled SS operator generates analogous p-string configurations.}\label{fig:ccstrloop2}

\end{figure}

\begin{figure}

\centering

\includegraphics[scale=.15]{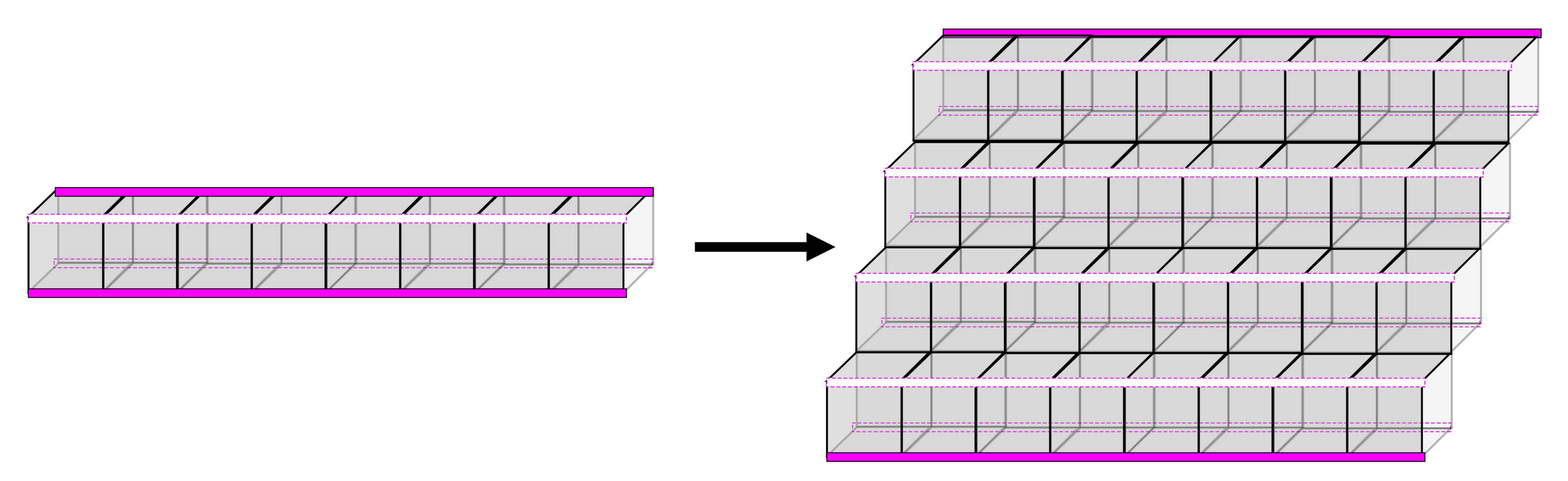}

\caption{Demonstration how distilled SS operators can be combined to form the constraints responsible for fractonic behavior. The solid line represents the X-type support in $\mc R$ whereas the dotted line represents the X-type support in $\mc R^{\perp_\Lambda}$. One can see that if the pattern on the right is continued, all support is contained in $\mc R^{\perp_\Lambda}$.}\label{fig:ccSStoconst}

\end{figure}

\begin{figure}

\centering

\includegraphics[scale=.2]{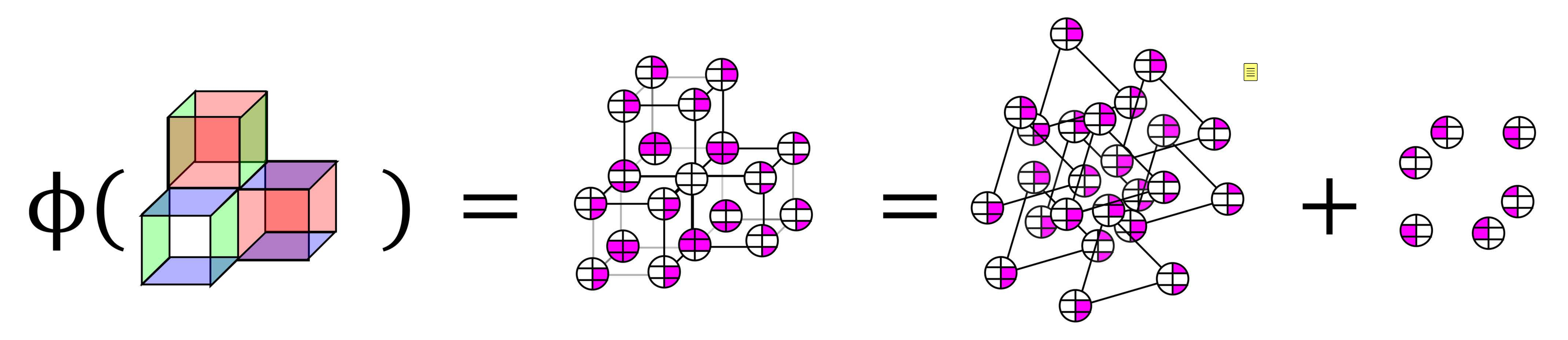}

\caption{Depiction of the null distilled SS member of $\mc A_{\triangle}^{\perp_\Omega}$ and how it maps into $ \mc R^{\perp_\Lambda}$ as well as its decomposition into a sum of $ t'^X, t'^Z$ and single unit cell completion operators.}\label{fig:ccnull}

\end{figure}

\begin{figure}

\centering

\includegraphics[scale=.1]{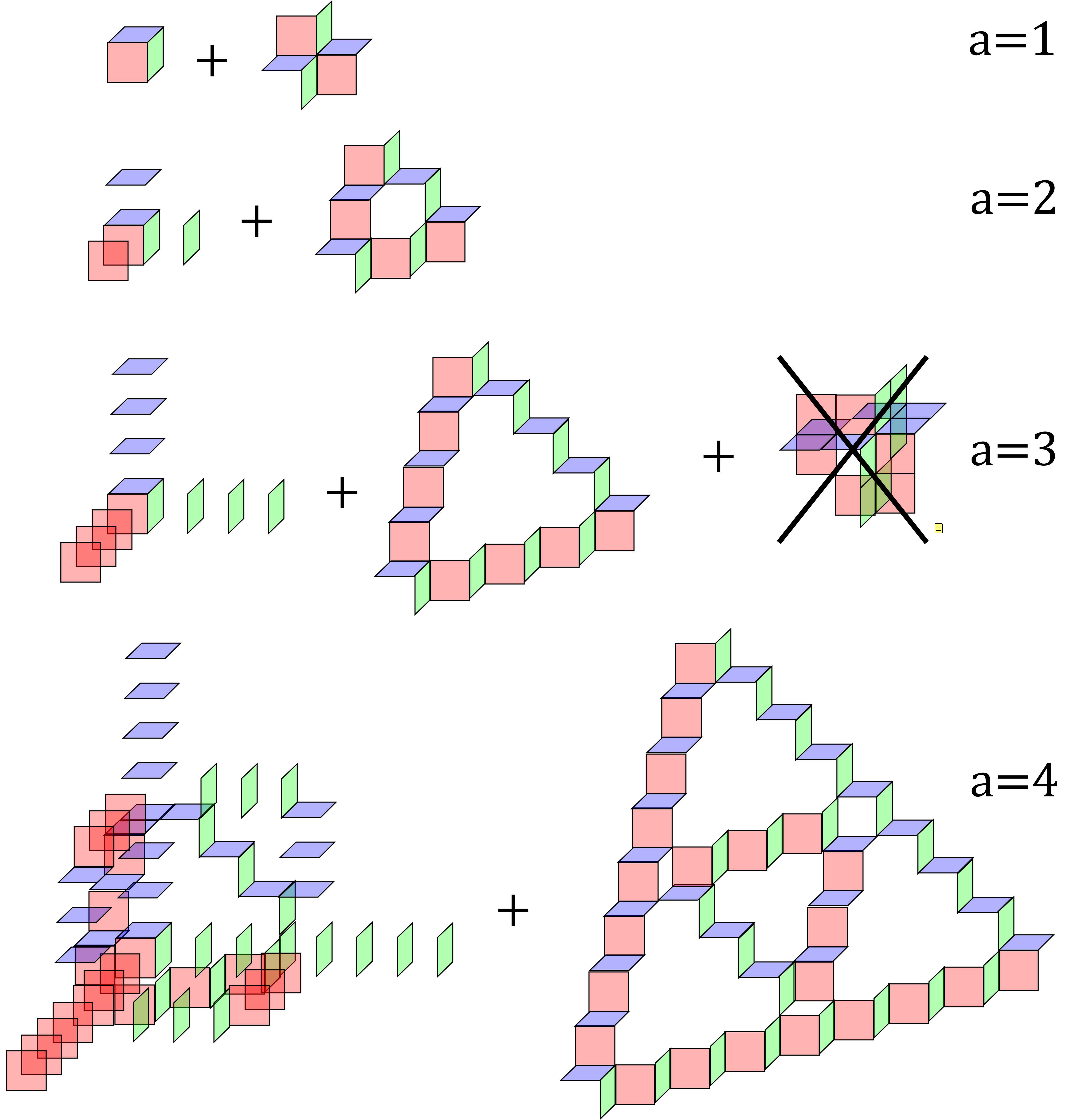}

\caption{Depiction of the p-string configuration generated by the ZZ-type fractal logical operator of the cluster-cube model for the first four generations as indexed by $a$. Note in the third generation, we eliminated the p-string configuration member of $\mc A_{\triangle}$ as we are only concerned with p-strings mod members of this set. }\label{fig:ccfracloop}

\end{figure} 

As for the fractal logical operators, we argue that, singularly, these are accidental. Our argument draws directly on the definition of accidental operators. If one forms the p-string configuration generated by the fractal logical operator, it is given by 3 p-strings wrapping each direction of the 3-torus, which is then dressed by a complicated configuration of locally generated p-strings, i.e. it can be generated by some configuration of the smallest p-strings as shown in Fig. \ref{fig:ccfracloop}. The p-stings which generate $\psi[\mc G_{\triangle}^{\perp_\Lambda}/\mc R]$ in Fig. \ref{fig:ccdistloops} are also locally generated, but as the topologically non-trivial part of the fractal p-string is not locally generated, it is therefore the case that such a logical operator must be accidental. The only caveat is that the string logical operators of the $[111]$ $TC_2$ layers are also in $\mc G_{\triangle}^{\perp_\Lambda}/\mc R$ and the stair-stepping p-strings they form are not locally generated. However, these are null and cannot be deformed out of the $[111]$ plane by any other member of $\mc G_{\triangle}^{\perp_\Lambda}/ \mc R$ so they can't account for the topologically non-trivial part of the fractal p-string. The other fractal logical operators following by a similar argument. Note however, the sum of two such fractal logical operator could be distilled as only an odd number of stings wrapping the torus are topologically non-trivial. By Proposition \ref{prop:acc}, one then expects that there are other distilled SS logical operators. However, the proof requires that $\mc L$ is projective, but this is not strictly true when $L$ is even, precisely when these fractal logical operators exist. Moreover, one can form a hexagonal-like pattern in the $[111]$ plane out of  cube-corner products of plaquettes as shown in Fig. \ref{fig:ccSSop}. One can see that this pattern overlaps an even number of times with the p-strings in $\mc A_{c'\hyph \triangle}$ and requires a periodicity of 4. We conjecture there are several other patterns for different values of $L$ which form distilled SS logical operators.

\begin{figure}

\centering

\includegraphics[scale=.15]{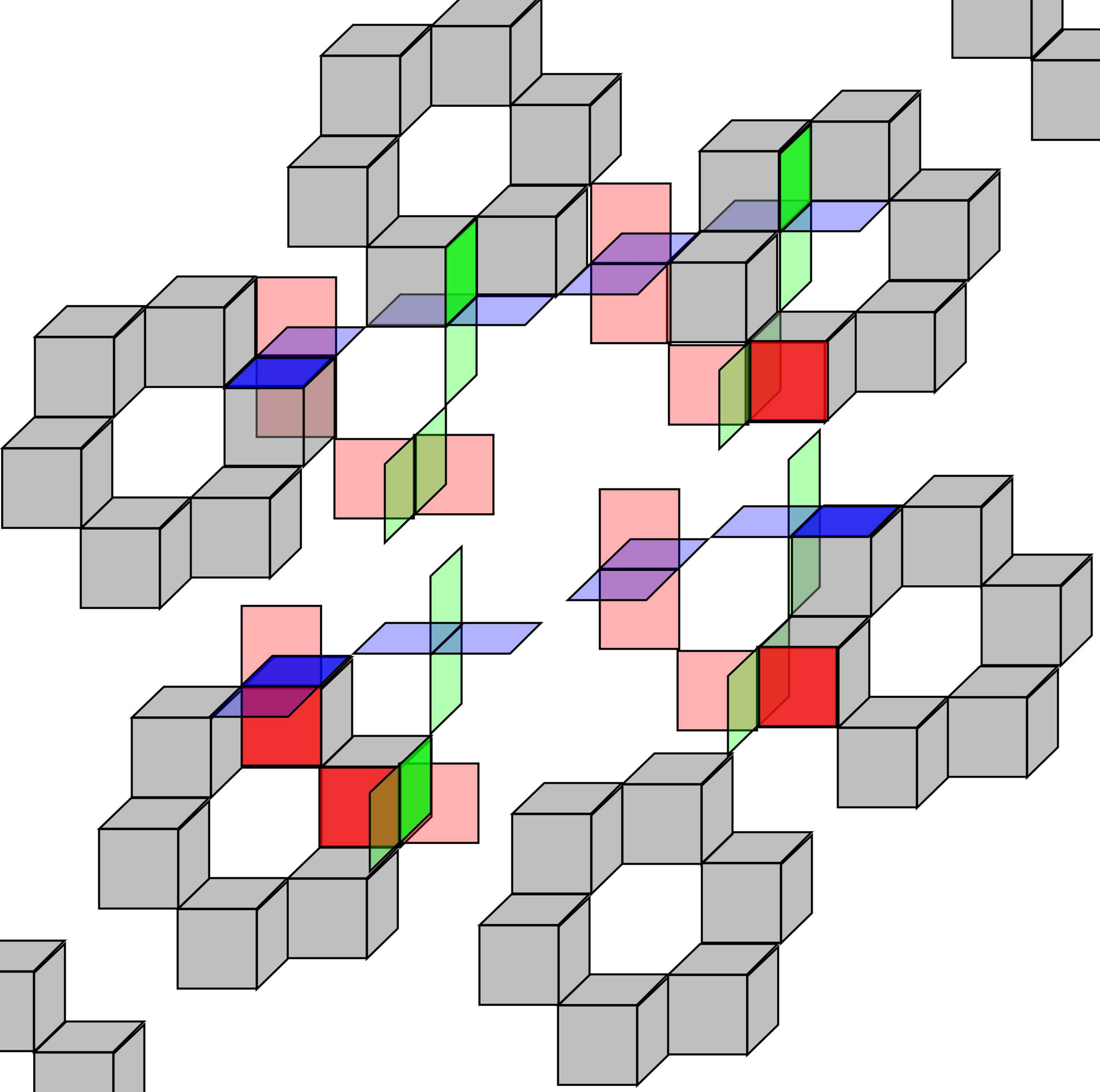}

\caption{2D X-type member of $(\mc A_{c' \hyph \triangle})^{\perp_\Omega}$ (gray) superimposed with several members of $\mc A_{c' \hyph \triangle}$. Highlighted plaquettes show overlap.}\label{fig:ccSSop}

\end{figure} 

\subsubsection{Type-II Distillation Example: Haah's Cubic Code}\label{sec:disthaah}

\begin{figure}

\centering

\includegraphics[scale=.2]{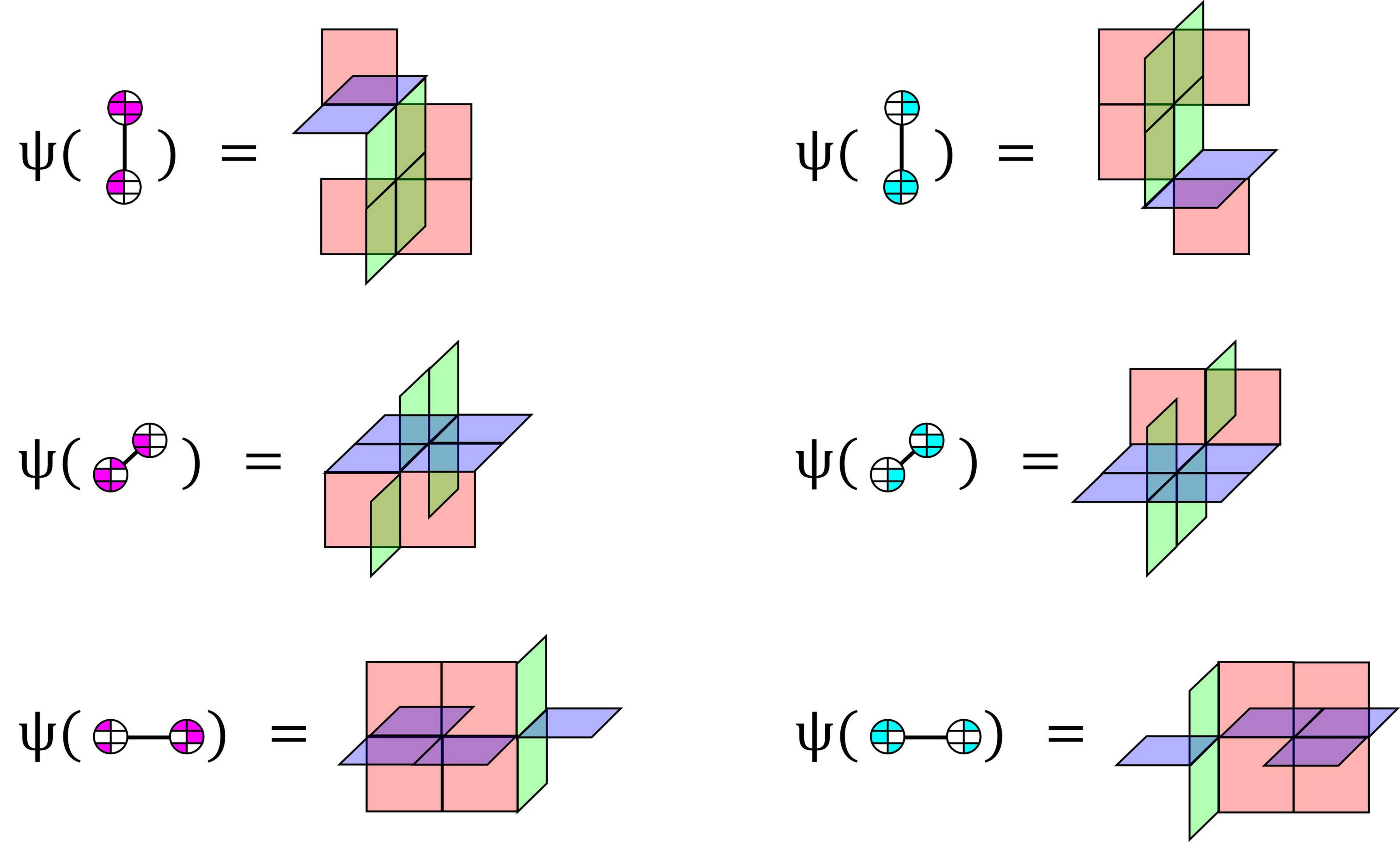}

\caption{Edge operators which generate most of $\mc G_{\text{hex}}^{\perp_{\Lambda}}$ and the p-strings they generate  under the action of $\psi$. }\label{fig:haahfundloops}

\end{figure} 

\begin{figure}

\centering

\includegraphics[scale=.25]{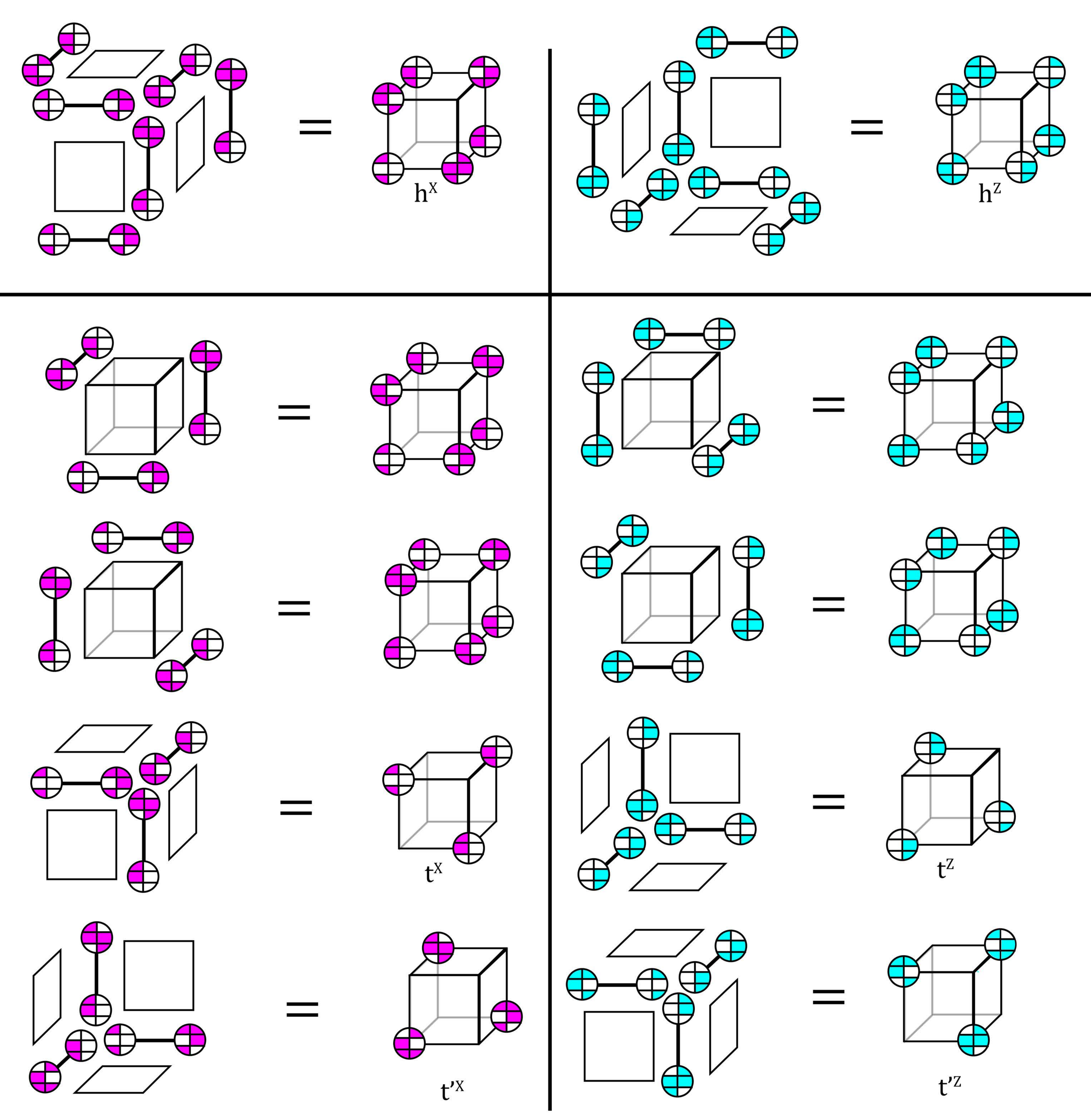}

\caption{(top) Members of $S_{\text{hex}}$ as broken in to edge operators. (bottom) Some of the many operators which can be used to form the completion, including the $t^X, t^Z$ operators discussed in Fig. \ref{fig:hexdef} and $t'^X t'^Z$ operators discussed below, broken into edge operators. }\label{fig:haahbuild2}

\end{figure}

Here we describe the SRE$\to$LRE process which takes the $d=2$ quasi-cluster model into Haah's cubic codel. Once again, $\mc L= \mc D^X \oplus \mc D^Z$ is the left restriction and \break $\mc R= \bigoplus_{v}\left(\mc R_{CLR}^3 \oplus \mc R_{CLR}^3\right)_v$ is the right restriction.

Unlike the cluster-cube SRE$\to$LRE example, members of $\mc A_{\text{distill}}= \mc A_{\text{hex}}$ are more complicated, as shown in Fig. \ref{fig:haahloops}. Furthermore, there are many roughly equivalent choices of completions. To demonstrate this, we can find a set of operators which we refer to as ``edge'' operators and generates most of $\mc G_{\text{hex}}^{\perp_\Lambda}/ \mc R$. These operators, as well as the p-strings they generate, are shown in Fig. \ref{fig:haahfundloops}. Just as with $\mc G_{c'\hyph \triangle}$ from the cluster-cube example, the completions that we discuss form $[111]$ layers of the variant of $TC_2$ on the triangular lattice. This implies there are string operators which are also in $\mc G_{\text{hex}}^{\perp_\Lambda}/ \mc R$, but are not null configurations in this case. The edge operators do not mutually commute, so they cannot be used to form a completion, as originally defined. But they can be used to form $h^X, h^Z$ as shown in Fig. \ref{fig:haahbuild2} and if for every $h^X, h^Z$ operator we choose some other mutually commuting product of edge operators, these operators, along with $h^X, h^Z$, collectively form a completion, $S_{c \hyph \text{hex}}$. Fig. \ref{fig:haahbuild2} shows several options, all of which form a version of $TC_2$. Unfortunately, none of these completions simplify the terms of the base Hamiltonian. However, the completion containing the $t'^X, t'^Z$ operators from Fig. \ref{fig:haahbuild2} splits the $h^X, h^Z$ generated p-strings as show in Fig. \ref{fig:haahloops}, where the left p-string is generated by $t'^Z$ and the right is generated by $h'^Z= t'^Z + h^Z$. An analogous splitting holds for $t'^X$, and $h'^X= t'^X  +h^X$.

\begin{figure}

\centering

\includegraphics[scale=.2]{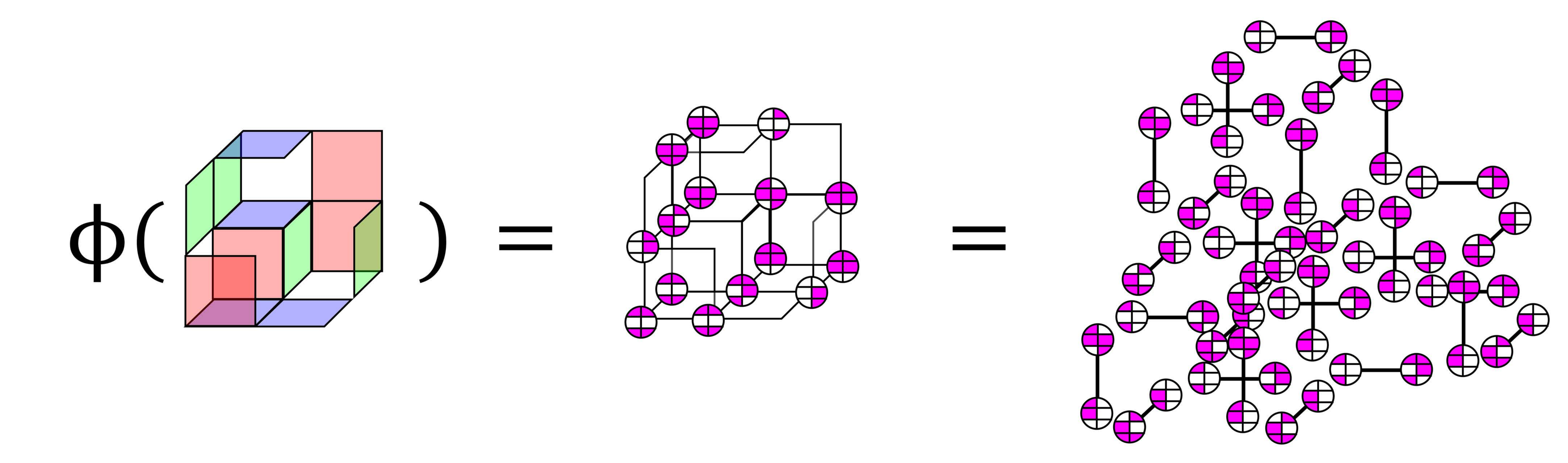}

\caption{Depiction of the null distilled SS member of $\mc A_{\text{hex}}^{\perp_\Omega}$ and how it maps into $ \mc R^{\perp_\Lambda}$ as well as its decomposition into edge operators.}\label{fig:haahnull}

\end{figure} 

To understand the distilled SS, we consider two important configurations. The first is a local pattern of plaquettes which belongs to $ \mc A_{\text{hex}}^{\perp_\Omega}$ but maps to $\mc R^{\perp_\Lambda}$ under $\phi$ as shown in Fig. \ref{fig:haahnull} i.e. it is a null distilled SS configuration. Once again, the image of a null configuration must map into $ \mc G_{\text{junk}}^{\perp_\Lambda}/ \mc R$ and indeed Fig. \ref{fig:haahnull} shows how this null operator can be decomposed into 33 edge operators. Once again, these operators represent terms that would appear in our effective Hamiltonian if the base only contained terms from $S_{\text{hex}}$. The second configuration we consider is the cube-corner which ``excites'' the pattern (has odd over lap with members of $\mc A_{\text{hex}}$) as shown in Fig. \ref{fig:haahorth}, which also shows other patterns of excitations generated by cube-corners. These cube-corner configurations can be used to form periodic patterns such that all excitations are removed, and we are left with a distilled SS logical operator. Note that all the excitations lie in a $[111]$ plane of cubes, so any configuration we consider for forming distilled SS logical operators using cube-corners may as well lie entirely in the same $[111]$ plane.  For example, one recognizes a pattern similar to Fig. \ref{fig:ccSSop} can be formed by these cub-corners. This as well as several others distilled SS operators have a finite periodicity, such as 1,2,3,4 and 6, and thus only represent a constant number of logical operators. These would appear to be connected to the constraints found in Ref. \cite{Schmitz2019a} which are periodic in the $[111]$ direction i.e. they are formed in the intersection of such constraints with a boundary along the $[111]$ direction, via Theorem 2 of the Reference. These constraints as well as the fractal constraints discussed in Ref. \ref{sec:haahdef} also map in $\mc R^{\perp_\Lambda}$ under $\phi$ by Proposition \ref{sec:genker}, though the members they map onto are complicated.  

We can also form fractal distilled SS logical operators from cube corners. In particular, the configuration on the left of Fig. \ref{fig:haahorth} is similar to the simplest excitation pattern of the Newman-Moore model, i.e. it forms a triangle, except for an additional excitation in the center. This can be used as the seed or first generation of a fractal sequence, whereby the next generation is given by four translated copies of the former generation as shown in Fig. \ref{fig:haahfracSS}. Excitations are spread along the corners of these triangular fractal patterns with one excitation remaining at the center. Because of the central excitation, these fractals do not themselves form distilled SS logical operators. However, fractal generation $a$ on the lattice of size $L=2^{a+1}$ forms a pattern of four excitation at the corners of a parallelogram half the size of the full $[111]$ plane. \footnote{Note that for the cube, a full $[111]$ plane wraps the cube at least twice, exactly twice when the plane intersects the cube at three corners for each wrap. In this case, the three corners form an ``upward-pointing'' triangle on one wrap and a ``downward-pointing'' triangle on the second wrap. Together these two triangles form a parallelogram which represents the full $[111]$ plane and is related to the triangular lattice with the periodicity of the underlying cube. } This implies that two copies of the fractal translated relative to each other by $2^a$ (half the period) in either of the primitive vector directions of the $[111]$ plane generates no excitations and thus represents a fractal distilled SS logical operator as depicted in Fig \ref{fig:haahfracSS}. These appear as rotated versions of one another. Each translated version of the resulting pattern is also a distilled SS logical operators, but only those along the primitive vector directions up to a period of $2^a$ are independent. Including both rotated versions, there are $4\times2^a -1= 2L-1$ independent patterns, where the minus one is due to the shared pattern between the rotated versions which consists of all cube-corners in the plane. Including the analogous Z-type patterns, we have $4L-2$ independent distilled SS logical operator which all mutually commute. From Ref. \cite{Haah2011, HaahThesis}, we know that for $L=2^{a+1}$, this is the maximum number of independent, mutually commuting logical operators.  To see that each plane is equivalent to every other, note there is a sum of these fractals such that the result is the sum of all cube-corners in that plane. As the sum of two such planes is equivalent to a member of $\mc D$ i.e. is the sum of cubes between the two planes, this implies each set of independent patterns in a given $[111]$ plane is equivalent to any other. This is a concrete example of a consequence of Theorem 2 of Ref. \cite{Schmitz2019a} which is that there is a complete mutually-commuting set of logical operators whose support is contained within a boundary for topological stabilizer codes. To relate these back to the fractal logical operators discussed in Section \ref{sec:haahdef}, one can see that internal to the fractal pattern of Fig. \ref{fig:haahfracSS}, there is a downward-facing triangle pattern formed by the cube-corners (this becomes more apparent the larger the fractal becomes). These can be deformed by a tetrahedral product of cubes out of the plane. Doing this at the various length scales forms a three-dimensional fractal pattern which is reminiscent of the logical operators discussed in Section \ref{sec:haahdef}. However, doing so does not give the exact fractal logical operators of Section \ref{sec:haahdef}. By the same arguments given in Section \ref{sec:distCC} for the cluster-cube fractal operators, the fractal operator similar to Fig. \ref{fig:fracbuild} should be accidental. So, our planar fractal logical operators are probably a sum of two non-planar fractal logical operators. We conjecture it is the sum of one of each geometry (one right-angle and one regular) based upon the counting of independent mutually commuting logical operators. That is, all non-planar fractal logical operators of both right-angle and regular geometry of either X or Z-type, but not both, suffice to generate a maximal, mutually commuting set, whereas with the planar fractal operators, we require both X and Z-types to find such a set.

Finally, we can find which completion results in a ground state for the perturbation analysis that is the simultaneous $+1$ eigenstate of these fractal distilled SS operators. Looking back to the p-strings generated by $S_{c' \hyph \text{hex}}$ as shown in Fig. \ref{fig:haahloops}, one can see that the p-strings generated by $t'^Z$ have the same overlap pattern with the cube-corners as $h^Z$, as shown in Fig. \ref{fig:haahcomp}. Thus our distilled SS logical operators belong to $\psi[\mc G_{c'\hyph\text{hex}}]^{\perp_\Omega}$ and by Proposition \ref{prop:loops}, $S_{c'\hyph\text{hex}}$ results in our desired ground state for the perturbation analysis. We also can find the members of $\psi[\mc G_{c'\hyph \text{hex}}]$ from Proposition \ref{prop:loops} generated by these logical operators. Fig. \ref{fig:haahdistloops} shows the plaquette configuration generated by the Z-type cube-corner.  We already see the loop configuration generated by $t'^Z$ plus some additional plaquettes. Thus, we can focus on only these remaining plaquettes. Fig. \ref{fig:haahdistloops} also shows the plaquette configuration generated by the three cube-corner operator used as the first generation for building the planar fractal distilled SS operators. One can see that mod p-strings generated by $t'^Z$, these plaquettes are confined to the $[111]$ plane, just like the configuration generated by $h'^Z$ (again, see Fig. \ref{fig:haahloops}). One can then show that once the distilled SS operator is formed, the resulting p-string configuration can be formed from  a sum of p-strings generated by $h'^Z$ operators.

\begin{figure}

\centering

\includegraphics[scale=.15]{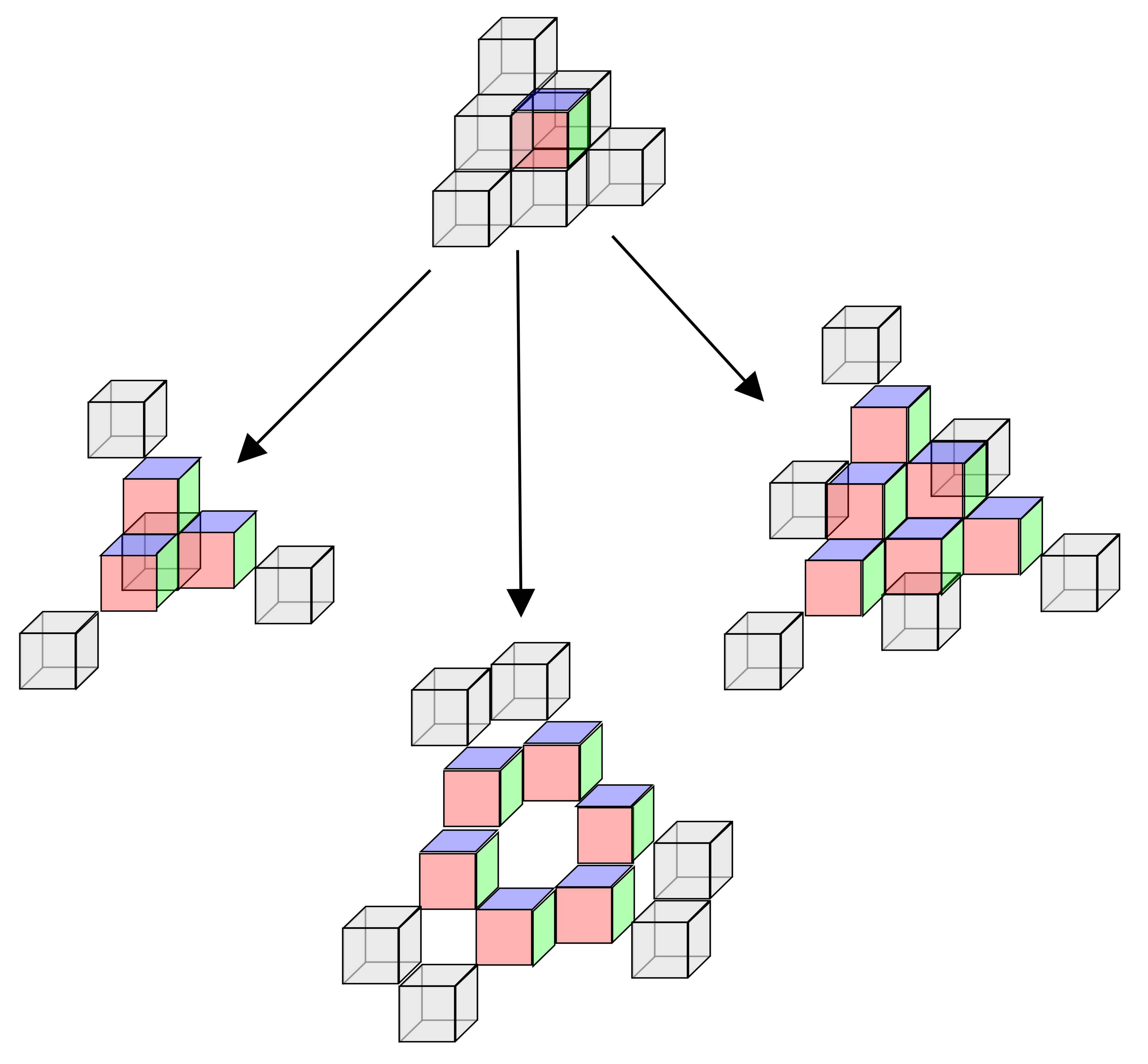}

\caption{Various patterns of excitations generated by the cube corners for the Haah's cubic code example.}\label{fig:haahorth}

\end{figure} 

\begin{figure}

\centering

\includegraphics[scale=.07]{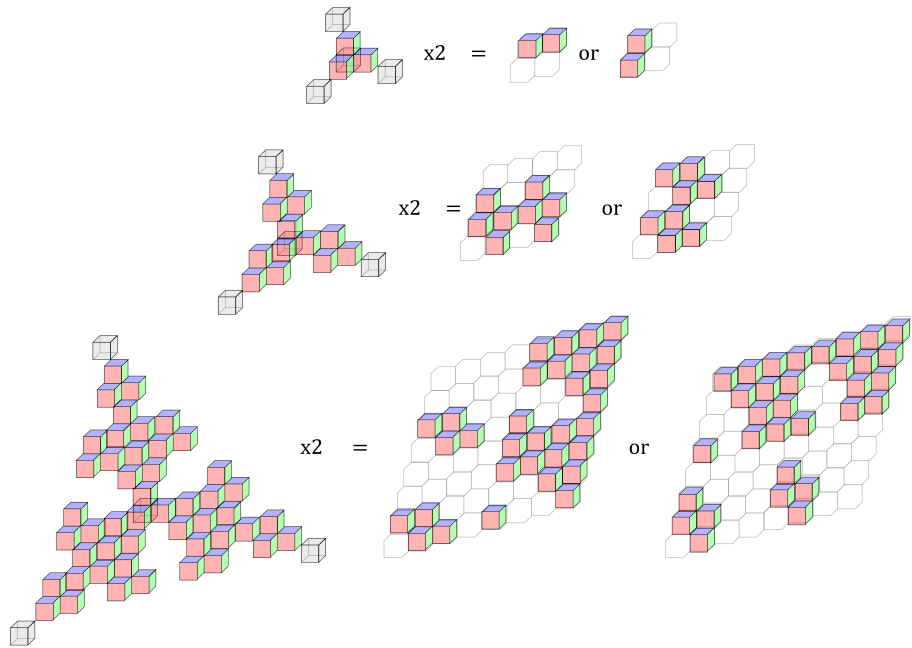}

\caption{Generations $a=1,2,3$ of the cube-corner fractal patterns and the excitations they create (gray) as well as how two copies of the pattern can be used to create distilled SS logical operators for $L=2^{(a+1)}$.}\label{fig:haahfracSS}

\end{figure} 

\begin{figure}

\centering

\includegraphics[scale=.15]{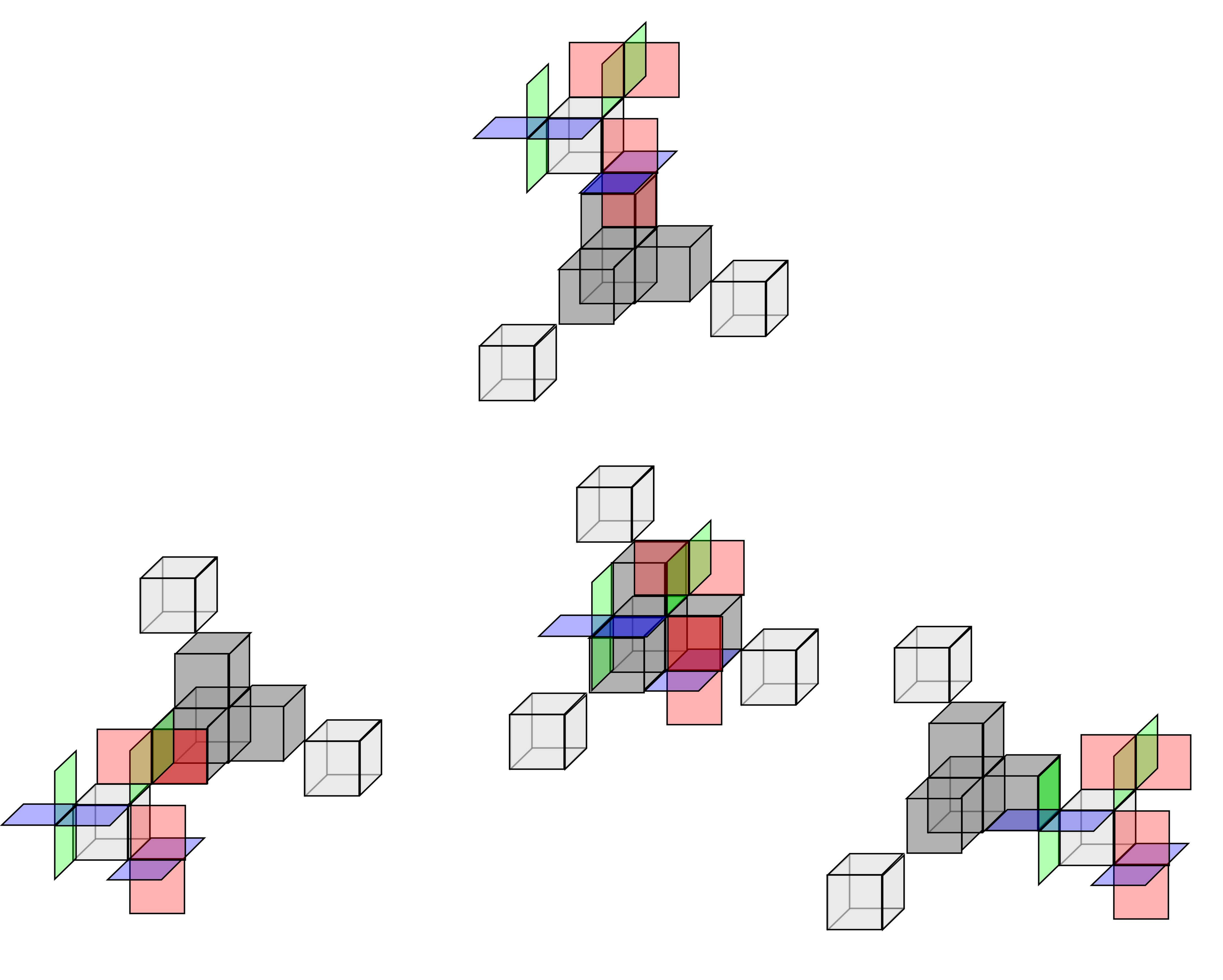}

\caption{Demonstration that the p-strings generated by the completion $S_{c'\hyph\text{hex}}$ have the same overlap pattern as the fractal configurations used to form the distilled SS logical operators of Fig. \ref{fig:haahfracSS}. Thus the perturbation analysis using $S_{c'\hyph\text{hex}}$ results in the simultaneous $+1$ grounds state of these logical operators.} \label{fig:haahcomp}

\end{figure} 

\begin{figure}

\centering

\includegraphics[scale=.15]{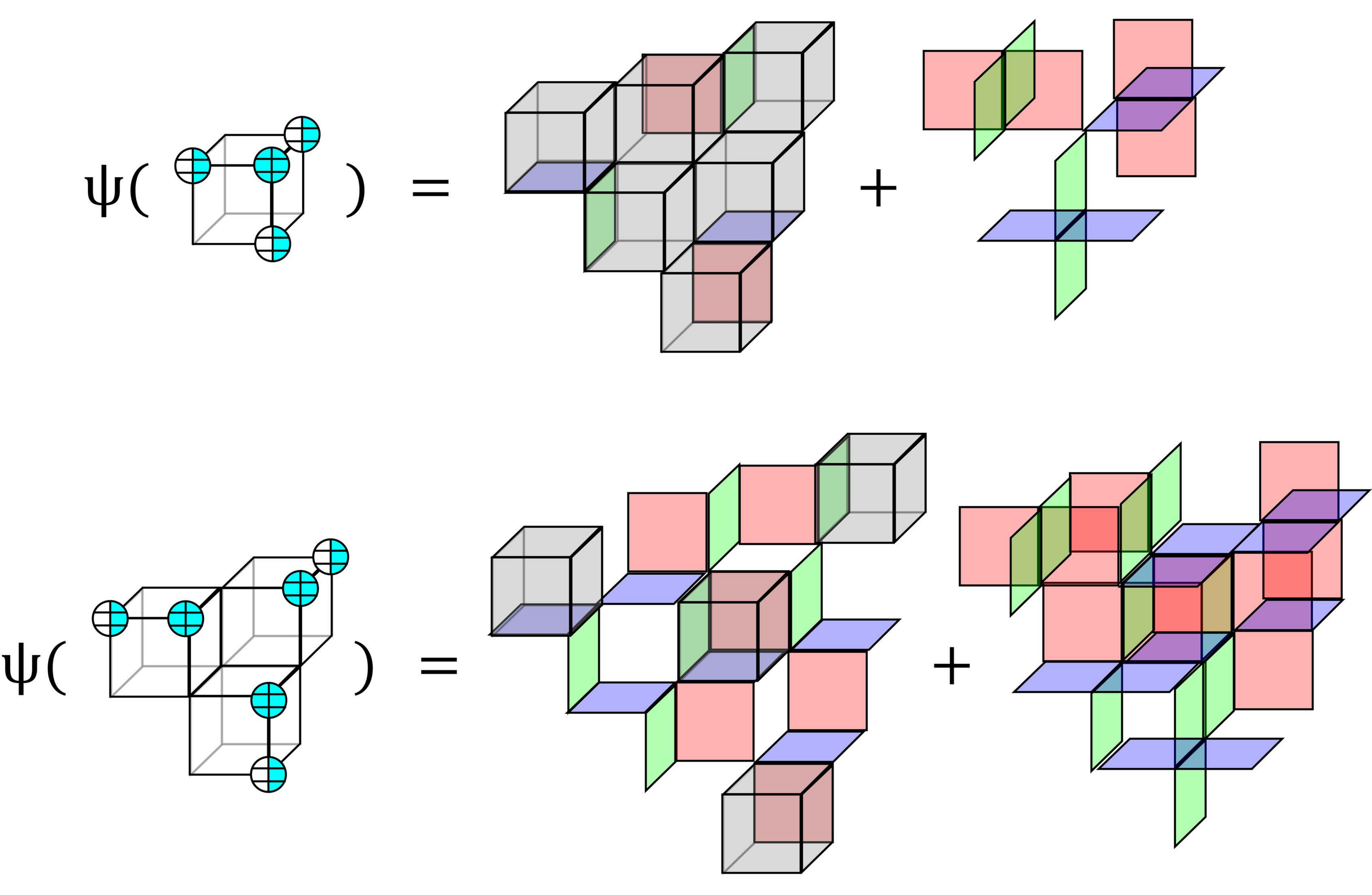}

\caption{Demonstrations of the X-type p-string configuration generated by the Z-type fractal distilled SS logical operator. After removing the p-string configuration generated by $t'^Z$, we see that the remaining plaquette configuration is confined to the $[111]$ plane and can be formed using p-strings generated by $h'^Z = h^Z + t'^Z$ operators, thus confirming Proposition \ref{prop:loops}. The gray cubes represent the Haah's code excitations generated by the operator. } \label{fig:haahdistloops}

\end{figure} 

\subsubsection{Simplifying the Base Model for the Perturbation Analysis of Haah's Cubic Code}\label{sec:haahsimple}

Every completion of $S_{\text{hex}}$ for Haah's code considered so far contains $h^X, h^Z$ in which case the base model Hamiltonian contains terms of weight 18 and supported in 6 unit cells. Thus, the base model is just as if not more complicated than the target. This is undesirable if one wants to use the layer construction to realize Haah's code in a realistic system. Moreover, no reasonable operators which are added to the completion simplify $h^X, h^Z$.
 
To remedy this, we relax the requirement that the base model is a stabilizer code. Instead we require the base model Hamiltonian be a sum over members of $\mc G_{\text{hex}}^{\perp_\Lambda}/ \mc R$ such that some product of the terms is equal to $h^X, h^Z$. In particular, we can consider an arbitrary sum over edge operators which are all weight 6 and supported on 2 unit cells.\footnote{It is worth noting that the set of edge operators along one line forms a version of the $d=1$ cluster model and edge operators in a given direction only anti-commute with edge operators along a different direction. This suggests one could obtain such a  base model from coupled $d=1$ cluster models.} For such a base model, $h^X, h^Z$ are symmetries and thus their eigenvalues are good quantum numbers which we can label with members of $\mc D^X \oplus \mc D^Z$. Moreover, since some product of the terms form $h^X, h^Z$, the energy eigenvalues must have a dispersion relation with respects to these quantum numbers. So the goal is to choose the coefficients of this new base Hamiltonian such the ground state of the $0 \in \mc D^X \oplus \mc D^Z$ sector is the unique overall ground state. This is reminiscent of models  whose ground space realizes a {\it sub-system} QECC \cite{Terhal2015}. This results in the same effective target Hamiltonian as that generated by a stabilizer base model, however the ground state of the target model which results from using such a base model is less obvious. We leave the study of this ground state and a more concrete discussion of a non-stabilizer base model for Haah's code to future work.

\section{Similarity to Entanglement Renormalization and Measurement-based Quantum Computing}\label{sec:ERGandMBQC}

The results discussed here should be reminiscent of entanglement renormalization (see Ref. \cite{Vidal2007} for a review). In particular, one can see the local right restriction as analogous to the process of applying a disentangler circuit so as to un-entangle local degrees of freedom and some fraction of them can be removed. Analogously for a local projective right restriction, there exists a local, constant-depth Clifford circuit which rotates the effective qubits of the restriction to actual degrees of freedom, but as these circuits are used to generate LRE, we refer to such a circuit as a {\it pre-entangler}. For example, we use CLR on three qubits for both our examples of SRE$\to$LRE. If one applies the circuit, $\left(\text{CNOT}_{12}, \text{CNOT}_{23}, \text{CNOT}_{31}\right)$, qubit 3 acts as our effective degree of freedom. So to achieve the SRE$\to$LRE layer construction on a hypothetical quantum computer:
\begin{enumerate}
\item Prepare the SRE stabilizer state for each layer. By definition, this can be done with a constant depth, local circuit.
\item Apply the pre-entangler circuit to each local unit cell, which collectively is again a constant depth, local circuit.
\item Measure the stabilizers of the base model (or in general, the terms of the base Hamiltonian). From this, one can infer an error syndrome for the target, including that of the distilled SS operators.
\item Decode and correct errors in the target code based upon the syndrome of the last step.
\item Discard the qubits which do not contain the degrees of freedom of the right restriction.   
\end{enumerate}  

This process results in a QECC corresponding to the LRE model as prepared in a fiducial state as determined by the base model and Proposition \ref{prop:loops}. In the case that the base model is simple as is the case for $S_{c\hyph \triangle}$ for the cluster-cube model, such a procedure is practically more efficient than measuring the target stabilizers directly. Furthermore, if we allow for an analogous relaxation of the base model as in Section \ref{sec:haahsimple} away from a stabilizer model and allow for more general measurements between members of $\mc G_{\text{junk}}^{\perp_\Lambda}/ \mc R$, it is reasonable to suspect that we can generate any state in the code space of the LRE model (except those associated with accidental logical operators). Devising such a scheme is equivalent to a fault-tolerant, measurement-based quantum computer (assuming the results can tolerate faults in the hypothetical measurement scheme) or could serve as a quantum memory write procedure. Such a conjecture is also supported by the fact that the $d=2$ cluster state can be used to realize a universal measurement-based quantum computer\cite{Raussendorf2003}. We look to explore this possibility more concretely, as well as the use of gauge substructures for entanglement renormalization, in future work.

\section {Conclusion}\label{sec:con}

In this paper, we have discussed how layers of 2D SSPT models can be coupled to form 3D fracton models in a process we have deemed distillation of long-range entanglement from short-range entanglement (SRE $\to$ LRE) in contrast to condensation in anyonic models. We started by analyzing the process for two examples, the cluster cube model and Haah's cubic code, from the perspective of Brillouin-Wigner perturbation theory as well as reviewed the previous of results of Refs. \cite{Ma2017, Vijay2017} that the X-cube model can be obtained by condensing anyonic p-strings in layer of the $d=2$ toric code. We then explored the results in great detail using the language of linear gauge structures as well as introduced an extension of the idea, which we have deemed a gauge substructure. The use of gauge substructures allowed us to rigorously distinguish between distillation and condensation as well as constructively prove some important general results which practically allowed us to answer the two primary questions raised by SRE$\to$LRE: what is the source of the LRE and what determines the resulting ground state after the infinite-strength perturbation? These abstract results were then applied to better understand the cluster cube and Haah's code results. 

We have already discussed how these results might be used in future work. In particular, this might be used to better understand entanglement renormalization, and could lead to a protocol for realizing a fault-tolerant, measurement-based quantum computer. Furthermore, the methods developed are highly constructive. For example, the author uniquely derived the quasi-cluster model assuming only that the spatial symmetries of Haah's code are also the symmetries of the layers which form the parent. Likewise, the cluster-cube model was only found after applying the right and left restrictions to layers of the cluster model. Thus, one only needs to provide either the parent or the target, and the other tends to follow. This could lead to many new exotic models and layer constructions for existing models with relative ease. Then the remaining pieces follow directly for the abstract analysis provided by the gauge structure and substructure formalism.  More broadly, we see the development of the gauge substructures as perhaps the most important use for the linear gauge structure formalism as we look to further broaden the use and scope of these ideas.

\section{Acknowledgments}

The author would especially like to thank Rahul Nandkishore for his support. Also the author thanks Abhinav Prem and Michael Pretko for comments on the manuscript as well as Trithep Devakul, Jeongwan Haah, Michael Hermele, Sheng-Jie Huang, Han Ma, Kevin Slagle, and Dominic Williamson for feedback on the methods presented. This material is based upon work supported by the Air Force Office of Scientific Research under award number FA9550-17-1-0183. 

\appendix

\section{Review of Brillouin-Wigner Perturbation Theory}\label{apx:PT}

The following section builds on Appendix A of Ref. \cite{Ma2017} as well as Ref. \cite{Sakurai2011}. In this appendix we justify that in general, the only terms of the effective Hamiltonian which survive the layer construction of Section \ref{sec:PT} are those products of parent stabilizers, which commute with the base model which by construction is given by the left restriction $\mc L$. We also show that the remaining degrees of freedom are generally reduced to those of the right restriction, $\mc R$.  

We start by factoring out $K$ from Eq. \ref{eq:Ham1} such that our Hamiltonian becomes
\begin{align}
H= H_{\text{base}} + K' H_{\text{parent}},
\end{align}
where $K'= \frac{1}{K}$ and we can now treat the parent model as an Infinitesimal perturbation to the base model. Let $\ket{\Psi}$ be any low energy state of $H$ with energy $E$. In particular, $\ket{\Psi}$ is a state with high overlap with the base model ground space.  We can then write the Shr\"{o}dinger's equation in the form of Brillouin-Wigner perturbation theory as 
\begin{align}\label{app:schrod}
\ket{\Psi} =  p^{(0)}_{\text{base}} \ket{\Psi} + H_{\text{base}}^{-1} \left(1- p^{(0)}_{\text{base}}\right) \left(E- K' H_{\text{parent}}\right) \ket{\Psi},
\end{align}
where we have used the fact that the base ground space energy is $0$. To solve this iteratively, we need a seed state in the base model ground space, so we use the ansatz state $p^{(0)}_{\text{base}} \ket{\psi_{\text{parent}}(f)}$ for some arbitrary operator $f \in \mc P$. Our solution based on this ansatz is 
\begin{align}
\ket{\Psi}= \sum_n \left[H_{\text{base}}^{-1} \left(1- p^{(0)}_{\text{base}}\right) \left(E- K' H_{\text{parent}}\right)\right]^n p^{(0)}_{\text{base}} \ket{\psi_{\text{parent}}(f)}.
\end{align}
We wish to express this in the form of an effective Schr\"{o}dinger's equation for some effective Hamiltonian. To do so we multiply by $p^{(0)}_{\text{base}} H$ on both sides and note that $p^{(0)}_{\text{base}} H \ket{\Psi}= E p^{(0)}_{\text{base}}\ket{\psi_{\text{parent}}(f)}$. As a result, we find
\begin{align}
 H_{\text{eff}}\left(p^{(0)}_{\text{base}}\ket{\psi_{\text{parent}}(f)}\right) = E \left(p^{(0)}_{\text{base}}\ket{\psi_{\text{parent}}(f)}\right),
\end{align}
where
\begin{align}
 H_{\text{eff}}=K' p^{(0)}_{\text{base}} H_{\text{parent}}\sum_n \left[H_{\text{base}}^{-1} \left(1- p^{(0)}_{\text{base}}\right) \left(E- K' H_{\text{parent}}\right)\right]^n p^{(0)}_{\text{base}}.
\end{align}
This effective Schr\"{o}dinger's equation is a self-consistent equation as both sides depend on the energy $E$. Furthermore because $H_{\text{eff}}$'s dependence on $E$, the orders of the sum do not directly correspond to the orders in perturbation theory. As usual in perturbation theory, one must recursively solve the equation order-by-order. However, we still want to confirm that only products of parent stabilizers which commute with all base stabilizers are the only terms which survive. Because $H_{\text{eff}}$ is projected on to the ground space of the base model, we can expand the operator sandwiched between $p^{(0)}_{\text{base}}$ in the Pauli operator basis. We can expand 
\begin{align}
H_{\text{base}}^{-1} \left(1- p^{(0)}_{\text{base}}\right)= \sum_{J\neq0 \in \im \psi_{\text{base}}} \frac{1}{E^{(J)}_{\text{base}}} p^{(J)}_{\text{base}},
\end{align}
and clearly $H_{\text{parent}}$ is expanded in $S_{\text{parent}}$ operators. So at order $n$ in the perturbation theory, one has terms of the form
\begin{align}
(-1) p^{(0)}_{\text{base}} s_{i_{n-1}} p^{(J_{n-1})}_{\text{base}} \dots s_{i_1} p^{(J_1)}_{\text{base}} s_{i_0}p^{(0)}_{\text{base}},
\end{align}
 where $s_i \in S_{\text{parent}}$, as well as other terms. We can push $s_i$ through the projection operator such that $p^{(J)}_{\text{base}} s_{i}=s_i p^{(J +\psi_{\text{base}}(s_i))}_{\text{base}}$. Such a term only survive once it hits the ground space projector if $J_1 = \psi(s_{i_0})$, $J_2= \psi(s_{i_0}+ s_{i_1})$, and so on, which determines the energy denominators for this term. More importantly for us, however, is the fact that once we push all the stabilizers to the left, the term only survives if $\sum_{j=0}^{n-1}\psi_{\text{base}}(s_{i_j}) =0$, i.e. the resulting product of parent stabilizers commutes with all stabilizers of the base model. The other terms are dependent on the energy and have fewer than $n$ stabilizer products at order $n$, but they result in the same condition for the term to be non-zero. So after a shift in energy, the effective Hamiltonian can be written as 
\begin{align}
H_{\text{eff}} =\frac{1}{2} \sum_{A \in \ker \psi_{\text{base}} \phi_{\text{parent}}} j_A(1- \phi(A)),
\end{align}
 where $j_A$ are coefficients which can be derived in principle by evaluating $\tr(\phi(A)H_{\text{eff}})$. Without doing this explicitly, we do know that $j_A = \mc O(K'^{|A|})$ and at leading order--i.e. for small enough $K'$--$j_A>0$. So our ansatz that $  \left(p^{(0)}_{\text{base}}\ket{\psi_{\text{parent}}(f)}\right)$ are the correct eigenstates of our effective Hamiltonian is correct provided that $\psi_{\text{parent}}(f) \in \ker \psi_{\text{base}} \phi_{\text{parent}}= (\im \psi_{\text{parent}}\phi_{\text{base}})^{\perp_\omega}$, using the composite gauge structure BrLE rules as discussed in Lemma \ref{perpperp}. By construction, $(\im \psi_{\text{parent}}\phi_{\text{base}})^{\perp_\omega} \supseteq \mc L$.

As a final point, we want to show which members of $\im\psi_{\text{parent}}$ survive the projection by $ p^{(0)}_{\text{base}}$. To do this, we first define a new gauge structure $\left(\left(\mc A_{\text{base}} \oplus \mc A_{\text{target}}, \Phi, \Omega'' \right), \left( \mc P, \Psi, \Lambda\right)\right)$, such that $\Phi(A \oplus B)= \phi_{\text{base}}(A) + \phi_{\text{target}}(B)$ and $\Omega''= \Omega_{\text{base}} \oplus \Omega_{\text{target}}$. We can uniquely derive $\Psi$ by noting 
\begin{align}
\lambda(\Phi(A \oplus B), f)=& \lambda(\phi_{\text{base}}(A), f) + \lambda(\phi_{\text{target}}(B), f) \nonumber \\
=& \omega_{\text{base}}(A, \psi_{\text{base}}(f)) + \omega_{\text{target}}(B, \psi_{\text{target}}(f))\nonumber \\
=& \omega''(A \oplus B, \psi_{\text{base}}(f) \oplus \psi_{\text{target}}(f)).
\end{align}
Thus $\Psi(A \oplus B)= \psi_{\text{base}}(f) \oplus \psi_{\text{target}}(f)$. Note that for all $A \oplus B \in \ker \Phi$, $\phi_{\text{base}}(A) = \phi_{\text{target}}(B)$. With this gauge structure in mind, let $J \in \mc A_{\text{parent}}$ and consider
\begin{align}
\tr\left( p^{(0)}_{\text{base}} p_{\text{target}}^{(J)}\right)\propto &  \sum_{A\in \mc A_{\text{base}}} \sum_{B \in \mc A_{\text{target}}} (-1)^{\omega_{\text{target}}(B, J)} \tr \left(\phi_{\text{base}}(A)\phi_{\text{target}}(B)\right) \nonumber \\
\propto &  \sum_{A\oplus B\in \ker \Phi} (-1)^{\omega''(A \oplus B, 0 \oplus J)} \propto \left[ 0 \oplus J \in (\ker \Phi)^{\perp_{\Omega''}}\right],
\end{align}
where $[]$ is the Iverson bracket which is $1$ if the statement is true and $0$ otherwise. Thus, the only $J= \psi_{\text{parent}}(f)$ whose ket survives the projection by $p_{\text{base}}^{(0)}$ is such that $ 0 \oplus \psi_{\text{target}}(f) \in (\ker \Phi)^{\perp_{\Omega''}}= \im \Psi$, where we invoke the BrLE rules for our new gauge structure. This implies $\psi_{\text{base}}(f) =0$, which  implies $f \in (\im \phi_{\text{base}})^{\perp_\Lambda} = \mc R \oplus (\im \phi_{\text{base}})^{\perp_\Lambda} / \mc R$. The first equality is a consequence of the BrLE rules for the base model and the second is by construction of the gauge substructure. When the completion used for the base model is maximally mutually commuting in $\mc R^{\perp_\Lambda}$, i.e. $(\im \phi_{\text{base}})^{\perp_\Lambda} / \mc R= \im \phi_{\text{base}}$, then our effective degrees of freedom are $(\mc R \oplus \im \phi_{\text{base}}) /\im \phi_{\text{base}} \simeq \mc R$ as the projection by $p_{\text{base}}^{(0)}$ freezes out $\im \phi_{\text{base}}$ degrees of freedom. Therefore, we have reduced all degrees of freedom down to $\mc R$ as intended.

\section{Peripheral Details of $\mb F_2$ Vector Spaces}\label{apx:F2vec}

In this appendix, we establish a few basic points regarding $\mb F_2$ vector spaces which were used throughout the main body. 

Ref. \cite{Bernhard2009} discusses general frame theory of binary vector spaces equipped with a non-degenerate 2-form. By proofs provided in the reference, every invertible form on a binary space satisfies the same properties so all results apply equally to the symplectic Pauli space $(\mc P, \lambda)$, as it does for a binary dot-product space. Let $(\mc A, \omega)$ be an $N$-dimensional $\mb F_2$ vector space and $\omega$ be the non-degenerate binary 2-form. Of interest to us is Section 3.1 of the Reference, which establishes the existence of a linear, idempotent project map on to a space $\mc L\subseteq \mc A$. Such a map exists if and only if there is a {\it dual basis pair}, $(\{b_i\}, \{\tilde{b}_i\})$ for $\mc L$ such that for all $L \in \mc L$, 
\begin{align}
L = \sum_i \omega(b_i,L) \tilde{b}_i.
\end{align}
One can also define the {\it Grammian}, $G$, for a dual basis pair as the binary matrix for which $G_{ij}= \omega(b_i, \tilde{b}_j)$. They also establish that a dual basis pair exists if and only if there exists a Grammian (i.e. any matrix formed from two bases) which is invertible, and a matrix is invertible if and only $G* L =0$ implies $L=0$. So for any Grammian $G$, suppose there exists an $L\in \{0,1\}^{\dim \mc L}$ such that $G*L=0$. This implies 
\begin{align}
0= \sum_j \omega(b_i, \tilde{b}_j) L_j = \omega(b_i, \sum_j L_j \tilde{b}_j).
\end{align}
As $\{b_i\}$ is a basis for $L$, this implies $\tilde{L}= \sum_j L_j \tilde{b}_j \in \mc L$ is a vector in our subspace which is orthogonal to all of $\mc L$, i.e. $\tilde{L} \in \mc L \cap \mc L^{\perp_\omega}$. Therefore, this establishes the following:

\begin{proposition}\label{app:proj}
$\mc L \subseteq \mc A$ is a projective subspace, i.e. there exist a linear, idempotent map $\pi_{\mc L}$ if and only if $\mc L \cap \mc L^{\perp_\omega}$ is trivial.
\end{proposition}
As a corollary, if $\mc L$ is projective, then so is $\mc L^{\perp_\omega}$ and there exist of decomposition of the space $\mc A \simeq \mc L \oplus \mc L^{\perp_{\omega}}$ which also immediately implies that $\left(\mc L^{\perp_\omega}\right)^{\perp_\omega}= \mc L$ when we also use the non-degeneracy of $\omega$.

We also made the claim that of $(\mc P, \lambda)$, and any mutually commuting space $\mc G$, $\left(\mc G^{\perp_\lambda}\right)^{\perp_\lambda}= \mc G$. Such a subspace maximally violates the conditions for projectivity, so we cannot use this to prove the claim. However, we can use gauge structures. $\mc G$ always has a basis and we can use this to generate a gauge structure where the potential space is a binary vector space of size $\dim G$ and $\im \phi =\mc G$. Because binary dot product is invertible, such a gauge structure is transposable as discussed in Lemma \ref{perpperp}. One then uses the same logic as that of Lemma \ref{perpperp} to establish the assertion. Other methods for proving this can be found in Ref. \cite{HaahThesis}.

Finally, we discuss some details regarding the space $\mc D$ or the space of all closed 2D membranes. For a lattice of size $L$, $\dim \mc D = L^3-1$, where as we typically describe the resulting space as that of $\mc V$, the set of points of the 3D lattice or alternatively the space of all cubes of the 3D lattice for which $\mc D \simeq \text{div} [\mc V]$. But $ \dim \mc V= L^3$. This is because the divergence of all cubes is the identity, resulting in one constraint. So it would see that when using $\mc D$ for SRE$\to$LRE as our left restriction, we have inserted LRE by hand as one constraint implies one logical qubit (when $\dim \mc D= N$). This is exactly the case in our examples as we technically violate the maximality condition stated in Def. \ref{def:max}. For each example, one can find a topologically non-trivial membrane $A \in \mc A_{\text{layers}}$ wrapping the torus and thus $A \notin \mc D$ and yet $\phi(A)\in \mc R$. As discussed there, this always results in a logical operator. However, this appears to be the only violation of the maximality condition, so if we imagine modifying $ \mc D \to \mc D \oplus\{0, A\}$, the restriction is maximal and all results follow without changing the primary results, i.e. we distill $ \sim L$ logical operators from subsystem symmetries. Beyond our examples, we can also argue that using $\mc D$ always results in a distilled SS operator. As discussed, $\psi \pi_{\mc R}\phi[\mc D] \subseteq \mc D^{\perp_{\Omega}}$, i.e. the members of $\mc D$ alway generate p-string under $\psi \pi_{\mc R}\phi$, which a priori does not exclude the topologically non-trivial p-strings wrapping the torus.  But in fact, such topologically non-trivial p-strings cannot be in $\psi \pi_{\mc R}\phi[\mc D]$ as this set is locally generated, so all p-strings must be  topologically trivial. As topologically non-trivial p-strings are the only p-string that have odd overlap with topologically non-trivial membranes, such membranes must be in $(\psi \pi_{\mc R}\phi[\mc D])^{\perp_\Omega}$ and by Proposition \ref{prop:LREsource} become a distilled SS operators. In our example of condensation, this is avoided because such membranes are equivalent to a constraint in the $\text{TC}_2$ layers, i.e. planes are null distilled configurations. $\mc D$ also fails to be projective when $L$ is even because one can form the product of every wrapping membrane perpendicular to a given coordinate direction in $\mc D$ as this is an even number of such layers. One can also form such a configuration in $\mc D^{\perp_\Omega}$ from all p-strings wrapped in the same direction. This violates the condition of Proposition \ref{app:proj} for a projective $\mc L$. In general even when not projective, we use ${}_{\mc D}\psi_{\mc R}= \beta \iota^{\dagger}_{\mc D} \psi \iota_{\mc R}$, where $\beta: \mc A/\mc D^{\perp_\Omega} \to \mc D$ is such that $[A] \mapsto \sum_{c} \omega(A, B_c) B_c$ and $B_c \in \mc D$ are the primitive cube configurations.

\printbibliography

\end{document}